\documentclass[11pt]{article}
\usepackage{geometry}
\usepackage{tikz,graphicx}
\usepackage{amssymb,amsfonts,amsmath,amsthm,url}
\newtheorem{theorem}{Theorem}
\newtheorem{lemma}{Lemma}
\newtheorem{corollary}{Corollary}
\newtheorem{proposition}{Proposition}
\newtheorem{definition}{Definition}

\usepackage{mathrsfs}
\usepackage[mathscr]{euscript}
\usepackage{amsfonts,amsmath}
\usepackage{amstext}
\usepackage{xspace}
\usepackage[]{algorithm2e}

\usepackage{tikz}
\usepackage{bbm}
\usepackage{lineno}
\usetikzlibrary{snakes}
\usetikzlibrary{arrows.meta}
\usepackage{subfigure}
\usepackage{latexsym}
\usepackage{mathrsfs}
\usepackage{amssymb,amsmath}
\usepackage{color}
\usepackage[english]{babel}
\usepackage{graphicx}
\usepackage{pifont,stmaryrd,float}
\usetikzlibrary{calc,shapes}
\usetikzlibrary{arrows,automata,backgrounds,calendar}
\usetikzlibrary{positioning,arrows,positioning,patterns,decorations.pathreplacing}

\newtheorem{observation}{Observation}
\newtheorem{redrule}{Reduction Rule}

\date{}

\newcommand{\OO}{\mathcal{O}}

\newcommand{\yes}{\textsf{Yes}}
\newcommand{\no}{\textsf{No}}

\newcommand{\tw}{\mathtt{tw}}
\newcommand{\NP}{\textsf{NP}}

\newcommand{\FPT}{\textsf{FPT}}
\newcommand{\XP}{\textsf{XP}}
\newcommand{\WO}{\textsf{W[1]}}
\newcommand{\WOH}{\textsf{W[1]}-hard}

\newcommand{\NPH}{\textsf{NP}-hard}

\newcommand{\paraH}{\textsf{para-NP}-hard}
\newcommand{\ETH}{\textsf{ETH}}

\renewcommand{\phi}{\varphi}

\newcommand{\multisetPackingR}{\textsc{$p$-Multiset $(r,q)$-Packing}}
\newcommand{\setPackingR}{\textsc{$p$-Set $(r,q)$-Packing}}
\newcommand{\monomDetR}{\textsc{$(r,k)$-Monomial Detection}}
\newcommand{\undiPathR}{\textsc{Undirected $r$-Simple $k$-Path}}
\newcommand{\undiPathRSpecial}{\textsc{Special Undirected $r$-Simple $k$-Path}}
\newcommand{\diPathR}{\textsc{Directed $r$-Simple $k$-Path}}
\newcommand{\diPathRst}{\textsc{Directed $r$-Simple Long $(s,t)$-Path}}
\newcommand{\diColPathRst}{\textsc{Directed Colorful $r$-Simple Long $(s,t)$-Path}}
\newcommand{\topologyProb}{\textsc{$(\ell,r)$-Enriched Topology}}

\newcommand{\topologyProbRoot}{\textsc{Rooted $(\ell,r)$-Enriched Topology}}
\newcommand{\colPathRst}{\textsc{Undirected Colorful $r$-Simple $k$-Path}}
\newcommand{\colPathRspecial}{\textsc{Special Undirected Colorful $r$-Simple $k$-Path}}
\newcommand{\eulerPart}{\textsc{(Walk,TW-2) Partition}}
\newcommand{\eulerEdge}{\textsc{(Walk,Edges) Partition}}

\begin{document}

\title{$r$-Simple $k$-Path and Related Problems Parameterized by $k/r$\thanks{A preliminary version of this paper appeared in the proceedings of 30th ACM-SIAM Symposium on Discrete Algorithms (SODA 2019). Research of Gutin was partially supported by Royal Society Wolfson Research Merit Award and Leverhulme Trust grant no.~RPG-2018-161. Research of Zehavi was partially supported by Israel Science Foundation (ISF) grant no.~1176/18.}}

\author{Gregory Gutin$^{1}${ }
Magnus Wahlstr{\"o}m$^{1}${ } Meirav Zehavi$^{2}$
\\$^{1}$ Royal Holloway, University of London, Egham, United Kingdom\\ g.gutin@rhul.ac.uk,  magnus.wahlstrom@rhul.ac.uk\\
\\$^{2}$ Ben-Gurion University of the Negev\\
meiravze@bgu.ac.il}

\maketitle 

\begin{abstract}
Abasi et al.~(2014) introduced the following two problems. In the {\sc $r$-Simple $k$-Path} problem, given a digraph~$G$ on $n$~vertices and positive integers $r,k$, decide  whether $G$ has an $r$-simple $k$-path, which is a walk where every vertex occurs at most $r$ times and the total number of vertex occurrences is $k$. In the {\sc $(r,k)$-Monomial Detection} problem, given an arithmetic circuit that succinctly encodes some polynomial $P$ on $n$ variables and positive integers $k,r$, decide whether $P$ has a monomial of total degree $k$ where the degree of each variable is at most~$r$. Abasi et al.~obtained randomized algorithms of running time $4^{(k/r)\log r}\cdot n^{\OO(1)}$ for both problems. Gabizon et al.~(2015) designed  deterministic  $2^{\OO((k/r)\log r)}\cdot n^{\OO(1)}$-time algorithms for both problems (however, for the {\sc $(r,k)$-Monomial Detection} problem the input circuit is restricted to be non-canceling). Gabizon et al.~also  studied the following problem. In the {\sc $p$-Set $(r,q)$-Packing} problem, given a universe $V$, positive integers $p,q,r$, and a collection $\cal H$ of sets of size $p$ whose elements belong to $V$, decide whether there exists a subcollection ${\cal H}'$ of $\cal H$ of size $q$ where each element occurs in at most $r$ sets of ${\cal H}'$. Gabizon et al.~obtained a deterministic $2^{O((pq/r)\log r)}\cdot n^{\OO(1)}$-time algorithm for {\sc $p$-Set $(r,q)$-Packing}.

The above results  prove that the three problems are {\em single-exponentially} fixed-parameter tractable (\FPT) parameterized by the product of 
{\em two} parameters, that is, $k/r$ and $\log r$, where $k=pq$ for {\sc $p$-Set $(r,q)$-Packing}. Abasi et al.~and Gabizon et al.~asked whether the $\log r$ factor in the exponent can be avoided. Bonamy et al.~(2017) answered the question for {\sc $(r,k)$-Monomial Detection} by proving that  unless the Exponential Time Hypothesis (\ETH) fails there is no $2^{o((k/r)\log r)}\cdot (n+\log k)^{\OO(1)}$-time algorithm for {\sc $(r,k)$-Monomial Detection}, i.e., {\sc $(r,k)$-Monomial Detection} is unlikely to be single-exponentially \FPT\ when parameterized by $k/r$ alone. The question remains open for {\sc $r$-Simple $k$-Path} and {\sc $p$-Set $(r,q)$-Packing}.

We consider the question from a wider perspective: are the above problems \FPT\ when parameterized by $k/r$ only, i.e., whether there exists a computable function $f$ such that the problems admit a $f(k/r)(n+\log k)^{\OO(1)}$-time algorithm? Since $r$ can be substantially larger than the input size, the algorithms of Abasi et al.~and Gabizon et al. do not even show that any of these three problems is in \XP\ parameterized by $k/r$ alone.
We resolve the wider question by  (a) obtaining a $2^{\OO((k/r)^2\log(k/r))}\cdot (n+\log k)^{\OO(1)}$-time algorithm for {\sc $r$-Simple $k$-Path} on digraphs and a $2^{\OO(k/r)}\cdot (n+\log k)^{\OO(1)}$-time algorithm for {\sc $r$-Simple $k$-Path} on undirected graphs (i.e., for undirected graphs we answer the original question in affirmative), (b) showing that {\sc $p$-Set $(r,q)$-Packing} is \FPT\ (in contrast, we prove that {\sc $p$-Multiset $(r,q)$-Packing} is \WO-hard), and (c) proving that {\sc $(r,k)$-Monomial Detection} is \paraH\ even if only two distinct variables are in polynomial $P$ and the circuit is non-canceling. For the special case of {\sc $(r,k)$-Monomial Detection} where $k$ is polynomially bounded by the input size (which is in \XP), we show \WO-hardness. Along the way to solve {\sc $p$-Set $(r,q)$-Packing}, we obtain a polynomial kernel for any fixed $p$, which resolves a question posed by Gabizon et al.~regarding the existence of polynomial kernels for problems with relaxed disjointness constraints. All our algorithms are deterministic.
\end{abstract}


\section{Introduction}\label{sec:intro}

Abasi et al.~\cite{DBLP:conf/mfcs/AbasiBGH14} introduced the following extension of the {\sc Directed $k$-Path} problem: 

\begin{defproblem}
{\diPathR}
{An $n$-vertex digraph $G$ and positive integers $k,r.$}
{Does $G$ have an $r$-simple $k$-path, that is, a walk where every vertex occurs at most $r$ times and the total number of vertex occurrences is $k$?}
\end{defproblem}

Note that in \diPathR\ $k,r$ can be substantially larger than $n$.\vspace{1mm}

At first glance, one may think that the time complexity of any algorithm for solving \diPathR\ is an increasing function in $r$. However, Abasi et al.~showed that this is not the case by designing a randomized algorithm of running time $4^{(k/r)\log r}\cdot n^{\OO(1)}$. Their algorithm was obtained by a simple reduction to the following problem:

\begin{defproblem}
{\monomDetR}
{An arithmetic circuit that succinctly encodes some $n$-variable polynomial $P$, and positive integers $k,r$.}
{Does $P$ have a monomial of total degree $k,$ where the degree of each variable is at most~$r$?}
\end{defproblem}\\
Abasi et al.~proved that  \monomDetR\ can be solved by a randomized algorithm with time complexity $4^{(k/r)\log r}\cdot n^{\OO(1)}.$ Gabizon et al.~\cite{DBLP:conf/esa/GabizonLP15} derandomized these two randomized algorithms, though at the expense of increasing the constant factor in the exponent and restricting the input of the  \monomDetR\ problem to non-canceling circuits.\footnote{Non-defined terms can be found in the next section.} Both algorithms of Gabizon et al.~run in time $2^{\OO((k/r)\log r)}\cdot n^{\OO(1)}.$

Gabizon et al.~\cite{DBLP:conf/esa/GabizonLP15} also studied the following problem:

\begin{defproblem}
{\setPackingR}
{An $n$-element universe $V$, positive integers $p,q,r$, and a collection $\cal H$ of sets of size $p$ whose elements belong to $V$.}
{Does there exist a subcollection ${\cal H}'$ of $\cal H$ of size $q$ where each element occurs in at most $r$ sets of ${\cal H}'$? (We will call ${\cal H}'$ an $r$-{\em relaxed parking}.)}
\end{defproblem}\\
Gabizon et al. designed an algorithm for \setPackingR\ of running time $2^{\OO((k/r)\log r)}\cdot n^{\OO(1)},$ where $k=pq.$ In other words, the above results show that the three problems are {\em single-exponentially} fixed-parameter tractable (\FPT) when parameterized by the product of {\em two} parameters,  $k/r$ and $\log r$.

The motivation behind the relaxation of disjointness constraints is to enable finding {\em substantially better (larger)} solutions at the expense of allowing elements to be used multiple (but bounded by $r$) times. For example, for any choice of $k,r$, Abasi et al.~\cite{DBLP:conf/mfcs/AbasiBGH14} presented digraphs that have at least one $r$-simple $k$-path but do not have even a single (simple) path on $4\log_rk$ vertices. Thus, even if we allow each vertex to be visited at most twice rather than once, already we can gain an {\em exponential} increase in the size of the output solution. The same result holds also for undirected graphs.\footnote{\undiPathR\ can be viewed as the special case of \diPathR\ where every pair of vertices has either no arc or arcs in both directions.} In addition, Abasi et al.~\cite{DBLP:conf/mfcs/AbasiBGH14} showed that the relaxation does not make the problem easy: both \undiPathR\ and \diPathR\ are shown to be \NP-hard with $k=(2r-1)p+2$ and $n=2p$ vertices. From this, we observe that \NP-hardness holds for a wide variety of choices of $r$, ranging for $r$ being any fixed constant to $r$ being super-exponential in $n$ (e.g., $r=2^{n^{c}}$ for any fixed constant $c\geq 1$). In addition, \NP-hardness holds when $k/r=k$ as well as when $k/r=\OO(\log^{1/c} k)$ for any fixed constant $c\geq 1$.

As an open problem, both Abasi et al.~and Gabizon et al.~asked whether it is possible to avoid an exponential dependency on $\log r$. In other words, they asked whether the above problems are single-exponentially \FPT\ when parameterized by $k/r$ alone.\footnote{The interpretation of $k/r$ is a tight lower bound on the number of distinct elements any solution must use.}  To answer this question for \monomDetR,
Bonamy et al.~\cite{DBLP:conf/esa/BonamyKPSW17} proved that the running time of the algorithms of Abasi et al.~\cite{DBLP:conf/mfcs/AbasiBGH14} and of Gabizon et al.~\cite{DBLP:conf/esa/GabizonLP15} for  \monomDetR\ are optimal under the Exponential Time Hypothesis (\ETH) in the following sense. Unless \ETH\ fails there is no $2^{o((k/r)\log r)}\cdot (n+\log k)^{\OO(1)}$-time algorithm for \monomDetR\ even if $r=\Theta(k^{\sigma})$ for any $\sigma \in [0,1).$ The question remains open for \diPathR\ and \setPackingR.

We consider the question from a wider perspective of parameterized complexity: are the above problems \FPT\ when parameterized by $k/r$ only, i.e., whether there exists a computable function $f$ such that the problems admit a $f(k/r)(n+\log k)^{\OO(1)}$-time algorithm? 


 
 Note that for \setPackingR, $r\le m$ and thus the above algorithm by Gabizon et al. shows that the problem is in XP. However, for \diPathR\  
the above algorithms by Abasi et al.~and Gabizon et al.~are not even \XP-algorithms in the parameter $k/r$ because $r$ (encoded in binary) can be much larger than the size of the problem instance under consideration. In particular, even when $k/r=1$, these algorithms can run in time exponential in the input size. In addition, note that all three problems are easily seen to be \FPT\ when parameterized by $k/r$ and $r$ simultaneously, since algorithms that run in time $2^{\OO(k)}n^{\OO(1)}$ immediately follow by simple modifications of known algorithms for the corresponding non-relaxed versions. When $r$ is large enough, the running times of $2^{\OO((k/r)\log r)}\cdot n^{\OO(1)}$ of the algorithms by Abasi et al.~and Gabizon et al.~are superior. Here, the $\log r$ factor in the exponent naturally arises, and seems to be perhaps unavoidable. To see this, first consider the very special case where the input contains only $\OO(k/r)$ distinct elements. Then, we can store {\em counters} that keep track of how many times each element is used. Our array of counters would have $2^{\OO((k/r)\log r)}$ possible configurations, hence a running time of $2^{\OO((k/r)\log r)}\cdot n^{\OO(1)}$ is trivial. However, counters are completely prohibited when dependence on $r$ is forbidden, which already renders this extreme special case  non-obvious. In fact, a running time of $f(k/r)\cdot (n+\log k)^{\OO(1)}$ not only disallows using such an array of counters,  but it forbids the usage of {\em even a single counter}. Thus, in advance, it might seem more natural to vote for \WO-hardness over \FPT\ for all three problems with respect to $k/r$.

\paragraph{Our Contribution.} We resolve the parameterized complexity of all three problems, namely \diPathR, \setPackingR\ and \monomDetR,  with respect to the parameter $k/r$. Our main contribution consists of a $2^{\OO((k/r)^2\log(k/r))}\cdot (n+\log k)^{\OO(1)}$-time algorithm for \diPathR\ and a $2^{\OO(k/r)}\cdot (n+\log k)^{\OO(1)}$-time algorithm for \undiPathR.\footnote{Recall that $n$ is the number of vertices in the input (di)graph.} For \undiPathR, this answers the question posed by Abasi et al.~\cite{DBLP:conf/mfcs/AbasiBGH14} and Gabizon et al.~\cite{DBLP:conf/esa/GabizonLP15}, and reiterated by Bonamy et al.~\cite{DBLP:conf/esa/BonamyKPSW17} and Socala~\cite{Socala2017}. (As also noted in previous works, it is easily seen that even when $k$ is polynomial in $n$, none of the three problems can be solved in time $2^{o(k/r)}\cdot n^{\OO(1)}$ unless the \ETH\ fails.)
In addition, we show that \setPackingR\ is \FPT\ based on the representative sets method in parameterized algorithmics. Along the way to design this algorithm, we obtain a polynomial kernel for any fixed $p$, which resolves another question posed by Gabizon et al.~regarding the existence of polynomial kernels for problems with relaxed disjointness constraints whose sizes are decreasing functions of $r$. We remark that all of our algorithms are deterministic, and are based on ideas completely different from those of Abasi et al.~\cite{DBLP:conf/mfcs/AbasiBGH14} and of Gabizon et al.~\cite{DBLP:conf/esa/GabizonLP15}. 

Next, we introduce the following extension of \setPackingR\ to multisets:

\begin{defproblem}
{\multisetPackingR}
{An $n$-element universe $V$, positive integers $p,q,r$, and a collection $\cal H$ of mutisets of size~$p$ whose elements belong to $V$.}
{Does there exist a subcollection ${\cal H}'$ of $\cal H$ of size $q$ where no element of~$V$ has more than $r$ occurrences in total (i.e., if a multiset $H$ in ${\cal H}'$ contains $t$ copies of element $v\in V$, all other multisets of ${\cal H}'$  can have at most $r-t$ occurrences of $v$ in total)? (We will call ${\cal H}'$ an $r$-{\em relaxed parking}.)}
\end{defproblem}\\
We prove that  \multisetPackingR\ parameterized by $k/r$ is \WOH. Using this result, we also prove that \monomDetR\ parameterized by $k/r$ is \WOH\ even if {\em (i)} $k$ is polynomially bounded in the input length, {\em (ii)} the number of distinct variables is $k/r$, and {\em (iii)} the circuit is non-canceling. Moreover, we show that \monomDetR\ is \paraH\ even if the input polynomial has only two variables
and the circuit is non-canceling.

The most technical parts of the paper deal with the \diPathR\ and \undiPathR\ problems. We prove that \diPathR\ can be solved in time $2^{\OO((k/r)^2\log(k/r))}\cdot (n+\log k)^{\OO(1)}$ and polynomial space using a chain of reductions from \diPathR\ that includes three auxiliary problems. The first 
 of these problems is the \diPathRst\ problem, where we are given a strongly connected digraph $G$, positive integers $k,r$, and vertices $s,t\in V(G)$. The objective is to either {\em (i)} determine that $G$ has an $r$-simple $k$-path between any pair of vertices or {\em (ii)} output the largest integer $i\le k$ such that $G$ has an $r$-simple $(s,t)$-path of size $i$. It is not hard to see that we may assume that $G$ has neither a path of size at least $2k/r$ nor a cycle of length at least $k/r$. The key result on \diPathRst\  is that under the assumption above, there is always, as a solution, an $r$-simple path with fewer than $30(k/r)^2$ distinct arcs.\footnote{In addition, we show that this bound is essentially tight.} For reductions using the other two problems we apply several parameterized algorithms approaches (including color coding and integer linear programming  parameterized by the number of variables) and new structural insights. Here, we often alternate between the view of the solution as an $r$-simple $k$-path and the view of the solution as an Eulerian digraph with degree constraints.

Our proof that \undiPathR\ can be solved in time $2^{\OO\left(\frac{k}{r}\right)}(n+\log k)^{\OO(1)}$, initially uses an approach similar to that applied for \diPathR. Using the fact that the input graph is undirected, we are able to show that the $30(k/r)^2$ bound above can be improved to $30(k/r)$. However, this result in itself is only sufficient to show the existence of an $2^{\OO\left(\frac{k}{r}\log(\frac{k}{r})\right)}(n+\log k)^{\OO(1)}$-time algorithm for \undiPathR\ using the reductions applied for \diPathR. Thus, we have to take a different route based on a deeper understanding of the structure of the solution. Our approach is partially inspired by an idea from the recent work of Berger et al.~\cite{DBLP:journals/corr/abs-1804-06361} and involves a special decomposition of the multigraph induced by a solution for \undiPathR\ into two multigraphs. In our case, one of the multigraphs, $H$, has treewidth at most 2, and all vertices of $H$ are of even degree and different color (in a special coloring), i.e.,~$H$ is {\em colorful}. The second multigraph corresponds to an $r$-simple path $W$ which visits each component of $H$ (which ensures the connectivity of the generated solution), and vertices of the same color are visited by $W$ in total a prescribed number of times. The existence of the decomposition above is verified by a two-level dynamic programming algorithm. This algorithm is followed by a way to bound $r$. Here, we identify that when $r$ is large enough compared to $k$, then the vertex cover number of the graph can be bounded. The decomposition is modified accordingly to enable the use of a flow network to handle its second multigraph. 

\paragraph{Related Work.}  Agrawal et al.~\cite{DBLP:conf/icalp/AgrawalLMMS16} showed the power of relaxed disjointness conditions in the context of a problem that otherwise admits no polynomial kernel. Specifically, Agrawal et al.~studied the {\sc Disjoint Cycle Packing} problem: given a graph $G$ and integer $k$, decide whether $G$ has $k$ vertex-disjoint cycles. It is known that this problem does not admit a polynomial kernel unless \NP\ $\subseteq$ \textsf{coNP/poly}~\cite{DBLP:journals/tcs/BodlaenderTY11}. The main result by Agrawal et al.~concerns a relaxation of {\sc Disjoint Cycle Packing} where every vertex can belong to at most $r$ cycles (rather than at most one cycle). Agrawal et al.~showed that this relaxation reveals a spectrum of upper and lower bounds. In particular, they obtained a (non-polynomial) kernel of size $\OO(2^{(k/r)^2}
k^{7+(k/r)} \log^3 k)$ when $(k/r) = o(\sqrt{k})$. Note that the size of the kernel depends on $k$.

Prior to the work by Gabizon et al.~\cite{DBLP:conf/esa/GabizonLP15}, packing problems with relaxed disjointness conditions have already been considered from the viewpoint of parameterized complexity (see, e.g., \cite{DBLP:journals/ijfcs/Lopez-OrtizPR18,DBLP:journals/toct/FernauLR15,DBLP:conf/csr/RomeroL14,DBLP:conf/walcom/RomeroL14}). Roughly speaking, these papers do not exhibit behaviors where relaxed disjointness conditions substantially (or at all) simplify the problem at hand, but rather provide parameterized algorithms and kernels with respect to $k$. Here, the work most relevant to us is that by Fernau et al.~\cite{DBLP:journals/toct/FernauLR15}, who studied the \setPackingR\ problem. In particular, for any $r\geq 1$, Fernau et al.~proved that several very restricted versions of \setPackingR\ with $p=3$ are already \NP-hard. Moreover, they obtained a kernel with $\OO((p+1)^pk^p)$ vertices.

In addition, we note that Gabizon et al.~\cite{DBLP:conf/esa/GabizonLP15} also studied the {\sc Degree-Bounded Spanning Tree} problem: given a graph $G$ and an integer $d$, decide whether $G$ has a spanning tree of maximum degree at most $d$. This problem demonstrates a limitation of the derandomization of Gabizon et al.~as the arithmetic circuit required is not non-canceling. Thus, only a randomized $2^{\OO((n/d)\log d)}$-time algorithm was obtained and designing a deterministic algorithm of such a running time remains an open problem. 

Finally, let us remark that {\sc $k$-Path} (on both directed and undirected graph) and {\sc $p$-Set $q$-Packing} are both among the most extensively studied problems in Parameterized Complexity. In particular, after a long sequence of works during the past three decades, the current best known parameterized algorithms for {\sc $k$-Path} have running times $1.657^kn^{\OO(1)}$ (randomized, undirected only)~\cite{DBLP:journals/jcss/BjorklundHKK17,DBLP:journals/siamcomp/Bjorklund14} (extended in \cite{DBLP:journals/siamdm/BjorklundKKZ17}), $2^kn^{\OO(1)}$ (randomized)~\cite{DBLP:journals/ipl/Williams09} and $2.597^kn^{\OO(1)}$ (deterministic)~\cite{DBLP:conf/esa/Zehavi15,DBLP:journals/jacm/FominLPS16,DBLP:journals/jcss/ShachnaiZ16}. In addition, {\sc $k$-Path} is known not to admit any polynomial kernel unless \NP\ $\subseteq$ \textsf{coNP/poly}~\cite{DBLP:journals/jcss/BodlaenderDFH09}.

 This paper is organized as follows. The next section contains preliminaries. Section \ref{sec:directedPath} describes reductions leading to our main result for \diPathR. Our proof of the main result for \undiPathR\ is given in Section \ref{sec:undirectedPath}. We show that \setPackingR\ parameterized by $(k/r)$ is FPT in Section \ref{sec:packingFPT}. In Section \ref{sec:monomDet}, we prove that \monomDetR\ is para-\NP-hard. Our  \WO-hardness results for  \multisetPackingR\ and \monomDetR\ are shown in Section \ref{sec:W1}. The last section of the paper discusses some open problems.  

\section{Preliminaries}\label{sec:prelims}

Given a multiset $M$ and an element $e\in M$, $[i] e$ stands for $i$ copies of $e$. The {\em size} of a multiset $M=\{[i_1]e_1,\dots ,[i_p]e_p\}$ is $\sum_{j=1}^p i_j.$

\paragraph{Graph Terminology and Notation.} For a directed or undirected graph $G$, the vertex set of $G$ is denoted by $V(G)$. If $G$ is undirected, its edge set is denoted by $E(G)$, and if $G$ is directed, its arc set is denoted by $A(G)$. Given a subset $U\subseteq V(G)$, the subgraph of $G$ induced by $U$ is denoted by $G[U]$, and the subgraph of $G$ obtained by deleting the vertices in $U$ and the edges/arcs incident to them is denoted by $G-U$. Given a subset of edges/arcs $U$ in $G$, the subgraph of $G$ obtained by deleting the edges/arcs in $U$ is denoted by $G-U$. For a directed multigraph $G$ and a vertex $v\in V(G)$, the out-degree and in-degree of $v$ in $G$ are denoted by $d^+(v)$ and $d^-(v)$, respectively. 

A digraph $G$ is {\em strongly connected} if for any pair $u,v$ of distinct vertices, $G$ has a path from $u$ to $v$. 
The {\em underlying undirected graph} of a directed graph $G$ is an undirected graph $U(G)$ with the same vertex set and $uv\in E(U(G))$ if and only if either $uv\in A(G)$ or $vu\in A(G)$ (or both).  A digraph~$G$ is {\em weakly connected} if $U(G)$ is connected. The {\em weakly connected components} of a digraph~$G$ are subgraphs of $G$ induced by the vertex sets of connected components of $U(G).$ 
A {\em directed acyclic graph (DAG)} is a digraph with no directed cycle. For any positive integer $\ell\in\mathbb{N}$, an {\em $\ell$-colored (di)graph} is a vertex-colored (di)graph where each vertex is colored by exactly one color from $\{1,2,\ldots,\ell\}$. 

For an undirected graph $G$, a {\em vertex cover} of $G$ is a subset of vertices $U\subseteq V(G)$ such that every edge in $E(G)$ is incident to at least one vertex in $U$, and a {\em matching} in $G$ is a subset of edges $U\subseteq E(G)$ such that no two edges in $U$ have a common endpoint. A matching $U$ is {\em maximal} if there does not exist $e\in E(G)\setminus U$ such that $U\cup\{e\}$ is a matching. The {\em vertex cover number} of $G$ is the minimum size of a vertex cover of $G$. A {\em cactus} is a connected graph in which any two cycles have at most one vertex in common. 
For an undirected multigraph $G$ and a vertex $v\in V(G)$, the {\em degree} $d(v)$ of  $v$ is the number of edges incident to $v.$ The {\em underlying simple graph} $G$ of an undirected multigraph $H$ is obtained from $G$ by deleting all but one edge among every set of multiple edges.

\paragraph{Paths, Walks and Trails.}  For an undirected multigraph $G,$ a walk $W$ is an alternating sequence $v_1e_1v_2\dots e_{\ell-1}v_\ell$ such that $e_i$ is an edge between $v_i$ and $v_{i+1}$ for all $i\in\{1,2,\ldots,\ell-1\}$. For a directed multigraph $G,$ the definition of a walk is the same, but we require that $e_i$ is an arc from $v_i$ to $v_{i+1}.$ 
When $G$ is a graph, i.e., has no multiple edges/arcs, then $W$ will be denoted by $v_1-v_2-v_3-\ldots-v_\ell$.
For any $i\in\{1,2,\ldots,\ell\}$, $v_i$ is called a {\em vertex occurrence} or a {\em vertex visit}, and for any $i\in\{1,2,\ldots,\ell-1\}$, $\{v_{i-1},v_i\}$ (resp.~$(v_{i-1},v_i)$) is called an {\em edge occurrence} ({\em arc occurrence}) or an {\em edge visit} ({\em arc visit}), respectively. The {\em length} of a walk is the number of edges/arcs visits on the walk, that is, $\ell-1$, and the {\em size} of a walk is the number of vertex visits on the walk, that is, $\ell$. If the first and last vertex visits of a walk are equal, then the walk is said to be {\em closed}. For a walk $P$, the {\em multisets} of vertex visits and edge (arc) visits are denoted by $V(P)$ and $E(P)$ ($A(P)$), respectively.

An {\em $r$-simple path} is a walk where every vertex occurs at most $r$ times. Moreover, an {\em $r$-simple $k$-path} is an $r$-simple path of size $k$. Note that a $1$-simple path is just a path. A {\em cycle} is a closed walk where every vertex occurs once, except for the last and first vertex which occurs twice. Note that by this definition, the first and last vertex of a cycle are well defined. Given vertices $s,t\in V(G)$, an {\em $(s,t)$-path} is a path that starts at $s$ and ends at $t$. Similarly, an {\em $(s,t)$-cycle} is a cycle that starts at $s$ and ends at $t$, in which case $s=t$. To avoid writing some explanations twice, we refer to an $(s,s)$-cycle also as an $(s,s)$-path.
More generally, an {\em $r$-simple $(s,t)$-path} is an $r$-simple $k$-path that starts at $s$ and ends at $t$.

Given a directed or undirected multigraph $G$ and vertices $s,t\in V(G)$, a walk $W$ in $G$ is called an {\em Euler $(s,t)$-trail} if $W$ visits every edge/arc in $G$ exactly once, and starts at $s$ and ends at $t$.  A directed multigraph $G$ is {\em balanced} if $d^+(v)=d^-(v)$ for every vertex $v$ of $G.$ Let $s,t$ be distinct vertices of a directed multigraph $G$. Then $G$ is {\em $(s,t)$-almost balanced} if $d^+(v)=d^-(v)$ for every vertex $v\in V(G)\setminus \{s,t\},$ $d^+(s)=d^-(s)+1$ and $d^+(t)=d^-(t)-1.$
An undirected multigraph $G$ is called {\em even} if for every $v\in V(G)$, $d(v)$ is even.

\paragraph{Perfect Hash Families.} The construction of a perfect hash family is a basic tool to derandomize parameterized algorithms. Formally, perfect hash families are defined as follows.

\begin{definition}\label{def:hash}
Let $n,k\in\mathbb{N}$, $n\ge k.$ An {\em $(n,k)$-perfect hash family} $\cal F$ is a family of functions $f: \{1,2,\ldots,n\}\rightarrow\{1,2,\ldots,k\}$ such that for any subset $I\subseteq \{1,2,\ldots,n\}$ of size $k$, there exists a function in $\cal F$ that is injective on $I$.
\end{definition}

The following proposition asserts that small perfect hash families can be constructed efficiently.

\begin{theorem}[\cite{DBLP:journals/jacm/AlonYZ95,DBLP:journals/jcss/New18}]\label{prop:hashfam}
Let $n,k\in\mathbb{N}$. An $(n,k)$-perfect hash family of size $e^{k+o(k)}\log n$ can be constructed in $e^{k+o(k)}n\log n$ time. Moreover, the functions in the family can be enumerated with polynomial space and polynomial delay in $n$.
\end{theorem}

\paragraph{Treewidth.} Tree decompositions and treewidth are defined as follows.

\begin{definition}\label{def:tw} A {\em tree decomposition} of a graph $G$ is a pair $(T,\beta)$, where $T$ is a rooted tree and $\beta: V(T)\rightarrow 2^{V(G)}$ is a mapping that satisfies the following conditions.
\begin{enumerate}
\item\label{cond:tw1} For each vertex $v\in V(G)$, the set $\{x\in V(T): v\in\beta(x)\}$ induces a nonempty (connected) subtree of $T$.
\item\label{cond:tw2} For each edge $\{u,v\}\in E(G)$, there exists $x\in V(T)$ such that $\{u,v\}\subseteq\beta(x)$.
\end{enumerate}
The {\em width} of $(T,\beta)$ is $\max_{v\in V(T)}\{|\beta(v)|\}-1$. The {\em treewidth} of $G$ is the minimum width over all tree decompositions of $G$.
\end{definition}

The vertices of $T$ are called nodes. A set $\beta(x)$ for $x\in V(T)$ is called the {\em bag at $x$}. 

A {\em nice tree decomposition} is a tree decomposition of a form that simplifies the design of dynamic programming (DP) algorithms. 

\begin{definition}
A tree decomposition $(T,\beta)$ of a graph $G$ is {\em nice} if 
each node $x\in V(T)$ is of one of the following types.
\begin{itemize}
\item {\bf Leaf}: $x$ is a leaf in $T$ and $\beta(x)=\emptyset$.
\item {\bf Forget}: $x$ has exactly one child $y$, and there exists a vertex $v\in\beta(y)$ such that $\beta(x)=\beta(y)\setminus\{v\}$.
\item {\bf Introduce}: $x$ has exactly one child $y$, and there exists a vertex $v\in\beta(x)$ such that $\beta(x)\setminus\{v\}=\beta(y)$.
\item {\bf Join}: $x$ has exactly two children, $y$ and $z$, and $\beta(x)=\beta(y)=\beta(z)$.
\end{itemize}
\end{definition}

It is well-known that a graph $G$ of treewidth $\tw$ admits a nice tree decomposition of width $\tw$ (see, e.g., \cite{DBLP:books/sp/Kloks94,DBLP:books/sp/CyganFKLMPPS15}). 

\begin{theorem}[\cite{DBLP:books/sp/Kloks94}]\label{prop:niceTreeDecomp}
Let $G$ be a graph of treewidth $\tw$. Then, $G$ admits a nice tree decomposition of width $\tw$.
\end{theorem}

\paragraph{Integer Linear Programming (ILP)} 
The {\sc Feasibility Linear Programming} problem ({\sc Feasibility LP}) is given by a set $X$ of variables and a system of linear equations and inequalities with real-valued coefficients and variables from $X$, and the aim is  to decide whether all the linear equations and inequalities (called {\em linear constraints}) can be satisfied by an assignment $\alpha$ of non-negative reals to variables in $X$. If only integral values are allowed in $\alpha$, then the problem is called the {\sc Feasibility integer Linear Programming} problem ({\sc Feasibility ILP}).  The {\sc Linear Programming} problem ({\sc LP}) is given by a set $X$ of variables, a system of linear constraints with real-valued coefficients and variables from $X$ and a linear function $z$ with real-valued coefficients and variables from $X$, and the aim is  to find an assignment $\alpha$ of non-negative reals to variables in $X$ that satisfies all linear constraints and minimizes/maximizes $z$ over all such (feasible) assignments. If only integral values are allowed in $\alpha$, then the problem is called the {\sc Integer Linear Programming} problem ({\sc ILP}). The {\em cost vector} $c$ is the vector of coefficients of the variables in $X$ in the function $z$.

The following well-known result (cf.~Section 6.2 in \cite{DBLP:books/sp/CyganFKLMPPS15}) will be used in this paper.

\begin{theorem}[\cite{DBLP:journals/mor/Lenstra83,DBLP:journals/mor/Kannan87,DBLP:journals/combinatorica/FrankT87}]\label{prop:ILP}
{\sc ILP}  ({\sc Feasibility ILP}, resp.) of size $L$ with $p$ variables can be solved using
\[\OO(p^{2.5p+o(p)}\cdot(L+\log M_x)\cdot\log(M_xM_c)) \mbox{    } (\OO(p^{2.5p+o(p)}L))  \]
arithmetic operations and space polynomial in $L+\log M_x$ (in $L$), respectively.
Here $M_x$ is an upper bound on the absolute value a variable can take in a solution, and $M_c$ is the largest absolute value of a coefficient in the cost vector $c$.
\end{theorem}

\paragraph{Flow Networks.} 
A {\em flow network} is a digraph $N=(V,A)$ with two special vertices $s$ and $t$ called a {\em source} and {\em sink}, respectively, and two functions $u:\ A\rightarrow \mathbb{R}_{\ge 0}$ and $c:\ A\rightarrow \mathbb{R}_{\ge 0}$. For an arc $a\in A$, $u(a)$ and $c(a)$ are  the {\em upper capacity} and {\em cost} of $a.$ A {\em flow} in $N$ is a function $f: A\rightarrow \mathbb{R}_{\ge 0}$ such that $f(a)\le u(a)$ for every $a\in A$ and $\sum_{xv\in A} f(xv)=\sum_{vy\in A} f(vy)$ for every $v\in V\setminus \{s,t\}$. The {\em value} of $f$ is $\sum_{sy\in A} f(sy)$ and the {\em cost} of~$f$ is $\sum_{a\in A} c(a)f(a).$ It is well-known that  $\sum_{sy\in A} f(sy)=-\sum_{xt\in A} f(xt)$ \cite{DBLP:books/JBJGG}. A flow $f$ is {\em integral} if $f(a)$ is an integer for every $a\in A$.

\paragraph{Arithmetic Circuits.} 
Let $M$ be a monomial in a polynomial $P$. The {\em degree} of $M$ is the sum of degrees of the variables of $M$. An arithmetic circuit $C$ over the field $\mathbb{F}$ and the set $X$ of 
variables is a DAG $D$ as follows. Every vertex in $D$ with in-degree zero is called an {\em input gate} and is labeled by either a variable $x\in X$ or a field element in $\mathbb{F}.$ Every other gate is labeled by either $+$ or $\times$ and called a {\em sum gate} and a {\em product gate}, respectively. The {\em size} of $C$ is the number of gates in $C$, and the {\em depth} of $C$ is the length of the longest directed path in $C$. A circuit $C$ computes a polynomial $P$ in the following natural way. An input gate computes the polynomial it is labeled by. A sum (product) gate $v$ computes the sum (product), respectively, of the polynomials computed by its in-neighbors in $D$. 
A circuit is called  {\em non-cancelling}, if its input gates are labelled only by variables (no labelling by field elements).

\paragraph{Parameterized Complexity.} 
A parameterized problem $\Pi$ can be considered as a set of pairs
$(I,k)$ where $I$ is the \emph{problem instance} and $k$ (usually a nonnegative
integer) is the \emph{parameter}.  $\Pi$ is called
\emph{fixed-parameter tractable (FPT)} if membership of $(I,k)$ in
$\Pi$ can be decided by an algorithm of runtime $\OO(f(k)|I|^c)$, where $|I|$ is the size
of $I$, $f(k)$ is a computable function of the
parameter $k$ only, and $c$ is a constant
independent from $k$ and $I$. Such an algorithm is called an {\em FPT} algorithm.
Let $\Pi$ and $\Pi'$ be parameterized
problems with parameters $k$ and $k'$, respectively. An
\emph{FPT-reduction~$R$ from $\Pi$ to $\Pi'$} is a many-to-one
transformation from $\Pi$ to $\Pi'$, 
mapping each instance $(I,k)$ to an output $(I',k')$
such that (i) $(I,k)\in \Pi$ if
and only if $(I',k')\in \Pi',$ (ii) $k'\le g(k)$ for a fixed
computable function $g$, and (iii) $R$ is of complexity
$\OO(f(k)|I|^c)$.

When the decision time is replaced by the much more powerful $|I|^{\OO(f(k))},$
we obtain the class XP, where each problem is polynomial-time solvable
for any fixed value of $k.$ There is a number of parameterized complexity
classes between FPT and XP (for each integer $t\ge 1$, there is a class W[$t$]) and they form the following tower:
$$FPT \subseteq W[1] \subseteq W[2] \subseteq \dots \subseteq XP.$$
For the definition of classes W[$t$],
see, e.g., \cite{DBLP:series/txcs/DowneyF13,DBLP:books/sp/CyganFKLMPPS15}. Due to a number of results obtained, 
it is widely believed that FPT$\neq$W[1], i.e., no W[1]-hard problem admits an FPT algorithm \cite{DBLP:series/txcs/DowneyF13,DBLP:books/sp/CyganFKLMPPS15}. 

A parameterized problem $\Pi$ is in \emph{para-NP} if membership of $(I,k)$ in
$\Pi$ can be decided in nondeterministic time $O(f(k)|I|^c)$,
where $|I|$ is the size of $I$, $f(k)$ is a computable function of the
parameter~$k$ only,
and $c$ is a constant independent from $k$ and $I$. Here,
nondeterministic time means that we can use nondeterministic
Turing machine. A parameterized problem $\Pi'$ is {\em
para-NP-hard}, if for any parameterized
problem $\Pi$ in para-NP there is an FPT-reduction from $\Pi$ to
$\Pi'$.

For a parameterized problem  $\Pi$,
a \emph{generalized kernelization from $\Pi$ to $\Pi'$} is a polynomial-time
algorithm $\cal A$ that maps an instance $(I,k)$ to an instance $(I',k')$ (the
\emph{generalized kernel}) such that (i)~$(I,k)\in \Pi$ if and only if
$(I',k')\in \Pi'$, (ii)~$k'\leq g(k)$ for some computable 
function $g$, and (iii)~$|I'|\leq g(k).$  The function $g(k)$ is called the {\em size} of the generalized kernel.
If $\Pi = \Pi'$, $\cal A$ is a \emph{kernelization} and $(I',k')$ is a {\em kernel}~\cite{KernelBook}.

\section{Directed $r$-Simple $k$-Path: FPT}\label{sec:directedPath}

In this section, we focus on the proof of the following theorem.

\begin{theorem}\label{thm:directedPathFPT}
\diPathR\ is \FPT\ parameterized by $k/r$. In particular, \diPathR\ is solvable in time $2^{\OO((\frac{k}{r})^2\log(\frac{k}{r}))}(n+\log k)^{\OO(1)}$ and polynomial space.
\end{theorem}

We remark that by polynomial space, we mean polynomial in $n+\log k$.

\subsection{Reduction to a Simpler Problem}\label{sec:directedPathReduction}

In order to prove Theorem \ref{thm:directedPathFPT}, we begin with two simple claims that reduce the \diPathR\ problem to a special case of a related problem that is defined as follows. 

\begin{defproblem}
{\diPathRst}
{A digraph $G$, positive integers $k,r$, and vertices $s,t\in V(G).$}
{Either {\em (i)} determine that $G$ has an $r$-simple $k$-path or {\em (ii)} output the largest integer $i\le k$ such that $G$ has an $r$-simple $(s,t)$-path of size $i.$} 
\end{defproblem}\\
We first observe that \diPathR\ can be reduced to the special case of \diPathRst\ where the input digraph is strongly connected.

\begin{lemma}\label{lem:stronglyConnected}
Suppose that \diPathRst\ on strongly connected digraphs can be solved in time $f(k/r)\cdot (n+\log k)^{\OO(1)}$ and polynomial space. Then, \diPathR\ can be solved in time $f(k/r)\cdot (n+\log k)^{\OO(1)}$ and polynomial space.
\end{lemma}

\begin{proof}
Let $\cal A$ be an algorithm that solves \diPathRst\ on strongly connected digraphs in time $f(k/r)\cdot (n+\log k)^{\OO(1)}$ and polynomial space. In what follows, we describe how to solve \diPathR.
To this end, let $(G,k,r)$ be an instance of \diPathR. Let $\cal C$ be the set of strongly connected components of $G$. For every component $C\in {\cal C}$, and vertices $s,t\in V(C)$, we perform the following computation. We call $\cal A$ with $(C,k,r,s,t)$ as input. If $\cal A$ concludes that $C$ has an $r$-simple $k$-path, then we correctly conclude that $(G,k,r)$ is a \yes-instance. Else, we denote by $k_{st}$ the integer that  $\cal A$ outputs. Then, $k_{st}\leq k$ is the largest integer $p$ such that $C$ has an $r$-simple $(s,t)$-path of size $p$.
So far, the time spent is at most $f(k/r)\cdot (n+\log k)^{\OO(1)}$ and the space used is polynomial.

Let $C_1,C_2,\ldots,C_{|{\cal C}|}$ be an ordering of the components in $\cal C$ with the property that for all $i<j$, $u\in C_j$ and $v\in C_i$, it holds that $(u,v)\notin A(G)$. 
Now, we solve \diPathR\ by dynamic programming (DP) as follows. 

Let $\mathsf{M}$ be a DP vector with an entry $\mathsf{M}[v]$ for every $v\in V(G)$. This entry will store the largest integer $p$ such that $G$ has an $r$-simple $p$-path that ends at $v$.
At Step 1, we set $\mathsf{M}[v]=\max_{u\in V(C_1)}k_{uv}$ for every $v\in V(C_1)$. At Step $i$, where $i=2,3,\dots , |\cal C|$, we set
\[\mathsf{M}[v] = \displaystyle{\max\left\{\max_{(u,w)\in A(G) \atop\ \mathrm{s.t.}~u\notin V(C_i), w\in V(C_i)}\left(\mathsf{M}[u] + k_{wv}\right), \max_{u\in V(C_i)}k_{uv}\right\}}\]
for every $v\in V(C_i)$.

It is straightforward to verify that the DP computation is correct and can be executed using polynomial time and space. After this computation is terminated, we correctly conclude that $(G,k,r)$ is a \yes-instance if and only if there exists  $v\in V(G)$ such that $\mathsf{M}[v]\geq k$. This completes the proof.
\end{proof}

From now on, we focus on the \diPathRst\ problem on strongly connected digraphs. Our second claim shows that the existence of a ``long'' path or a ``long'' cycle in the input digraph $G$ implies that it has an $r$-simple $k$-path.

\begin{lemma}\label{lem:longCycPath}
Let $G$ be a strongly connected digraph. If any of the following two conditions is satisfied, then $G$ has an $r$-simple $k$-path.
\begin{itemize}
\item The graph $G$ has a cycle of length at least $k/r$.
\item The graph $G$ has a path with at least $2k/r$ vertices.
\end{itemize}
\end{lemma}

\begin{proof}
First, suppose that $G$ has a cycle $C$ of length at least $k/r$. Then, $C$ is a sequence of distinct vertices, besides the first and last vertex, of the form $v_1-v_2-\cdots-v_\ell$ for some integer $\ell\geq k/r+1$. In this case, $(v_1-v_2\cdots-v_{\ell-1})-(v_1-v_2\cdots-v_{\ell-1})-\cdots-(v_1-v_2\cdots-v_{\ell-1})$ where $v_1-v_2\cdots-v_{\ell-1}$ is duplicated exactly $r$ times, is an $r$-simple $k'$-path for $k'=r(\ell-1)\geq k$. Thus, $G$ has an $r$-simple $k$-path.

Second, suppose that $G$ has a path $P$ with at least $2k/r$ vertices. Then, $P$ is a sequence of $\ell$ distinct vertices for some integer $\ell\geq 2k/r$. Since $G$ is strongly connected, it has at least one path from the last vertex of $P$ to the first vertex of $P$. Let $Q=u_1-u_2-\ldots-u_q$ denote any such path. Moreover, let $Q'$ denote the subsequence of $Q$ where the first and last vertex visits are omitted. If $r\mod 2=0$, denote 
$W=(P-Q')-(P-Q')-\cdots-(P-Q'),$ 
where $P-Q'$ is duplicated exactly $r/2$ times, and otherwise denote $W=(P-Q')-(P-Q')-\cdots-(P-Q')-P$ where $P-Q'$ is duplicated exactly $(r-1)/2$ times. Since every vertex occurs at most twice in $P-Q'$, we have that every vertex occurs at most $r$ times in $W$. Moreover, if $r\mod 2=0$, then the size of $W$ is $r(\ell+q-2)/2\geq r\ell/2\geq k$, and otherwise the size of $W$ is $(r-1)(\ell+q-2)/2 + \ell\geq (r-1)\ell/2 + \ell \geq k$. Therefore, $W$ is an $r$-simple $k'$-path for some integer $k'\geq k$, which means that $G$ has an $r$-simple $k$-path.
\end{proof}

The following known proposition asserts that we can efficiently determine whether the input digraph has a long  path or a long cycle.

\begin{theorem}[\cite{DBLP:journals/ipl/New18,DBLP:journals/ipl/Zehavi16}]\label{prop:longPath}
There exists a deterministic algorithm that given a digraph $G$, vertices $s,t\in V(G)$, and $k\in\mathbb{N}$, determines in time $2^{\OO(k)}\cdot n^{\OO(1)}$ and polynomial space whether $G$ has a path from $s$ to $t$ on at least $k$ vertices.
\end{theorem}

Thus, from now on, we may assume not only that the input digraph is strongly connected, but that it also has neither a path of size at least $2k/r$ nor a cycle of length at least $k/r$. Accordingly, we say that an instance $(G,k,r,s,t)$ of \diPathRst\ is {\em nice} if $G$ is strongly connected and it has neither a path with at least $2k/r$ vertices nor a cycle of length at least $k/r$. Moreover, we say that $(G,k,r,s,t)$ is {\em positive} if $G$ has an $r$-simple $k$-path, and otherwise we say that it is {\em negative}.
\subsection{Bounding the Number of Distinct Arcs}\label{sec:numDistinctArc}

Having established the two simple claims above, the second part of our proof concerns the establishment of an upper bound on the number of distinct arcs in at least one $r$-simple $k$-path (if at least one such walk exists) or at least one $r$-simple $(s,t)$-path of maximum size. The main definition in this part of the proof is the following one.

\begin{definition}\label{def:simpleMulti}
Let $(G,k,r,s,t)$ be an instance of \diPathRst. Let $P$ be an $r$-simple path in $G$. 
\begin{itemize}
\item Let $P_{\mathrm{simple}}$ be the (directed) subgraph of $G$ that consists of the vertices and arcs in $G$ that are visited at least once by $P$,
and  let $P_{\mathsf{multi}}$ be the directed multigraph obtained from $P_{\mathrm{simple}}$ by replacing each arc $a$ by 
its $c_a$ copies, where $c_a$ is the number of times $a$ is visited by $P.$
\item Let $V(P,r)$ be the set that contains $s,t$ and every vertex that occurs $r$ times in $P$, and $P_{\mathrm{simple}}^{-r} = P_{\mathrm{simple}} - V(P,r)$.
\item For any two (not necessarily distinct) vertices $u,v\in V(P)$, denote $P_{\mathrm{simple}}^{u,v,-r} = P_{\mathrm{simple}} - (V(P,r)\setminus\{u,v\})$.  
(In case $u,v\notin V(P,r)$, it holds that $P_{\mathrm{simple}}^{u,v,-r} = P_{\mathrm{simple}}-V(P,r)$.)
\end{itemize}
\end{definition}

Before we begin our analysis, we relate our problem to the notion of an Euler trail by a well-known proposition, to which we will repeatedly refer later.

\begin{theorem}[\cite{DBLP:books/JBJGG,DBLP:books/daglib/0030488}]\label{prop:euler}
Let $G$ be a weakly connected directed multigraph. Let $s,t\in V(G)$.
\begin{itemize}
\item If $s\neq t$, then there exists an Euler $(s,t)$-trail in $G$ if and only if $G$ is $(s,t)$-almost balanced.
\item If $s=t$, then there exists an Euler $(s,t)$-trail in $G$ if and only if $G$ is balanced.
\end{itemize}
\end{theorem}

\newcommand{\PsimpleSize}{2k}

Our argument will modify a given walk in a manner that might increase its length to keep certain conditions satisfied. To ensure that we never need to handle a walk that is too long, we utilize the following lemma.

\begin{lemma}\label{lem:shortenPath}
Let $(G,k,r,s,t)$ be a nice instance of \diPathRst. Let $P$ be an $r$-simple $k'$-path in $G$ for some integer $k'\geq 2k$. Then, $G$ has an $r$-simple $k''$-path $Q$, for some integer $k''\geq k$, such that $Q_{\mathrm{simple}}$ is a subgraph of $P_{\mathrm{simple}}$ that is not equal to $P_{\mathrm{simple}}$.
\end{lemma}

\begin{proof}
First, observe that since $G$ has no path of size at least $2k/r$, it holds that $P_{\mathrm{simple}}$ contains at least one cycle. We choose such a cycle $C$ arbitrarily.
In what follows, we use the cycle $C$ to modify the walk $P$ in order to obtain a walk $Q$ that has the desired property. To this end,  let $\Delta$ be the {\em minimum} number of times an arc of $C$ occurs in $P$. Let $H$ be the directed multigraph obtained from $P_{\mathrm{multi}}$ by removing first $\Delta$ copies of every arc in $C$ and then isolated vertices, if any. In addition, let ${\cal Q}$ be the set of weakly connected components of $H.$
Let $u$ and $v$ denote the first and last (not necessarily distinct) vertices visited by $P$. We consider two subcases depending on $|{\cal Q}|$. 
\begin{figure}[ht]
 \subfigure[$P_{\rm simple}$]
 {
 \begin{minipage}{\textwidth}\centering
 \tikzstyle{every node}=[draw,circle,fill=black,minimum size=1pt,inner sep=1.5pt]%
  \begin{tikzpicture}[>=stealth,->,thin]
   \node[label=above:$x$] at (2,0)  (x) {};
   \node[label=below:$u$] at (0,-2)  (u) {};
   \node[label=below:$y$] at (2,-2)  (y) {};
   \node[label=below:$v$] at (4,-2)  (v) {};
   \draw (u) -- (y);
   \draw (y) -- (v);
   \draw (x) -- (u);
   \draw (y) -- (x); 
  \end{tikzpicture}
  \end{minipage}
 }
 
 \vspace*{\baselineskip}
 
 \subfigure[$P_{\rm multi}$]
 {
 \begin{minipage}{.475\textwidth}\centering
 \tikzstyle{every node}=[draw,circle,fill=black,minimum size=1pt,inner sep=1.5pt]%
  \begin{tikzpicture}[>=stealth,->,thin]
   \node[label=above:$x$] at (2,0)  (x) {};
   \node[label=below:$u$] at (0,-2)  (u) {};
   \node[label=below:$y$] at (2,-2)  (y) {};
   \node[label=below:$v$] at (4,-2)  (v) {};
   \draw (u) -- (y);
   \draw (u) edge [bend left] (y);
    \draw (u) edge [bend right] (y);
   \draw (y) -- (v);
   \draw (x) -- (u);
   \draw (x) edge [bend left] (u); 
   \draw (y) -- (x); 
   \draw (y) edge [bend right] (x); 
  \end{tikzpicture}
  \end{minipage}

 }
 \hfill
 \subfigure[$H$]
 {
  \begin{minipage}{.475\columnwidth}\centering
   \tikzstyle{every node}=[draw,circle,fill=black,minimum size=1pt,inner sep=1.5pt]%
  \begin{tikzpicture}[>=stealth,->,thin]
     \node[label=below:$u$] at (0,-2)  (u) {};
   \node[label=below:$y$] at (2,-2)  (y) {};
   \node[label=below:$v$] at (4,-2)  (v) {};
   \draw (u) -- (y);
   \draw (y) -- (v);
  \end{tikzpicture}
  \end{minipage}
 }
 \caption{Illustrations for a 3-simple path $P=uyxuyxuyv;$ $C=uyxu.$}
 \label{Case1fig}
\end{figure}

\begin{figure}[ht]
 \subfigure[$P_{\rm multi}$]
 {
 \begin{minipage}{\textwidth}\centering
 \tikzstyle{every node}=[draw,circle,fill=black,minimum size=1pt,inner sep=1.5pt]%
  \begin{tikzpicture}[>=stealth,->,thin]
   \node[label=above:$x$] at (2,0)  (x) {};
   \node[label=below:$u$] at (0,-2)  (u) {};
   \node[label=below:$y$] at (2,-2)  (y) {};
   \node[label=below:$v$] at (4,-2)  (v) {};
    \node[label=above:$z$] at (0,0)  (z) {};
  \draw (u) -- (y);
   \draw (u) edge [bend left] (y);
    \draw (u) edge [bend right] (y);
   \draw (y) -- (v);
   \draw (x) -- (u);
   \draw (x) edge [bend left] (u); 
   \draw (y) -- (x); 
   \draw (y) edge [bend right] (x); 
     \draw (x) -- (z); 
   \draw (z) edge [bend left] (x);  
  \end{tikzpicture}
  \end{minipage}
 }
 
 \vspace*{\baselineskip}
 
 \subfigure[$\cal Q$]
 {
 \begin{minipage}{.475\textwidth}\centering
 \tikzstyle{every node}=[draw,circle,fill=black,minimum size=1pt,inner sep=1.5pt]%
  \begin{tikzpicture}[>=stealth,->,thin]
    \node[label=above:$x$] at (2,0)  (x) {};
     \node[label=below:$u$] at (0,-2)  (u) {};
   \node[label=below:$y$] at (2,-2)  (y) {};
   \node[label=below:$v$] at (4,-2)  (v) {};
    \node[label=above:$z$] at (0,0)  (z) {};
   \draw (u) -- (y);
   \draw (y) -- (v);
     \draw (x) -- (z); 
   \draw (z) edge [bend left] (x);  
  \end{tikzpicture}
  \end{minipage}

 }
 \hfill
 \subfigure[$H^{\star}$]
 {
  \begin{minipage}{.475\columnwidth}\centering
   \tikzstyle{every node}=[draw,circle,fill=black,minimum size=1pt,inner sep=1.5pt]%
  \begin{tikzpicture}[>=stealth,->,thin]
 \node[label=above:$x$] at (2,0)  (x) {};
   \node[label=below:$u$] at (0,-2)  (u) {};
   \node[label=below:$y$] at (2,-2)  (y) {};
   \node[label=below:$v$] at (4,-2)  (v) {};
   \draw (u) -- (y);
   \draw (u) edge [bend left] (y);
    \draw (u) edge [bend right] (y);
   \draw (y) -- (v);
   \draw (x) -- (u);
   \draw (x) edge [bend left] (u); 
   \draw (y) -- (x); 
   \draw (y) edge [bend right] (x); 
  \end{tikzpicture}
  \end{minipage}
 }
 \caption{Illustrations for a 3-simple path $P=uyxzxuyxuyv;$ $C=uyxu.$}
 \label{Case2fig}
\end{figure}

	\begin{enumerate}
	\item First, suppose that $|{\cal Q}|=1$  (e.g. see Fig. \ref{Case1fig}). Then, $H$ is weakly connected. Since $P_{\mathrm{multi}}$ has a $(u,v)$-path that visits every arc (that is the path $P$), by Theorem \ref{prop:euler}, 
$P_{\mathrm{multi}}$ is balanced if $u=v$ and 	$(u,v)$-almost balanced, otherwise.
	By the definition of $H$, every vertex in $V(G)$ has either both its out-degree and in-degree in $H$ equal to those in $P_{\mathrm{multi}}$ or both its out-degree and in-degree in $H$ smaller by $\Delta$ compared to those in $P_{\mathrm{multi}}$. Hence, $H$ is balanced if $u=v$ and $(u,v)$-almost balanced, otherwise.
	Thus, by Theorem \ref{prop:euler}, $H$ has an Euler trail $Q$. Moreover, since $(G,k,r,s,t)$ is nice, $|A(C)|<k/r$ and therefore $|A(Q)|>|A(P)|-\Delta(k/r) \geq k$. Lastly, since the out- and in-degrees of at least one vertex of $C$ was reduced from $\Delta$ in $P_{\mathrm{multi}}$ to $0$ in $H$, it holds that $|V(P_{\mathrm{simple}})|>|V(Q_{\mathrm{simple}})|$. Thus, $Q$ is an $r$-simple $k''$-path, for some integer $k''\geq k$, such that $Q_{\mathrm{simple}}$ is a subgraph of $P_{\mathrm{simple}}$ that is not equal to $P_{\mathrm{simple}}$.
	
	\item\label{item:lemshortenPathItem2} Now, suppose that $|{\cal Q}|\geq 2$  (e.g. see Fig. \ref{Case2fig}). Let $Q_{\mathrm{min}}$ be a component in ${\cal Q}$ that has minimum number of arcs. Then, $|A(Q_{\mathrm{min}})|<|A(P)|/2$. Let $H^\star$ be the directed multigraph obtained from $P_{\mathrm{multi}}$ by first removing all the arcs in $Q_{\mathrm{min}}$ and then isolated vertices, if any. Since $P_{\mathrm{multi}}$ has a $(u,v)$-path that visits every arc (that is the path $P$), by Theorem \ref{prop:euler},  $P_{\mathrm{multi}}$ is balanced if $u=v$ and $(u,v)$-almost balanced, otherwise. As in the previous case, every vertex in $V(G)$ has either both its out-degree and in-degree in $H$ equal to those in $P_{\mathrm{multi}}$ or both its out-degree and in-degree in $H$ smaller by $\Delta$ compared to those in $P_{\mathrm{multi}}$. If $u\neq v$, this means that either both $u,v\in V(H^\star)$ or both $u,v\notin V(H^\star).$ 
	Indeed, as in any directed multigraph, in $Q_{\mathrm{min}}$, the sum of in-degrees of all vertices equals the sum of out-degrees of all vertices. Thus, if $u\in V(Q_{\mathrm{min}})\setminus V(H^\star)$ then $v\in V(Q_{\mathrm{min}})\setminus V(H^\star)$ as well.
	However, this means that $H^\star$ is balanced if $u=v$ and $(u,v)$-almost balanced, otherwise. Moreover,  $H^\star$ is weakly connected (because $H^\star$ consists of a collection of components in ${\cal Q}$ together with the arcs in $C$ that connect their underlying undirected graphs). 
Thus, by Theorem  \ref{prop:euler}, $H^\star$ has an Euler trail $Q$. Moreover, $|A(Q)|>\frac{1}{2}|A(P)|\geq k$. In addition, $|A(P_{\mathrm{simple}})|>|A(Q_{\mathrm{simple}})|$ since by definition of $\cal Q$, $Q_{\mathrm{min}}$ has arcs and none of them can be in $Q.$
 Thus, $Q$ is an $r$-simple $k''$-path, for some integer $k''\geq k$, such that $Q_{\mathrm{simple}}$ is a subgraph of $P_{\mathrm{simple}}$ that is not equal to $P_{\mathrm{simple}}$.
\end{enumerate}
In both cases, we constructed a walk $Q$ with the desired property, hence the proof is complete.
\end{proof}

A repeated application of Lemma \ref{lem:shortenPath} brings us the following corollary.

\begin{corollary}\label{cor:shortenPath}
Let $(G,k,r,s,t)$ be a nice instance of \diPathRst. Let $P$ be an $r$-simple $k'$-path in $G$ for some integer $k'\geq 2k$. Then, $G$ has an $r$-simple $k''$-path $Q$, for some integer $k''\in\{k,k+1,\ldots,2k\}$, such that $Q_{\mathrm{simple}}$ is a subgraph of $P_{\mathrm{simple}}$ that is not equal to $P_{\mathrm{simple}}$.
\end{corollary}

In fact, $k''\in\{k,k+1,\ldots,2k\}$ above can be clearly replaced by $k''\in\{k,k+1,\ldots,2k-1\},$ but for simplicity in what follows we will use the former rather than the latter.

We now prove that if $(G,k,r,s,t)$ is a positive instance of \diPathRst, then $G$ has an $r$-simple $k'$-path for some $k'\in\{k,k+1,\ldots,\PsimpleSize\}$ such that $V(P,r)$ and $P_{\mathrm{simple}}^{-r}$ satisfy three properties regarding their structure. In addition, we prove that if $(G,k,r,s,t)$ is a negative instance of \diPathRst, then at least one $r$-simple $(s,t)$-path $P$ in $G$ of maximum size satisfies these three properties as well. 

\begin{lemma}\label{lem:Psimple}
Let $(G,k,r,s,t)$ be a nice instance of \diPathRst. If $(G,k,r,s,t)$ is a positive instance, then $G$ has an $r$-simple $k'$-path $P$ for some $k'\in\{k,k+1,\ldots,\PsimpleSize\}$ that satisfies the following three properties.
\begin{enumerate}
\item\label{item:Psimple1} $P_{\mathrm{simple}}^{-r}$ is an acyclic digraph.
\item\label{item:Psimple2} For any (not necessarily distinct) $u,v\in V(P)$, $P_{\mathrm{simple}}^{u,v,-r}$ has at most one $(u,v)$-path.\footnote{Recall that if $u=v$, by a $(u,v)$-path we mean a $(u,u)$-cycle.}
\item\label{item:Psimple3} $|V(P,r)|\leq \PsimpleSize/r + 2$.
\end{enumerate}

Otherwise (if $(G,k,r,s,t)$ is a negative instance), $G$ has an $r$-simple $(s,t)$-path $P$ of maximum size that satisfies these three properties.
\end{lemma}

\begin{proof}
We define a collection of walks $\cal P$ as follows: if $(G,k,r,s,t)$ is a positive instance, then ${\cal P}$ is the set of all $r$-simple $k'$-paths in $G$ where $k'\in\{k,k+1,\ldots,\PsimpleSize\}$; otherwise, ${\cal P}$ is the set of all $r$-simple $(s,t)$-paths in $G$ of maximum size. In both cases, ${\cal P}\neq\emptyset$. For any $\ell\in\mathbb{N}$ and $r$-simple path $P$ of size $\ell$, $V(P,r)\setminus\{s,t\}$ can contain at most $\lfloor\ell/r\rfloor$ vertices.
Therefore, in the first case, since $k'\leq \PsimpleSize$, every walk in $\cal P$ satisfies Property \ref{item:Psimple3}. In the second case, every walk $P\in{\cal P}$ contains less than $k$ vertices (since the instance is negative), therefore $P$ satisfies Property \ref{item:Psimple3}. Thus, it suffices to show that there exists a walk in $\cal P$ that satisfies Properties \ref{item:Psimple1} and \ref{item:Psimple2}.

Let ${\cal P}'$ be the set of walks $P\in {\cal P}$ with minimum number of arcs in $P_{\mathrm{simple}}$. Moreover, let ${\cal P}''$ be the set of walks $P\in {\cal P}'$ that maximize $|V(P,r)|$.

We claim that every walk in ${\cal P}''$ satisfies Properties \ref{item:Psimple1} and \ref{item:Psimple2}. For this purpose, we consider an arbitrary walk $P\in {\cal P}''$. Let $u$ and $v$ denote the first and last (not necessarily distinct) vertices visited by $P$. (If $(G,k,r,s,t)$ is a negative instance, then $u=s$ and $v=t$.) Suppose, by way of contradiction, that $P$ does not satisfy Property \ref{item:Psimple1}. Then, $P_{\mathrm{simple}}^{-r}$ has a directed cycle $C$. Let $\Delta$ be the {\em maximum} out-degree in $P_{\mathrm{multi}}$ of a vertex in $C$. Note that $\Delta<r$ because $V(C)\cap V(P,r)=\emptyset$. Let $H$ be the directed multigraph obtained from $P_{\mathrm{multi}}$ by adding $r-\Delta$ copies of every arc in $C$. Since $P_{\mathrm{multi}}$ has a $(u,v)$-path that visits every arc (that is the path $P$), by Theorem  \ref{prop:euler}, $P_{\mathrm{multi}}$ is balanced if $u=v$ and $(u,v)$-almost balanced, otherwise. By our construction of $H$, it has the same property. Indeed, every vertex in $V(G)$ has either both its out-degree and in-degree in $H$ equal to those in $P_{\mathrm{multi}}$ or both its out-degree and in-degree in $H$ larger by $r-\Delta$ compared to those in $P_{\mathrm{multi}}$. 
Thus, by Theorem \ref{prop:euler}, $H$ has an Euler trail $P'$ {\em with the same endpoints as $P$}. 
Let us consider two cases, depending on the size of $P'$.
\begin{enumerate}
\item First, suppose that $P'$ is of size at most $2k$. Then, $P'\in{\cal P}$, and since $P'_{\mathrm{simple}}=P_{\mathrm{simple}}$, it further holds that $P'\in{\cal P}'$. However, $|V(P',r)|>|V(P,r)|$ because at least one vertex of $C$ belongs to $V(P',r)$ but not to $V(P,r)$ and clearly $V(P,r)\subseteq V(P',r)$. Thus, we have a contradiction to the inclusion $P\in{\cal P}''$.
\item Second, suppose that $P'$ is of size larger than $2k$. We stress that in this case, $(G,k,r,s,t)$ is positive. By Corollary \ref{cor:shortenPath}, $G$ has an $r$-simple $k''$-path $Q$, for some integer $k''\in\{k,k+1,\ldots,2k\}$, such that $Q_{\mathrm{simple}}$ is a subgraph of $P'_{\mathrm{simple}}$ that is not equal to $P'_{\mathrm{simple}}$. Then, $Q\in{\cal P}$ because in this case, to be included in ${\cal P}$, a walk does not need to have the same start and end vertices as $P$. Since $P'_{\mathrm{simple}}=P_{\mathrm{simple}}$, we have that $|A(Q_{\mathrm{simple}})|<|A(P'_{\mathrm{simple}})|=|A(P_{\mathrm{simple}})|$, which is a contradiction to the inclusion $P\in{\cal P}''$.
\end{enumerate}

It remains to argue that $P$ satisfies Property \ref{item:Psimple2}. Suppose, by way of contradiction, that this claim is false. Then, for some vertices $x,y\in V(P)$, it holds that $P_{\mathrm{simple}}^{x,y,-r}$ has at least two pairwise {\em internally vertex disjoint} $(x,y)$-paths. Denote two such different vertex disjoint paths (chosen arbitrarily) by $P^{x\rightarrow y}_1$ and $P^{x\rightarrow y}_2$ such that $|A(P^{x\rightarrow y}_1)|\geq |A(P^{x\rightarrow y}_2)|$. Note that $V(P^{x\rightarrow y}_1)\setminus\{x,y\}=V(P^{x\rightarrow y}_1)\setminus V(P^{x\rightarrow y}_2)\neq\emptyset$ and $A(P^{x\rightarrow y}_2)\cap A(P^{x\rightarrow y}_1)=\emptyset$. (Note that $V(P^{x\rightarrow y}_2)\setminus V(P^{x\rightarrow y}_1)$ can be empty since $P^{x\rightarrow y}_2$ can consist of a single arc).
Let $\Delta_1$ denote the {\em maximum} out-degree in $P_{\mathrm{multi}}$ of a vertex in $V(P^{x\rightarrow y}_1)\setminus \{x,y\}$. In addition, let $\Delta_2$ denote the {\em minimum} number of times an arc of $A(P^{x\rightarrow y}_2)$ occurs in $P$. Now, denote $\Delta=\min\{r-\Delta_1,\Delta_2\}$. Let $H$ be the directed multigraph obtained from $P_{\mathrm{multi}}$ by adding $\Delta$ copies of every arc of $P^{x\rightarrow y}_1$, and removing $\Delta$ copies of every arc of $P^{x\rightarrow y}_2$  as well as isolated vertices. In addition, let ${\cal Q}$ be the set of weakly connected components of $H$. We consider two subcases depending on the size of $|{\cal Q}|$. 

	\begin{enumerate}
	\item First, suppose that $|{\cal Q}|=1$. Then, $H$ is weakly connected. Since $P_{\mathrm{multi}}$ has a $(u,v)$-path that visits every arc (that is the path $P$), by Theorem \ref{prop:euler}, 
	$P_{\mathrm{multi}}$ is balanced if $u=v$ and $(u,v)$-almost balanced, otherwise. By the definition of $H$, every vertex in $V(G)$ has {\em (i)} both its out-degree and in-degree in $H$ equal to those in $P_{\mathrm{multi}}$, or {\em (ii)} both its out-degree and in-degree in $H$ larger by $\Delta$ compared to those in $P_{\mathrm{multi}}$, or {\em (iii)} both its out-degree and in-degree in $H$ smaller by $\Delta$ compared to those in $P_{\mathrm{multi}}$. Thus, $H$ is balanced if $u=v$ and $(u,v)$-almost balanced, otherwise. Thus, by Theorem \ref{prop:euler}, $H$ has an Euler trail $P'$. Moreover, since $|A(P^{x\rightarrow y}_1)|\geq |A(P^{x\rightarrow y}_2)|$, we have that $|A(P')|\geq|A(P)|\geq k$. In addition, $P'_{\mathrm{simple}}$ is a subgraph of $P_{\mathrm{simple}}$. We consider three subcases depending on $\Delta$ and the size of $P'$.
		\begin{enumerate}
		\item Suppose that $\Delta=r-\Delta_1>\Delta_2$  and the size of $P'$ is at most $2k$. Then, $P'\in{\cal P}'$ and at least one vertex in $V(P^{x\rightarrow y}_1)\setminus V(P^{x\rightarrow y}_2)$ has out-degree $r$ in $P'_{\mathrm{multi}}$ but not in $P_{\mathrm{multi}}$, while clearly $V(P,r)\subseteq V(P',r)$. However, this is a contradiction to the inclusion $P\in{\cal P}''$.
		\item Suppose that $\Delta=\Delta_2$ and the size of $P'$ is at most $2k$. Then, $P'\in{\cal P}'$ but $P'_{\mathrm{simple}}$ is not equal to $P_{\mathrm{simple}}$ (at least one arc of $P^{x\rightarrow y}_2$ is present in $P_{\mathrm{simple}}$ but not in $P'_{\mathrm{simple}}$). However, this is a contradiction to the inclusion $P\in{\cal P}'$.
		\item Suppose that the size of $P'$ is larger than $2k$. Then, by Corollary \ref{cor:shortenPath}, $G$ has an $r$-simple $k''$-path~$Q$, for some integer $k''\in\{k,k+1,\ldots,2k\}$, such that $Q_{\mathrm{simple}}$ is a subgraph of $P'_{\mathrm{simple}}$ that is not equal to $P_{\mathrm{simple}}$. However, this is a contradiction to the inclusion $P\in{\cal P}'$.
		\end{enumerate}
	
	\item Now, suppose that $|{\cal Q}|\geq 2$. Then, exactly like in Case \ref{item:lemshortenPathItem2} in the proof of Lemma \ref{lem:shortenPath}, we derive that $G$ has an $r$-simple $k''$-path $P'$, for some integer $k''\geq k$, such that $P'_{\mathrm{simple}}$ is a subgraph of $P_{\mathrm{simple}}$ that is not equal to $P_{\mathrm{simple}}$. By Corollary \ref{cor:shortenPath}, this means that $G$ has an $r$-simple $\widehat{k}$-path~$Q$, for some integer $\widehat{k}\in\{k,k+1,\ldots,2k\}$, such that $Q_{\mathrm{simple}}$ is a subgraph of $P'_{\mathrm{simple}}$ that is not equal to $P'_{\mathrm{simple}}$. However, this is a contradiction to the inclusion $P\in{\cal P}'$.
\end{enumerate}
Since both cases led to a contradiction, the proof is complete.
\end{proof}

%
%

Having Lemma \ref{lem:Psimple} at hand, we can already bound the number of distinct arcs. In Section \ref{sec:numDistinctArcTight}, we present additional arguments on top of Lemma \ref{lem:Psimple} to make the bound tight.

\begin{lemma}\label{lem:distinct}
Let $(G,k,r,s,t)$ be a nice instance of \diPathRst. If $(G,k,r,s,t)$ is positive, then $G$ has an $r$-simple $k$-path with fewer than $10(k/r)^3$ distinct arcs. Otherwise, $G$ has an $r$-simple $(s,t)$-path of maximum size with fewer than $20(k/r)^3$ distinct arcs.
\end{lemma}

\begin{proof}
Let $P$ be a walk with the properties guaranteed by Lemma \ref{lem:Psimple}. Let ${\cal W}$ be the multiset that contains every subwalk of $P$ on at least two vertices, with both endpoints in $V(P,r)$ and with no internal vertex from $V(P,r)$. (The walks in $\cal W$ can be closed walks.) Let $\ell=|V(P,r)|.$  

By Property \ref{item:Psimple1}, every walk in ${\cal W}$ has no vertex that occurs more than once except for its endpoints which may be equal, and hence all walks in $\cal W$ are paths and cycles. Moreover, Property \ref{item:Psimple2} implies that the number of distinct walks in ${\cal W}$ is at most $\ell^2$. By Property \ref{item:Psimple3}, $\ell\leq \PsimpleSize/r + 2$. Therefore, the number of distinct walks in $\cal W$ is at most $(\PsimpleSize/r + 2)^2$. Since the instance $(G,k,r,s,t)$ is nice, $G$ has neither a path with at least $2k/r$ vertices nor a cycle of length at least $k/r$. This means that every walk in ${\cal W}$ has at most $2k/r-1$ arc visits. Thus, we conclude that the number of distinct arcs in $P$ is at most $(\PsimpleSize/r + 2)^2\cdot (2k/r-1) < 20(k/r)^3$. In case $P$ is of size larger than $k$ (then, $(G,k,r,s,t)$ is positive), we can choose any subwalk of $P$ of size $k$ to obtain an $r$-simple $k$-path with fewer than $20(k/r)^3$ distinct arcs.
\end{proof}

\subsection{Tightening the Bound on the Number of Distinct Arcs}\label{sec:numDistinctArcTight}

We proceed to prove that the upper bound $20(k/r)^3$ in Lemma \ref{lem:distinct} can be reduced to a bound whose dependence on $(k/r)$ is quadratic rather than cubic. Afterwards, we show that this upper bound is tight. To obtain the improved upper bound, we need the following definition.

\begin{definition}\label{def:projection}
Let $P$ be an $r$-simple $k$-path in a digraph $G$, and let $X\subseteq V(P)$. The \emph{projection of $P$ onto $X$} is a directed multigraph $H_X=(X, A_X)$ defined as follows. Traverse $P$ in order, from its first to last vertex, and add one arc $(u,v)$ to $A_X$ for every subwalk of $P$ between distinct vertices $u, v\in X$ whose internal vertices (if any) are not in $X$. 
\end{definition}

We show that for some $X\supseteq V(P,r)$ of size at most $|V(P,r)|+2$, we may assume that $H_X$ contains at most $3|X|$ distinct arcs (that is, omitting arc copies). To facilitate the proof, let us make another definition.

\begin{definition}
Let $G$ be a digraph and $S \subseteq V(G)$ a set of vertices. The \emph{split of $G$ on $S$} is the digraph defined by replacing every vertex $v \in S$ by two vertices: $v^h$, retaining all in-arcs incident with $v$, and $v^t$, retaining all out-arcs incident with $v$. 
\end{definition}

We now show the result. 

\begin{lemma}\label{lem:improvedBound}
Let $(G,k,r,s,t)$ be a nice instance of \diPathRst. There is an $r$-simple path $P$ such that the following hold.  If $(G,k,r,s,t)$ is a positive instance, then $P$ is a $k'$-path $P$ for some $k'\in\{k,k+1,\ldots,\PsimpleSize\}$, otherwise $P$ is an $(s,t)$-path of maximum size.  Furthermore, $P$ satisfies the three properties in Lemma \ref{lem:Psimple}, and for some set $X$ with $V(P,r)\subseteq X\subseteq V(G)$, with $|X| \leq |V(P,r)|+2$, the projection $H_X$ of $P$ onto $X$ contains a set of fewer than $3|X|$ distinct arcs whose corresponding walks cover all distinct arcs used by $P$. 
%
\end{lemma}

\begin{proof}
Recall that in the proof of Lemma~\ref{lem:Psimple}, we define a collection of walks $\cal P$ as follows: if $(G,k,r,s,t)$ is a positive instance, then ${\cal P}$ is the set of all $r$-simple $k'$-paths in $G$ where $k'\in\{k,k+1,\ldots,\PsimpleSize\}$; otherwise, ${\cal P}$ is the set of all $r$-simple $(s,t)$-paths in $G$ of maximum size. Moreover, ${\cal P}'$ is the set of walks $P\in {\cal P}$ with minimum number of arcs in $P_{\mathrm{simple}}$, and ${\cal P}''$ is the set of walks $P\in {\cal P}'$ that maximize $|V(P,r)|$. We have shown that there exists a path $P\in{\cal P}''$ which satisfies the three properties in Lemma~\ref{lem:Psimple}. Consider such a path $P$,  and let $s'$ and $t'$ denote the start and end vertices of $P$. (In case  $(G,k,r,s,t)$ is a negative instance, $s'=s$ and $t'=t$.)


Let $X=V(P,r)\cup\{s',t'\}$. Let $H_X$ be the projection of $P$ onto $X$, and let $F$ be the set of arcs of $H_X$ without multiplicity. 
Let $F' \subseteq F$ be a minimal set of arcs whose corresponding walks cover all arcs of $P$.
We will show that if $|F'| \geq 3|X|$, then 
there exists a different solution $P'$ which meets all the above conditions and is preferable to $P$ by our criteria, thereby deriving a contradiction. 
Thus, assume $|F'|\geq 3|X|$, and decompose $F'=F_1 \cup F_2$ where $F_1$ is a spanning tree for the underlying undirected graph of $H_X$ and $F_2 = F' \setminus F_1$. 

Let $G_0$ be the split of $H_X-F_1$ on $X$, where we remove all copies of arcs in $F_1$.  This is a directed multigraph with $2|X|$ vertices and $|F_2| > 2|X|$ arcs, each of which represents a walk in $G$.  Next, consider ``unrolling'' each of the arcs in $F_2$ in some arbitrary order, replacing each arc by all the arcs and vertices of the corresponding walk. 
This adds, for each expanded arc, some $\ell \geq 1$ additional arc copies (while removing the represented arc) and at most $\ell-1$ additional vertices (fewer if several arcs represent walks on a shared vertex set). Furthermore, let $\ell' \leq \ell$ be the number of arcs thus created for which there did not exist a copy already. Then a new vertex can be created only if $\ell'>1$ and the number of created vertices is at most $\ell'-1$. Note that by the minimality of $F'$ we have $\ell'>0$ for every arc we unroll.
Let $G'$ be the resulting directed multigraph.  We show that the underlying undirected graph of $G'$ contains a cycle.
Clearly this holds for $G_0$, since $|E(G_0)|>|V(G_0)|$; we claim that this invariant holds throughout the process of unrolling.
Indeed, every time an arc of $F_2$ is unrolled, the number of new distinct arcs created (minus that removed) is at least as large as the number of new vertices created. Thus $G'$ has at least as many distinct arcs as vertices and its underlying undirected graph contains a cycle.
Let $C$ be the arc set of such a cycle and let $H=P_{\mathsf{multi}}$.
  
We now derive a modification of $H$ from $C$. Define a \emph{sign} for every arc in $C$ by traversing $C$ in an arbitrary direction and labelling every arc traversed in the forward direction as \emph{positive} and every arc traversed in the backwards direction as \emph{negative}. Let $d=1$ if $C$ contains at least as many positive as negative arcs, and otherwise $d=-1$.
We claim that 
modifying the multiplicity in $H$ of every positive arc of $C$ by $+d$ and the multiplicity of every negative arc by $-d$, yields a directed multigraph with an Euler $(s',t')$-trail and where every vertex has out- and in-degree at most $r$. For this, we first note that for every $v \in X$ that occurs on $C$, the modifications of in-arcs and the modifications of out-arcs both sum to zero (since the traversal was derived over a cycle in $G'$, where $v$ was split). Every other vertex either has its out- and in-degrees unmodified, like $v$, or has in- and out-degrees both modified by the same amount (either $+1$ or $-1$). 
Thus, the modification keeps the balances between in- and out-degrees unchanged, and produces a graph where every vertex has in- and out-degree at most $r$. Second, we show that all arcs of the modified graph are in one connected component. 
Assume the contrary, i.e., that due to some arcs having their multiplicities reduced to 0, the resulting graph has at least two connected components containing at least one arc. However, since all arcs represented in $F_1$ are untouched, the resulting graph has a large connected component that visits all vertices of $X$, thus any further component containing at least one arc must be entirely contained in $P_{\mathrm{simple}}^{-r}$. However, all vertices except for possibly $s'$ and $t'$ have in-degree equal to out-degree, and since $s', t' \in X$ it would have to follow that such a ``lost component'' contains a directed cycle outside of $X$. But by Property \ref{item:Psimple1}, $P_{\mathrm{simple}}^{-r}$ is acyclic.  We conclude that all arcs of the modified graph are contained in one connected component.  Hence, this component has an $(s',t')$-Euler trail, which forms an $r$-simple $(s',t')$-path $P'$.  Note furthermore, by the choice of $d$, that $P'$ is at least as long as $P$.  Thus we finally conclude that the modified graph has no isolated vertices and no arc whose multiplicity is reduced to 0, since this would contradict the choice of $P$. 

Moreover, the size of $P'$ cannot exceed $2k$, since then by Corollary \ref{cor:shortenPath} we derive a solution $Q$ (that is, $Q\in{\cal P}$) such that $|A(Q_{\mathrm{simple}})|<|A(P_{\mathrm{simple}})|$, which contradicts the inclusion $P\in{\cal P}'$.
  
Now consider performing this modification several times in the same direction $d$.
There are only two bounding events for this: Either the multiplicity of some arc reduces to 0, or some vertex not in $X$ reaches out-degree $r$. However, both events would contradict our priorities in choosing $P$ (that is, the inclusion $P\in{\cal P}'$ in the first event, and the inclusion $P\in{\cal P}''$ in the second event). This is a contradiction, showing that the cycle $C$ cannot exist, and we conclude that $|F_2|<2|X|$, hence $|F'|<3|X|$ and $F'$ is the required set.
\end{proof}

Let us now conclude our improved bound.

\begin{lemma}\label{lem:distinctImproved}
Let $(G,k,r,s,t)$ be a nice instance of \diPathRst. If $(G,k,r,s,t)$ is positive, then $G$ has an $r$-simple $k$-path with fewer than $30(k/r)^2$ distinct arcs. Otherwise, $G$ has an $r$-simple $(s,t)$-path of maximum size with fewer than $30(k/r)^2$ distinct arcs.
\end{lemma}

\begin{proof}
Let $P$ and $X$ be a walk and a set with the properties guaranteed by Lemma \ref{lem:improvedBound}. Let ${\cal W}$ be the multiset that contains every subwalk of $P$ on at least two vertices, with both endpoints in $X$ and with no internal vertex from $X$. (The walks in $\cal W$ can be closed walks.) 
By Lemma \ref{lem:improvedBound}, for the purpose of counting distinct arcs used in $P$, 
it suffices to consider a set of at most $3|X|$ walks of ${\cal W}$ with distinct endpoints, and by Properties \ref{item:Psimple1} and \ref{item:Psimple2} in Lemma \ref{lem:Psimple}, there do not exist two distinct walks $\cal W$ that have the same start and end vertices.
Moreover, by Properties \ref{item:Psimple2} and \ref{item:Psimple3} in Lemma \ref{lem:Psimple}, $\cal W$ has at most $|X|$ walks with equal endpoints. Thus, the number of distinct walks we need to consider is 
$\ell\leq 4|X|\leq 4(|V(P,r)|+2)\leq 4(2k/r + 4)$.

By Property \ref{item:Psimple1}, every walk in ${\cal W}$ has no vertex that occurs more than once except for its endpoints which may be equal. Since the instance $(G,k,r,s,t)$ is nice, $G$ has neither a path with at least $2k/r$ vertices nor a cycle of length at least $k/r$. This means that every walk in ${\cal W}$ has at most $2k/r-1$ arc visits. Thus, we conclude that the number of distinct arcs in $P$ is upper bounded by $4(2k/r + 4)\cdot (2k/r-1) < 30(k/r)^2$ (where $30$ is simply a conveniently chosen sufficiently large constant). In case $P$ is of size larger than $k$ (then, $(G,k,r,s,t)$ is positive), we can choose any subwalk of $P$ of size $k$ to obtain an $r$-simple $k$-path with fewer than $30(k/r)^2$ distinct arcs.
\end{proof}

\paragraph{The Tightness of the Bound.} We show that without devising new reduction rules in addition to those given in Section \ref{sec:directedPathReduction}, the bound on the number of distinct arcs in a solution must depend quadratically on $(k/r)$. More precisely, we prove the following result.

\begin{lemma}
For any integer $r\in\mathbb{N}_{\ge 2}$, there exists a nice positive instance $(G,k,r,s,t)$ of \diPathRst\ with $k/r=\Theta(r)$ such that every $r$-simple $k$-path in $G$ has $\Omega((k/r)^2)$ distinct arcs.
\end{lemma}

\begin{proof}
Let $r\in\mathbb{N}_{\ge 2}$. Consider a digraph $G$ with a vertex $u$ and $r$ cycles $C_i=u v_1^i\dots v^i_ru$ ($i=1,2,\dots ,r$)  sharing pairwise only vertex $u$. For every $i=1,2,\dots ,r$ add to $G$ a 2-cycle $w^iv^i_1w^i$, where $w^1,\dots ,w^{r}$ are new vertices in $G$. Let $P$ be an $r$-simple path of $G$ of maximum size. Observe that $P$ cannot traverse any $C_i$ twice (i.e., it cannot visit vertices of any $C_i$ twice apart from $v^i_1$) along with visiting $w^i$ $r-1$ times since $P$ will have more vertex visits if it traverses two cycles $C_i$ and $C_j$ instead along with visiting $w^i$ and $w^j$ $r-1$ times each. Thus, $P$ visits $r$ times $u$ and each $v^i_1$. It visits $r-1$ times each $w^i$ and only once every vertex of $C_i$ apart from $u$ and $v^i_1$ for all $i=1,2,\dots ,r.$ Hence,  $k=r(r+1)+r(r-1)+r(r-1)=\Theta(r^2)$ and $k/r = 3r-1=\Theta(r)$. Note that $P$ is an open walk which visits every arc of $G$ but one. Thus, $P$ has $|A(G)|-1=r(r+1)+2r=\Theta((k/r)^2)$ distinct arcs.
Finally, $G$ is nice since it is strongly connected, the longest cycle has $r+1<k/r$ vertices, and the longest path (which starts at some $w^i$ and ends at some $v_r^j$, $i\neq j$)
has $2r+2 < 2k/r$ vertices.
\end{proof}

\subsection{Color Coding} \label{sec:dir:colorcoding}

Knowing that it suffices for us to deal only with walks having a small number of distinct arcs (in light of Lemma \ref{lem:distinctImproved}) and hence a small number of distinct vertices, we utilize the method of color coding by Alon et al.~\cite{DBLP:journals/jacm/AlonYZ95}. Concretely, by Lemma~\ref{lem:distinctImproved} it suffices to consider solutions with fewer than $30(k/r)^2$ vertices, hence at most $30(k/r)^2$ arcs.
For the sake of brevity, we define the following problem. Here, $\mathsf{b}(k/r)=30(k/r)^2+1$ and  a walk is called {\em colorful} if every two distinct vertices visited by the walk have distinct colors. 

\begin{defproblem}  
{\diColPathRst}
{Integers $k,r\in\mathbb{N}$, a $\mathsf{b}(k/r)$-colored digraph $G$, and {\em distinct} vertices $s,t\in V(G).$}
{ Output an integer~$i$ such that {\em (i)} $G$ has an $r$-simple $(s,t)$-path of size $i$, and {\em (ii)} for any $j>i$, $G$ does not have a colorful $r$-simple $(s,t)$-path of size $j.$}
\end{defproblem}\\

Before we proceed to handle this variant, let us make an important remark. At first glance, it might seem that the objective in the problem definition above could be replaced by the following simpler condition: output the largest integer $i$ such that $G$ has a colorful $r$-simple $(s,t)$-path of size~$i$. However, we are not able to resolve this problem, and given the approach of guessing topologies that we define later, having the stronger condition will entail the resolution of a problem as hard as {\sc Multicolored Clique} (defined in Section \ref{sec:W1}) and hence lead to a dead-end.

Now, we show that we can focus on our colored variant \diColPathRst.

\begin{lemma}\label{lem:color}
Suppose that \diColPathRst\ can be solved in time $f(k/r)\cdot (n+\log k)^{\OO(1)}$ and polynomial space. Then, \diPathRst\ on strongly connected digraphs can be solved in time $2^{\OO((\frac{k}{r})^2)}\cdot f(k/r)\cdot (n+\log k)^{\OO(1)}$ and polynomial space.
\end{lemma}

\begin{proof}
Let $\cal A$ be an algorithm that solves \diColPathRst\ in time $f(k/r)\cdot (n+\log k)^{\OO(1)}$ and polynomial space. In what follows, we describe how to solve \diPathRst.
To this end, let $(G,k,r,s,t)$ be an instance of \diPathRst. By Lemma \ref{lem:longCycPath} and Theorem \ref{prop:longPath}, we may assume that $(G,k,r,s,t)$ is nice, which can be verified in time $2^{\OO(k/r)}n^{\OO(1)}$. Without loss of generality, denote $V(G)=\{1,2,\ldots,n\}$.  To handle the case that a solution in $G$ starts and ends at the same vertex, create a new graph $G'$ from $G$ by adding, for each vertex $v\in V(G)$, a new vertex $v'$ and the arc $(v,v')$.  
For each pair of (not necessarily distinct) vertices $u,v\in V(G)$, initialize $k_{uv}:=0$. By Theorem \ref{prop:hashfam}, we can enumerate the functions of some $(n,\mathsf{b}(k/r)-1)$-perfect hash family $\cal F$ of size $e^{\mathsf{b}(k/r)+o(\mathsf{b}(k/r))}\log n$ with polynomial delay. For each function $f\in{\cal F}$ and for every pair of (not necessarily distinct) vertices $u,v\in V(G)$, call $\cal A$ with $(G',k+1,r,u,v')$ as input where the color of  $w$ is $f(w)$ for any $w\in V(G)$ and $\mathsf{b}(k/r)$ for any $w\in V(G')\setminus V(G)$. Let $t$ be the output of this call. If it is larger than $k_{uv}+1$, then update $k_{uv}:=t-1$. After all calls were performed, compute $k^\star=\max_{u,v\in V(G)}k_{uv}$. If $k^\star\geq k$, then we determine that $G$ has an $r$-simple $k$-path; otherwise, we output $k_{st}$.

Clearly, the algorithm runs in time $2^{\OO((\frac{k}{r})^2)}\cdot f(k/r)\cdot (n+\log k)^{\OO(1)}$ and uses polynomial space. Next, we prove that the algorithm is correct, that is, that it indeed solves \diPathRst. On the one hand, we have the two cases as follows.
\begin{itemize}
\item First, suppose that $(G,k,r,s,t)$ is positive. By Lemma \ref{lem:distinctImproved}, $G$ has an $r$-simple $k$-path $P$ with fewer than $\mathsf{b}(k/r)-1$ distinct arcs. By Definition \ref{def:hash}, there exists $f\in{\cal F}$ that is injective on the set of distinct vertices of $P$. Let $u$ and $v$ be the start and end vertices of $P$, respectively. Then, in the iteration where $f$ is considered with $u$ and $v$, $\cal A$ must output an integer $t\geq k+1$. Hence, $k^\star\geq k$.
\item Second, suppose that $(G,k,r,s,t)$ is negative. By Lemma \ref{lem:distinctImproved}, $G$ has an $r$-simple $(s,t)$-path $P$ of maximum size with fewer than $\mathsf{b}(k/r)-1$ distinct arcs. By Definition \ref{def:hash}, there exists $f\in{\cal F}$ that is injective on the set of distinct vertices of $P$. Then, in the iteration where $f$ is considered with $s$ and $t$, $\cal A$ must output an integer $t$ that is at least as large as the size of $P$ plus $1$.
\end{itemize}
On the other hand, it is immediate that for any $u,v\in V(G)$, the final value $k_{uv}$ is at most the maximum size of an $r$-simple $(u,v)$-path in $G$. Thus, by the specification of the algorithm, we conclude that it is correct.
\end{proof}

\subsection{Guessing the Topology of a Solution}

We proceed to define the notion of a topology, which we need in order to sufficiently restrict our search space. Note that in the definition, the multiplicity of every arc is at most $1$, but we can have mutually-opposite arcs, i.e., arcs of the type $(x,y)$ and $(y,x)$. 

\begin{definition}\label{def:topology}
Let $\ell\in\mathbb{N}$. Then, an $\ell$-topology is an $\ell$-colored digraph with at most $\ell$ arcs such that each of its vertices has a distinct color, and whose underlying undirected graph is connected. Let ${\cal T}_\ell$ denote the set of all $\ell$-topologies.
\end{definition}

We first argue that there are not many topologies.

\begin{lemma}\label{lem:fewTopologies}
Let $\ell\in\mathbb{N}$. Then, $|{\cal T}_\ell|=2^{\OO(\ell\log\ell)}$.
\end{lemma}

\begin{proof}
A digraph $D$ on $n$ vertices is called {\em labelled} if the vertices of $D$ are $\{1,2,\dots ,n\}$ (called {\em labels}). Two labelled digraphs $D$ and $H$ are considered {\em equal} if they have the same number $n$ of vertices and for every $i\neq j\in \{1,2,\dots ,n\}$, we have $(i,j)\in A(D)$ if and only if $(i,j)\in A(H).$ Otherwise, $D$ and $H$ are not equal. 

To prove this lemma we relax the requirement for an $\ell$-topology to have a connected underlying undirected graph, but keep the requirement that its vertices have distinct colors. The number of (not equal) labelled digraphs on $n$ vertices and $m$ arcs is clearly ${n(n-1) \choose m}$.  Thus, the number of $p$-topologies with $p$ vertices and at most $\ell$ arcs  is ${p(p-1) \choose \le \ell}=2^{\OO(\ell\log p)}$. The claim of the lemma follows from this bound and the fact that the number of choices for $p$ colors is ${\ell \choose p} <2^{\ell}.$ 
\end{proof}

Now, we argue that there exists a walk of the form that we seek that ``complies'' with at least one of our topologies. We formalize this claim in the following definition and observation.

\begin{definition}\label{def:complies}
Let $G$ be an $\ell$-colored digraph, and let $P$ be a colorful $r$-simple path in $G$. Let $T$ be an $\ell$-topology. We say that $P$ {\em complies with} $T$ if $P_{\mathrm{simple}}$ and $T$ are isomorphic under color preservation, i.e., there exists an isomorphism $\psi$ between $P_{\mathrm{simple}}$ and $T$ 
such that for all $v\in V(P_{\mathrm{simple}})$, the colors of $v$ and $\psi(v)$ are equal. The function $\psi$ is called a {\em witness}.
\end{definition}

\begin{observation}\label{obs:complies}
Let $(G,k,r,s,t)$ be an instance of \diColPathRst. Then, for any colorful $r$-simple $(s,t)$-path $P$, there exists a unique topology $T\in{\cal T}_{\mathsf{b}(k/r)}$ with which $P$ complies.
\end{observation}

\paragraph{Enriching the topology via ILP} In light of Observation \ref{obs:complies}, a natural approach to solve \diColPathRst\ would be to guess a topology, test whether the input digraph has a subgraph isomorphic to it, and then try to answer the question of whether this topology can be extended into an $r$-simple $(s,t)$-path.  However, the second step of this approach already has a major flaw---for example, if the topology is a clique, then it captures the {\sc Multicolored Clique} problem (defined in Section \ref{sec:W1}). 
Instead, we will first try to extend the topology into an \emph{enriched topology} (defined below), effectively corresponding to computing a candidate isomorphism class of $P_{\mathrm{multi}}$ instead of just $P_{\mathrm{simple}}$.  This step is performed independently of the input graph.  Then, having chosen an enriched topology, we can look for a ``relaxed embedding'' of it into the input graph $G$, intuitively allowing different ``visits'' to a vertex $v$ in the topology to be implemented by different vertices in $G$, as long as every such vertex has the same color as $v$.

To achieve our desired running time, it is crucial that we only need to compute one candidate enriched topology for every topology.
This part will be done via integer linear programming (ILP). Notice that we cannot even explicitly write an $r$-simple $(s,t)$-path that the enriched topology encodes, since the size of it is already $\OO(k)$ (while the input size is only $\OO(n+\log k)$), hence checking whether the guess can be realized (i.e., looking for the relaxed embedding) is slightly tricky. However, we deal with this task later. For now, let us first explain how an enrichment of a topology is defined.

\begin{definition}\label{def:enrichedTopology}
Let $\ell,r\in\mathbb{N}$. In addition, let $i,j\in\{1,2,\ldots,\ell\}$, $i\neq j$. Then, an {\em $r$-enriched $\ell$-topology with endpoints $i,j$}
is a pair $(T,\varphi)$ of an $\ell$-topology $T$ and a function $\varphi: A(T)\rightarrow \{1,2,\ldots,r\}$
such that $T$ with arc multiplicities $\varphi$ admits an Euler trail with endpoints of colors $i$ and  $j$.
Explicitly, we require the following properties:
\begin{enumerate}
\item There exist vertices $s=s(T,\varphi)\in V(T)$ and $t=t(T,\varphi)\in V(T)$ colored $i$ and $j$, respectively. 
\item For every vertex $v\in V(T)\setminus \{s,t\}$, it holds that $\displaystyle{\sum_{u: (u,v)\in A(T)}\varphi(u,v) = \sum_{u: (v,u)\in A(T)}\varphi(v,u) \leq r}$.
\item $\displaystyle{\sum_{u: (u,s)\in A(T)}\varphi(u,s)}+1 = \displaystyle{\sum_{u: (s,u)\in A(T)}\varphi(s,u)} \leq r$.
\item $\displaystyle{\sum_{u: (u,t)\in A(T)}\varphi(u,t)} = \displaystyle{\sum_{u: (t,u)\in A(T)}\varphi(t,u)}+1 \leq r$.
\end{enumerate}
\end{definition}

Now, we show how to enrich a topology (if it is possible). For this purpose, we utilize Theorem \ref{prop:ILP}.
Note that the quantity $\sum_e \varphi(e)$ corresponds to the length of the solution.

\begin{lemma}\label{lem:enrich}
There exists an algorithm that given $\ell,r\in\mathbb{N}$, $i,j\in\{1,2,\ldots,\ell\}$, $i\neq j$, and an $\ell$-topology $T$, determines in time $\ell^{\OO(\ell)}\cdot(\log r)^{\OO(1)}$ and polynomial space whether there exists a function $\varphi: A(T)\rightarrow \{1,2,\ldots,r\}$ such that $(T,\varphi)$ is an $r$-enriched $\ell$-topology with endpoints $i,j$. In case the answer is positive, the algorithm outputs such a function $\varphi$ that maximizes $\sum_{e\in A(T)}\varphi(e)$.
\end{lemma}

\begin{proof}
If there do not exist vertices $s$ and $t$ in $V(T)$ colored $i$ and $j$, respectively, then there does not exist a function $\varphi$ such that $(T,\varphi)$ is an $r$-enriched $\ell$-topology with endpoints $i,j$, and hence we are done. Therefore, we next suppose that there exist such vertices, and since they are uniquely defined (since $T$ is an $\ell$-topology), we can denote them by $s$ and $t$ accordingly.
We formulate our task by using ILP. Here, we have a variable $x_e$ for every arc $e\in A(T)$ that encodes the value assigned by $\varphi$ to $e$. The objective function is $\max\sum_{e\in A(T)}x_e$. Now, the constraints are defined as follows.
\begin{itemize}
\item For every vertex $v\in V(T)\setminus \{s,t\}$, we have two constraints:
\[\begin{array}{l}
\displaystyle{\sum_{u: (u,v)\in A(T)}x_{(u,v)} = \sum_{u: (v,u)\in A(T)}x_{(v,u)}};\\
\displaystyle{\sum_{u: (u,v)\in A(T)}x_{(u,v)}} \leq r.
\end{array}\]
\item In addition, we have the following four constraints:
\[\begin{array}{l}
\displaystyle{\sum_{u: (u,s)\in A(T)}x_{(u,s)}}+1 = \displaystyle{\sum_{u: (s,u)\in A(T)}x_{(s,u)}};\\
\displaystyle{\sum_{u: (u,s)\in A(T)}x_{(u,s)}}+1 \leq r;\\
\displaystyle{\sum_{u: (u,t)\in A(T)}x_{(u,t)}} = \displaystyle{\sum_{u: (t,u)\in A(T)}x_{(t,u)}}+1;\\
\displaystyle{\sum_{u: (u,t)\in A(T)}x_{(u,t)}} \leq r.
\end{array}\]
\item For every arc $e\in A(T)$, we have the constraint   $x_e\in \mathbb{N}_{\ge 1}.$
\end{itemize}
This completes the description of the ILP formulation. 

The size of the ILP instance is $L=\OO(|V(T)|(|A(T)|+\log r))=\OO(\ell^2+\ell\log r)$, it consists of $p=|A(T)|\leq\ell$ variables, $M_x=r$ is the largest absolute value a variable can take in a solution, and $M_c=1$ is the largest absolute value of a coefficient in the cost vector.
 Thus, by Theorem \ref{prop:ILP}, this ILP instance can be solved using polynomial space and in time 
\[p^{2.5p+o(p)}\cdot(L+\log M_x)\cdot\log(M_xM_c) = \ell^{\OO(\ell)}\cdot(\log r)^{\OO(1)}.\]

The ILP formulation immediately implies that if the ILP instance does not have a solution, then there does not exist a function $\varphi: A(T)\rightarrow \{1,2,\ldots,r\}$ such that $(T,\varphi)$ is an $r$-enriched $\ell$-topology with endpoints $i,j$. If the ILP instance has a solution, then such a function $\varphi$ that maximizes $\sum_{e\in A(T)}\varphi(e)$ is defined as follows: for any $e\in A(T)$, define $\varphi(e)$ as the value assigned to $x_e$ by the solution.
\end{proof}

Next, we define what does it mean for a solution to ``comply'' with an enriched topology.

\begin{definition}\label{def:enrichComply}
Let $\ell,r\in\mathbb{N}$, $G$ be an $\ell$-colored digraph, and let $P$ be a colorful $r$-simple $(s,t)$-path in $G$. Let $i$ be the color of $s$, $j$ be the color of $t$, and $(T,\varphi)$ be an $r$-enriched $\ell$-topology with endpoints $i,j$.
We say that $P$ {\em complies with} $(T,\varphi)$ if $P$ complies with $T$, and for the function $\psi$ that witnesses this, for every arc $(u,v)\in P_{\mathrm{simple}}$, the number of copies $(u,v)$ has in $P_{\mathrm{multi}}$ is exactly $\varphi(\psi(u,v))$.
\end{definition}

Let us now argue that the choice of how to enrich a topology is immaterial as long as at least one enrichment exists (in which case, we also need to compute such an enrichment). 

\begin{lemma}\label{lem:enrichImmaterial}
Let $G$ be an $\ell$-colored graph, and let $P$ be a colorful $r$-simple $(s,t)$-path in $G$ with $s\neq t$. Let $i$ be the color of $s$, and $j$ be the color of $t$. Then, the following conditions hold.
\begin{enumerate}
\item\label{item:enrichImmaterial1} There exists an $r$-enriched $\ell$-topology with endpoints $i,j$ with which $P$ complies.
\item\label{item:enrichImmaterial2} Let $T$ be an $\ell$-topology with which $P$ complies. Then, for any $r$-enriched $\ell$-topology with endpoints $i,j$, say $(T,\varphi)$, there exists an $r$-simple $(s,t)$-path in $G$ that complies with $(T,\varphi)$.
\end{enumerate}
\end{lemma}

\begin{proof}
For the first condition, define $T$ as $P_{\mathrm{simple}}$ with loops removed. Moreover, define $\varphi: A(T)\rightarrow\mathbb{N}$ as follows. For all $e\in A(T)$, let $\varphi(e)$ be the number of copies of $e$ in $P_{\mathrm{multi}}$. Since $P$ is an Euler $(s,t)$-trail in $P_{\mathrm{multi}}$, by Theorem \ref{prop:euler}, the out-degree and in-degree of every vertex in $P_{\mathrm{multi}}$ are equal, except for $s$ and $t$ which satisfy $d^+(s)=d^-(s)+1$ and $d^-(t)=d^+(t)+1$. Thus, it is immediate that $(T,\varphi)$ is an $r$-enriched $\ell$-topology with endpoints $i,j$ with which $P$ complies.

For the second condition, let $T$ be an $\ell$-topology with which $P$ complies, and consider some function $\varphi: A(T)\rightarrow\mathbb{N}$ such that $(T,\varphi)$ is an $r$-enriched $\ell$-topology with endpoints $i,j$. Let $T'$ be the directed multigraph obtained from $T$ by duplicating every arc $e\in A(T)$ to have exactly $\varphi(e)$ copies. Let $s'$ and $t'$ be the (unique) vertices colored $i$ and $j$ in $T'$, respectively. Since $(T,\varphi)$ is an $r$-enriched $\ell$-topology with endpoints $i,j$, it holds that the out-degree and in-degree of every vertex in $T'$ are equal, except for $s'$ and $t'$ which satisfy $d^+(s')=d^-(s')+1$ and $d^-(t')=d^+(t')+1$. By Theorem \ref{prop:euler}, this means that there exists an Euler $(s',t')$-trail $P'$ in $T'$. Let $\psi: V(P')\rightarrow V(G)$ be the function that maps each vertex in $V(P')$ to the (unique) vertex of the same color in $P$. Then, for any arc $(u,v)\in A(T')$, the pair $(\psi(u),\psi(v))$ is an arc that is visited at least once by $P$, and hence $(\psi(u),\psi(v))\in A(G)$. This implies that $\psi$ maps $P'$ to an $r$-simple $(s,t)$-path $\widehat{P}$ in $G$. By construction, it holds that $\widehat{P}$ complies with $(T,\varphi)$.
\end{proof}

This lemma motivates a problem definition where the input includes an $r$-enriched $\ell$-topology with endpoints $i,j$, and we seek an $r$-simple $(s,t)$-path in $G$ that complies with it. However, like before, such a problem encompasses {\sc Multicolored Clique}. Instead, we need a relaxed notion of compliance, which we define as follows.
This corresponds to the notion of a relaxed embedding mentioned previously.

\begin{definition}\label{def:weaklyComply}
Let $\ell,r\in\mathbb{N}$. Let $(T,\varphi)$ be an $r$-enriched $\ell$-topology $(T,\varphi)$ with endpoints $i,j$. Let $P$ be an $r$-simple $(s,t)$-path in an $\ell$-colored digraph $G$, where $i$ is the color of $s$ and $j$ is the color of~$t$. Then, $P$ {\em weakly complies} with $(T,\varphi)$ if the following conditions hold.
\begin{itemize}
\item Every color that occurs in $P$ also occurs in $T$ and vice versa. That is, there exists a unique, surjective (but not necessarily injective) function $\psi: V(P_{\mathrm{simple}})\rightarrow V(T)$ where for all $v\in V(P_{\mathrm{simple}})$, the colors of $v$ and $\psi(v)$ are equal.
\item For every two colors $a,b$ that occur in $T$, the number of times arcs directed from a vertex colored $a$ to a vertex colored $b$  occur in $P$ is precisely $\varphi(u,v)$ where $u$ and $v$ are the (unique) vertices in $T$ colored $a$ and $b$, respectively.
\end{itemize} 
\end{definition}
Note that if a walk $P$ complies with $(T,\varphi)$, then it also weakly complies with $(T,\varphi)$, but the opposite is not true. In particular, a walk where some distinct vertices have the same color can weakly comply with $(T,\varphi)$, but it necessarily does not comply with $(T,\varphi)$.

\begin{defproblem}
{\topologyProb} 
{A tuple $(G, \ell, r, s, t, (T, \varphi))$ where
$\ell, r \in \mathbb{N}$, $G$ is an $\ell$-colored digraph, $s, t \in V(G)$ are distinct vertices, and $(T,\varphi)$ is an enriched $\ell$-topology
with endpoints $i,j$ where $i$ is the color of $s$ and $j$ is the color of $t$.}
{ Return \yes\ or \no\ as follows: 
{\em (i)} If $G$ has an $r$-simple $(s,t)$-path that complies with $(T,\varphi)$, then return \yes;
{\em (ii)} If $G$ has no $r$-simple $(s,t)$-path that weakly complies with $(T,\varphi)$, then return \no;
{\em (iii)} If none of the two conditions above holds, we can return either \yes\ or \no.
}
\end{defproblem}\\

The \topologyProb\ problem allows us to
determine whether there exists an $r$-simple $(s,t)$-path in $G$ that weakly complies with $(T,\varphi)$.

\begin{lemma}\label{lem:reduceToTopologyProb}
Suppose that \topologyProb\ can be solved in time $f(\ell)\cdot (n+\log r)^{\OO(1)}$ and polynomial space. Then, \diColPathRst\ can be solved in time $2^{\OO(\mathsf{b}(k/r)\log(\mathsf{b}(k/r)))}\cdot f(\mathsf{b}(k/r))\cdot (n+\log k)^{\OO(1)}$ and polynomial space.
Here, $\mathsf{b}(k/r)$ is the function defined at the start of Section~\ref{sec:dir:colorcoding}.
\end{lemma}

\begin{proof}
Let $\cal A$ be an algorithm that solves \topologyProb\ in time $f(k/r)\cdot (n+\log k)^{\OO(1)}$ and polynomial space. In what follows, we describe how to solve \diColPathRst. To this end, let $(G,k,r,s,t)$ be an instance of \diColPathRst. Let $i$ be the color of $s$ and $j$ be the color of $t$. Initialize $k^\star:=0$. For every topology $T\in{\cal T}_{\mathsf{b}(k/r)}$, we execute the following computation. First, call the algorithm in Lemma \ref{lem:enrich} to check in time $2^{\OO(\mathsf{b}(k/r)\log(\mathsf{b}(k/r)))}\cdot(\log r)^{\OO(1)}$ and polynomial space whether  there exists a function $\varphi: A(T)\rightarrow \{1,2,\ldots, r\}$ such that $(T,\varphi)$ is an $r$-enriched $\ell$-topology with endpoints $i,j$.  If the answer is positive, the algorithm outputs such a function $\varphi$ that maximizes $\sum_{e\in A(T)}\varphi(e)$. In this case, we proceed as follows. We call the algorithm $\cal A$ with $(G,\mathsf{b}(k/r),r,s,t,(T,\varphi))$. If the answer of $\cal A$ is positive and $\sum_{e\in A(T)}\varphi(e)\geq k^\star$, then update $k^\star:=1+\sum_{e\in A(T)}\varphi(e)$. 
In the case that no such $\varphi$ exists for $T$, proceed with the next topology.
After all topologies in ${\cal T}_{\mathsf{b}(k/r)}$ were examined, we return $k^\star$.

By Observation \ref{lem:fewTopologies}, $|{\cal T}_{\mathsf{b}(k/r)}|=2^{\OO(\mathsf{b}(k/r)\log(\mathsf{b}(k/r)))}$. Thus, it is clear that the algorithm runs in time $2^{\OO(\mathsf{b}(k/r)\log(\mathsf{b}(k/r)))}\cdot f(\mathsf{b}(k/r))\cdot (n+\log k)^{\OO(1)}$ and uses polynomial space. Next, we show that the algorithm is correct, that is, that it solves \diColPathRst.

In one direction, let $P$ be a colorful $r$-simple $(s,t)$-path in $G$, and let $q$ denote its size. We need to show $k^\star\geq q$. By Observation \ref{obs:complies}, there exists a unique topology $T\in{\cal T}_{\mathsf{b}(k/r)}$ with which $P$ complies. Further, Property \ref{item:enrichImmaterial1} in Lemma \ref{lem:enrichImmaterial} states that there exists an $r$-enriched $\mathsf{b}(k/r)$-topology $(T,\varphi')$ with which $P$ complies. Thus, when $T$ is examined, the algorithm in Lemma \ref{lem:enrich} returns a function $\varphi$ such that $(T,\varphi)$ is an $r$-enriched $\ell$-topology with endpoints $i,j$, and $\sum_{e\in A(T)}\varphi(e)\geq \sum_{e\in A(T)}\varphi'(e)$. By Property \ref{item:enrichImmaterial2} in Lemma \ref{lem:enrichImmaterial}, there exists an $r$-simple $(s,t)$-path in $G$ that complies with $(T,\varphi)$. Thus, $\cal A$ must return a positive answer. Since $\sum_{e\in A(T)}\varphi'(e)=q-1$, we have that $k^\star\geq q$.

In the other direction, we need to show that $G$ has an $r$-simple $(s,t)$-path of size at least $k^\star$. Consider the topology $T\in{\cal T}_{\mathsf{b}(k/r)}$ in whose examination $k^\star$ was updated to its final value. Then, there exists a function $\varphi: A(T)\rightarrow \{1,2,\ldots,r\}$ such that $(T,\varphi)$ is an $r$-enriched $\ell$-topology with endpoints $i,j$, and $1+\sum_{e\in A(T)}\varphi(e)=k^\star$. Moreover, by the correctness of $\cal A$, there exists an $r$-simple $(s,t)$-path $P$ in $G$ that weakly complies with $(T,\varphi)$. By the definition of weak compliance, the length of $P$ is exactly $\sum_{e\in A(T)}\varphi(e)$, and hence its size is $k^\star$.
\end{proof}

\subsection{Verifying Whether a Guess is Realizable}

It remains to solve the \topologyProb\ problem. Let us first remark that if we allowed a linear dependency on $k$ in the running time, then this task would have been easier than our actual task, since we could have used the following approach: first, we would have computed some walk $P^\star=v_{i_1} \ldots v_{i_d}$ that uses every arc $e$ in the input enriched topology exactly $\varphi(e)$ times---note that the size~$d$ of such a $P^\star$ can be $\Omega(\ell r)$ (that is, $\Omega(k)$ if we trace the source of $\ell$); then, we could have used a simple dynamic programming (DP) computation to check whether the input digraph $G$ contains such a colored walk (where vertices having the same color in $P^\star$ are allowed to be mapped to distinct vertices in $G$ as long as these vertices have the same color). 
This could be done by a simple table $T[v,j]$ that stores, for every $j \in [d]$ and every vertex $v \in V$ of the same color as $v_{i_j}$, whether there is a colored walk in $G$ ``implementing'' the prefix $v_{i_1} \ldots v_{i_j}$ of $P^\star$, ending at the vertex $v$.
Note how the use of vertex colors guarantees that the result is an $r$-simple path, even if it only weakly complies with the enriched topology.

To solve \topologyProb\ while attaining a logarithmic dependency on $k$, instead of searching for a walk one step a time, we decompose the enriched topology into cycles.  
Let $(T, \varphi)$ be an enriched topology, and let $C$ be a cycle in $T$. 
Let $M$ be the smallest value of $\varphi(e)$ for arcs $e$ in $C$.  
Then we can search for $(T,\varphi)$ as follows: find a copy of $C$ in $G$; remove $M$ copies of every arc of $C$ from $(T, \varphi)$, deleting arcs whose multiplicity reaches 0; then recursively, via DP, find a copy in $G$ of every connected component of the resulting enriched topology. 
If we ensure that colors are preserved in all steps, and that the resulting subgraph $H$ of $G$ is connected, 
then $H$ will admit a walk that forms an $r$-simple path which weakly complies with $(T, \varphi)$.

We now present the recursive algorithm that (combined with DP) solves \topologyProb.  Due to the nature of the recursion, we need to consider an annotated version of \topologyProb, defined as follows. 

\begin{defproblem}
{\topologyProbRoot} 
{Integers $\ell,r\in\mathbb{N}$, an $\ell$-colored digraph $G$, distinct vertices $s,t\in V(G)$, an $r$-enriched $\ell$-topology $(T,\varphi)$ with endcolors $i,j$ where $i$ is the color of $s$ and $j$ is the color of $t$, and a vertex $v_r \in V(G)$ called the \emph{root vertex}.}
{Return \yes\ or \no\ as follows. {\em (i)} If $G$ has an $r$-simple $(s,t)$-path that complies with $(T,\varphi)$ and visits the root vertex at least once, then return \yes. In this case, the input is called a {\em \yes-instance}.
{\em (ii)} If $G$ has no $r$-simple $(s,t)$-path that weakly complies with $(T,\varphi)$ and visits the root vertex at least once, then return \no. In this case, the input is called a {\em \no-instance}.
{\em (iii)} If none of the two conditions above holds, we can return either \yes\ or \no. In this case, the input is called an {\em irrelevant} instance.
}
\end{defproblem}\\

To describe the recursion, let us first make a simple observation about the structure of strong components of an $r$-enriched $\ell$-topology $(T,\varphi)$.

\begin{lemma} \label{lem:nonstrong}
  Let $(T,\varphi)$ be an $r$-enriched $\ell$-topology with endcolors $i,j$, and let its endpoints be $s$, $t$. 
  Let $\mathcal{Q}$ be the set of strong components of $T$, say $|\mathcal{Q}|=d$.
  Then we can arrange the strong components as $\mathcal{Q}=\{Q_1, \ldots, Q_d\}$ such that
  for every $c \in [d-1]$ there is a single arc $e$ from $Q_c$ to $Q_{c+1}$ with $\varphi(e)=1$,
  and $T$ contains no other arcs between distinct strong components in $T$.
  Furthermore, let $s_1=s$, $t_d=t$, and for each $c \in [d-1]$ let the arc from $Q_c$ to $Q_{c+1}$ be $t_cs_{c+1}$.
  Then for every $c \in [d]$, the graph $Q_c$ admits an $(s_c,t_c)$-walk that visits each arc $e \in A(Q_c)$ precisely $\varphi(e)$ times.
\end{lemma}
\begin{proof}
  By Definition \ref{def:enrichedTopology} and Theorem \ref{prop:euler}, there is a directed walk from $s$ to $t$ in $T$ that uses every arc $e$ precisely $\varphi(e)$ times.  Clearly, this is only possible under the conditions described. 
\end{proof}

We define the basis of our recursion as the case where the topology is a DAG. Then, we make use of the following lemma.

\begin{lemma}\label{lem:topologyBasis}
There exists an algorithm that, given an instance $I=(G,\ell,r,s,t,(T,\varphi),v_r)$ of \topologyProbRoot\ where $T$ is a DAG, solves $I$ in polynomial time and space.
\end{lemma}

\begin{proof}
Let $(G,\ell,r,s,t,(T,\varphi),v_r)$ be an instance of \topologyProbRoot\ where $T$ is a DAG. By Lemma~\ref{lem:nonstrong}, this means that $T$ is a (simple directed) path and that for every arc $e\in A(T)$, it holds that $\varphi(e)=1$.  If no vertex in $T$ has the same color as $v_r$, then it is clear that there is no $r$-simple $(s,t)$-path in $G$ that complies with $(T,\varphi)$ and which visits $v_r$ at least once. Thus, we next suppose that this is not the case. Let $G'$ be the digraph obtained by removing from $G$ all vertices whose color does not occur in $T$ as well as every vertex $v \neq v_r$ that has the same color as $v_r$. Then, $(G,\ell,r,s,t,(T,\varphi),v_r)$ is a \yes-instance if and only if $G'$ has an $(s,t)$-walk $P$ that is isomorphic to $T$ under color preservation, i.e.,~the isomorphism must map each vertex in $P$ to a vertex of the same color in $T$ (then, the walk is necessarily a path that uses $v_r$). However, this task can be easily checked by removing from $G'$ all arcs from a vertex colored $i$ to a vertex colored $j$ for all colors $i,j$ such that $T$ has no arc from a vertex colored $i$ to a vertex colored $j$, and then checking (e.g.,~by using BFS) whether $t$ is reachable from $s$.
\end{proof}

In each step, we decompose the current topology further, in two ways.  The first type of decomposition applies when $T$ contains multiple strong components.  
For technical reasons, we need to introduce \emph{rooted} topologies.

\begin{definition} \label{def:strong-decomp}
  A \emph{rooted $r$-enriched $\ell$-topology} is a triple $(T,\varphi,u)$ where $(T,\varphi)$ is an $r$-enriched $\ell$-topology and $u \in V(T)$ a vertex referred to as the \emph{root vertex} of the topology.  For a rooted $r$-enriched $\ell$-topology $(T, \varphi, u)$, the \emph{strong component decomposition} of $(T,\ell,u)$
  is the sequence $(Q_c, s_c, t_c)_{c=1}^d$ where $\mathcal{Q}=\{Q_1, \ldots, Q_d\}$ and $s_c, t_c \in V(Q_c)$ are as in Lemma~\ref{lem:nonstrong}.
  If $d>1$, then for each $c \in [d]$ the \emph{$c^{\text{th}}$ (rooted enriched) subtopology} of the decomposition is
  a rooted $r$-enriched $\ell$-topology $(Q_c', \varphi_c, u_c)$ defined as follows.
  \begin{enumerate}
  \item If $c<d$, then $Q_c'=T[V(Q_c) \cup \{s_{c+1}\}]$, $\varphi_c$ is $\varphi$ restricted to $V(Q_c')$, and $u_c=t_c.$
  \item If $c=d$, then $Q_c'=T[V(Q_c) \cup \{t_{c-1}\}]$, $\varphi_c$ is $\varphi$ restricted to $V(Q_c')$, and $u_c=s_c.$
  \end{enumerate}
  We say that $(T,\varphi,u)$ is \emph{decomposable} if $d>1$ and $u \in \{s_c,t_c: c \in [d]\}$.
\end{definition}

For topologies with a non-trivial decomposition into strong components, we define a collection of subinstances $I_{c,u,v}$ where $c \in [d]$, $u, v \in V(G)$ as follows.

\begin{definition} \label{def:strong-subinstance}
  Let $I=(G,\ell,r,s,t,(T,\varphi),v_r)$ be an instance of \topologyProbRoot\ and let
   $u_r \in V(T)$ be the vertex with the same color as $v_r$.
  We say that $I$ is \emph{decomposable} if $(T,\varphi,u_r)$ is decomposable. 
  Assume that $I$ is decomposable, and   
  let $(Q_c, s_c, t_c)_{c=1}^d$ be the strong component decomposition of $(T,\varphi,u_r)$. 
  For $c \in [d]$ and $u,v \in V(G)$, the triple $(c,u,v)$ is \emph{valid} if the following conditions apply.
  \begin{enumerate}
  \item If $c=1$, then $u=s$; otherwise $u$ is a vertex with the same color as $s_c.$
  \item If $c=d$, then $v=t$; otherwise $v$ is a vertex with the same color as $t_c.$
  \item If $c<d$, then $v$ has an out-neighbor with the same color as $s_{c+1},$ 
    otherwise $u$ has an in-neighbour with the same color as $t_{c-1}.$
  \item If the color of $v_r$ matches that of $u$ ($v,$ respectively),
    then $u=v_r$ ($v=v_r$, respectively).
  \end{enumerate}
  For any valid triple $(c,u,v)$, the \emph{subinstance} $I_{c,u,v}$ is the instance of \topologyProbRoot\ defined as follows.  Let $(Q_c', \varphi_c, u_c)$ be the $c$:th subtopology of the decomposition. Then 
  \[
    I_{c,u,v}=(G,\ell,r,s',t',(Q_c',\varphi_c),v_c),
  \]
  where $s'=u$ if $c<d$, and otherwise $s'$ is some in-neighbor of $u$ of the same color as $t_{c-1}$; 
  $t'=t$ if $c=d$, and otherwise $t'$ is some out-neighbor of $v$ of the same color as $s_{c+1}$; 
  and $v_c=v$ if $c<d$, and otherwise $v_c=u$. 

  Finally, a sequence $(u_c, v_c)_{c=1}^d$ of pairs of vertices of $G$
  is \emph{good} (\emph{excellent}, respectively) with respect to the strong component decomposition 
  if the following conditions hold:
  \begin{enumerate}
  \item For each $c \in [d-1]$, the arc $(v_c,u_{c+1})$ exists in $G.$
  \item For each $c \in [d]$, the triple $(c, u_c, v_c)$ is valid and
    the subinstance $I_{c,u_c,v_c}$ is not a \no-instance (a \yes-instance, respectively).
  \end{enumerate}
\end{definition}

Let us show the correctness condition for this decomposition.

\begin{lemma} \label{lem:strong-recurse}
  Let $I=(G,\ell,r,s,t,(T,\varphi),v_r)$ be a decomposable instance of \topologyProbRoot\ and
  let $(Q_c, s_c, t_c)_{c=1}^d$ be the strong component decomposition of $(T,\varphi,v_r)$.
  Then the following hold.
\begin{itemize}
\item If $I$ is a \yes-instance, then there exists a color-preserving 
  map $\psi \colon \{s_c, t_c:  c \in [d]\} \to V(G)$
  such that the sequence $(\psi(s_c), \psi(t_c))_{c=1}^d$
  is excellent with respect to $(Q_c, s_c, t_c)_{c=1}^d.$
\item If $I$ is a \no-instance, then there does not exist a color-preserving 
  map $\psi \colon \{s_c, t_c:  c \in [d]\} \to V(G)$
  such that the sequence $(\psi(s_c), \psi(t_c))_{c=1}^d$
  is good with respect to $(Q_c, s_c, t_c)_{c=1}^d.$
\end{itemize}
\end{lemma}
\begin{proof}
  First assume that $I$ is a \yes-instance, i.e., $G$ has an $r$-simple $(s,t)$-path $P$ that complies with $(T,\varphi)$ and visits the root vertex at least once.  Since $P$ complies with $(T,\varphi)$ there exists a color-preserving isomorphism between $P_{\mathrm{simple}}$ and $T$; let $\psi \colon V(T) \to V(P)$ be the mapping implied by this.   It is easy to check that $(\psi(s_c), \psi(t_c))_{c=1}^d$ is excellent with respect to $(Q_c, s_c, t_c)_{c=1}^d$.  Indeed, since $P_{\mathrm{simple}}$ is isomorphic to $T$ the decomposition $(Q_c, s_c, t_c)_{c=1}^d$ is also structurally consistent with $P_{\mathrm{simple}}$, i.e., $\psi$ is also a color-preserving isomorphism between $(Q_c)^d_{c=1}$ and the strong components of $P_{\mathrm{simple}}$, and every $(s,t)$-walk in $P_{\mathrm{simple}}$ must consist of an alternating sequence of $(\psi(s_c),\psi(t_c))$-walks in the image $\psi(Q_c)$ of $Q_c$ for $c \in [d]$, and single uses of arcs $(\psi(t_c),\psi(s_{c+1}))$.  Furthermore $s=\psi(s_1)$ and $t=\psi(t_d)$, and if the color of $v_r$ matches that of $s_c$ or $t_c$ for some $c \in [d]$, then $\psi$ maps that vertex to $v_r$.  Thus for every $c \in [d]$ the triple $(c,\psi(s_c),\psi(t_c))$ is valid, and by extending the $(\psi(s_c),\psi(t_c))$-walk by a first visit to $s'$ or a last visit to $t'$ as needed, we get an $(s',t')$-walk in $G$ that complies with the $c$:th enriched subtopology of the decomposition.  Furthermore, since $P$ visits both $v_r$, $\psi(s_c)$ and $\psi(t_c)$, this part of the walk must visit $v_c$.
  Thus $I_{c,\psi(s_c),\psi(t_c)}$ is a \yes-instance and $(\psi(s_c), \psi(t_c))_{c=1}^d$ is excellent with respect to $(Q_c, s_c, t_c)_{c=1}^d$.
  
  On the other hand, let $\psi$ be a color-preserving mapping such that $(\psi(s_c), \psi(t_c))_{c=1}^d$ is good with respect to $(Q_c, s_c, t_c)_{c=1}^d$.
  That is, $G$ contains an arc $(\psi(t_c),\psi(s_{c+1})))$ for every $c \in [d-1]$, and for every $c \in [d]$ the triple $(c,\psi(s_c),\psi(t_c)$ is valid and $I_{c, \psi(s_c), \psi(t_c)}$ is not a \no-instance.  By construction, due to the choice of endpoints $s'$, $t'$ and root vertex $v_c$ in $I_{c,\psi(s_c),\psi(t_c)}$, this implies that for every $c \in [d]$ there is an $r$-simple $(s_c,t_c)$-path that weakly complies with $(Q_c, \varphi_c')$, where $\varphi_c'$ is $\varphi$ restricted to $A(Q_c)$. Furthermore $\psi(s_1)=s$ and $\psi(t_d)=t$.  Thus the solutions to the subinstances can be concatenated into a single $(s,t)$-walk $P$.  Furthermore, since the components $Q_c$ have pairwise disjoint sets of vertex colors, and by the definition of weak compliance, these solutions are pairwise vertex-disjoint and $P$ is an $r$-simple $(s,t)$-path. We show that $P$ weakly complies with $(T,\varphi)$ and visits $v_r$ at least once.  By Definition~\ref{def:weaklyComply}, the former requires that $P$ and $T$ use the same sets of colors and that for every arc $e \in A(T)$ from some color $i$ to some color $j$, $P$ uses precisely $\varphi(e)$ arcs from a vertex of color $i$ to a vertex of color $j$. 
  For the first requirement, the colors used in $T$ are partitioned by the strong components $Q_c$, and for each $c \in [d]$, the solution to $I_{c,\psi(s_c),\psi(t_c)}$ uses the same set of colors as $Q_c$.  Hence this part follows.  For the second requirement, let $e \in A(T)$.  If $e$ goes between distinct components, then $\varphi(e)=1$ and $P$ uses precisely one arc with colors matching the endpoints of $e$.  Otherwise, $e \in A(Q_c)$ for some $c \in [d]$, and $P$ contains arcs matching the colors of $e$ only within the solution to $I_{c,\psi(s_c),\psi(t_c)}$, where it contains precisely $\varphi(e)$ such arcs by the definition of weak compliance.
  Finally, since $I$ is decomposable and each triple $(c,\psi(s_c),\psi(t_c))$ is valid, there is some $c \in [d]$ such that $v_r \in \{\psi(s_c),\psi(t_c)\}$, thus $P$ visits $v_r$ at least once. 
  We conclude that $G$ is not a \no-instance.
\end{proof}

We need a further decomposition step to decompose strong components. 
In each such step, we process a (directed simple) cycle from the current topology so that at least one of its arcs is eliminated.  Here, in order to eventually derive a logarithmic dependency on $k$, it is crucial that we completely eliminate an arc and not only decrease the value that $\varphi$ assigns to it. For this purpose, we utilize the following definition 
and lemma.  
 
\begin{definition}\label{def:goodCyc}
  For a rooted $r$-enriched $\ell$-topology $(T,\varphi,u)$,
  a tuple $B=(C,M,E,\mathcal{Q},f_T)$ is \emph{relevant} if
  $C$ is a simple, directed cycle $C$ in $T$ with $u \in V(C)$,
  $M=\min_{e\in A(C)}\varphi(e)$, $E=\{e\in A(C): \varphi(e)=M\}$,
  ${\cal Q}$ is the set of weakly connected components of $T-E$ and
  $f_T$ is the function that assigns to each $Q \in \mathcal{Q}$ a vertex of $V(C) \cap V(Q)$ as follows. If $|\mathcal{Q}|=1$, then $f_T(Q)=u$; otherwise $f_T(Q)$ is the last vertex in $V(C)\cap V(Q)$ of $C$, counting from $u$, such that the subsequent vertex along $C$ does not lie in $Q$.

  For $Q \in \mathcal{Q}$, the \emph{subtopology at $Q$ (of $(T,\varphi,u)$, with respect to $B$)} 
  is the rooted $r$-enriched $\ell$-topology $(T_Q, \varphi_Q, f_T(Q))$ defined as follows. 
  \begin{enumerate}
  \item Let the endpoints of $(T,\varphi)$ be $s',t'$. If $Q$ contains $s'$ and $t'$, then
    $T_Q=Q$, and $\varphi_Q$ is $\varphi$ restricted to $A(Q)$ where the value $\varphi_Q(e)$ has been decreased by $M$ for every arc $e \in A(C)$. 
  \item Otherwise, let $v$ be the successor of $f_T(Q)$ in $C$, and define $T_Q$ from $Q$ by adding a new vertex $t_Q$ colored by the same color as $v$, and add the arc $(f_T(Q),t_Q)$). Let $\varphi_Q$ be defined as in the previous case, extended with $\varphi_Q((f_T(Q),t_Q))=1$. 
  \end{enumerate}

  Let $(G,\ell,r,s,t,(T,\varphi),v_r)$ be an instance of \topologyProbRoot\ where $T$ is not a DAG and let $u \in V(T)$ have the same color as $v_r$. Let $(C,M,E,{\cal Q},f_T)$ be a relevant tuple for $(T,\varphi,u)$. 
  A cycle $C'$ in $G$ is {\em good} ({\em excellent}, respectively) with respect to $(C,M,E,{\cal Q},f_T)$ if {\em (i)} there exists a color-preserving isomorphism $\psi$ between $C'$ and $C$, and {\em (ii)} for every $Q\in{\cal Q}$, the instance $J_Q$ defined as follows is not a \no-instance (a \yes-instance, respectively). Let $(T_Q,\varphi_Q,u_Q)$ be the subtopology at $Q$.
\begin{enumerate}
\item\label{item:goodCyc1} If the endpoints of $(T,\varphi)$ are contained in $Q$, 
  then $J_Q=(G,\ell,r,s,t,(T_Q,\varphi_Q), \psi(f_T(Q)))$. 
\item\label{item:goodCyc2} Otherwise, let $v$ be the successor of $f_T(Q)$ in $C$,
  and let $J_Q=(G,\ell,r,s',t',(T_Q,\varphi_Q),\psi(f_T(Q)))$
    where $s'=\psi(f_T(Q))$ and $t'=\psi(v)$. 
\end{enumerate}
\end{definition}

\begin{lemma} \label{lem:relevanttuple}
  Let $(T,\varphi,u)$ be a rooted $r$-enriched $\ell$-topology
  and let $C$ be a simple cycle of $T$ that contains $u$.
  Then there is precisely one relevant tuple $(C,M,E,{\cal Q},f_T)$.
\end{lemma}
\begin{proof}
  All of $M$, $E$, ${\cal Q}$ and $f_T$ are uniquely defined by $(T,\varphi,u)$
  and $C$, and all are well-defined.
\end{proof}

\begin{lemma}\label{lem:topologyStep}
Let $I=(G,\ell,r,s,t,(T,\varphi),v_r)$ be an instance of \topologyProbRoot\ where $T$ is not a DAG. Let $(C,M,E,{\cal Q},f_T)$ be any relevant tuple. Then, the following conditions hold.
\begin{itemize}
\item If $I$ is a \yes-instance, then $G$ has a cycle that is excellent w.r.t.~$(C,M,E,{\cal Q},f_T)$.
\item If $I$ is a \no-instance, then $G$ has no cycle that is good w.r.t.~$(C,M,E,{\cal Q},f_T)$.
\end{itemize}
\end{lemma}

\begin{proof}
First, suppose that $I$ is a \yes-instance. That is, $G$ has an $r$-simple $(s,t)$-path $P$ that complies with $(T,\varphi)$ and visits $v_r$ at least once. Then, $P_{\mathrm{simple}}$ has a unique cycle $C'$ with an isomorphism $\psi$ between $C'$ and $C$ that preserves colors. Define $H=P_{\mathrm{simple}}-\{e: \varphi(e) \in E\}$, and let $H_{\mathrm{multi}}$ be the directed multigraph obtained by removing $M$ copies of every arc in $C'$ from $P_{\mathrm{multi}}$. Now, consider some component $Q\in{\cal Q}$. Then, there exists a unique component $R$ in $H$ that is isomorphic to $Q$ under color preservation. Note that either both $s,t\in V(R)$ or both $s,t\notin V(R)$. (In the later case, no vertex in $Q$ has the same color as $s$ or $t$.) Let $R_{\mathrm{multi}}$ be the digraph obtained by duplicating each arc in $R$ to have the number of copies it has in $H_{\mathrm{multi}}$. Note that every vertex in $V(R)\setminus\{s,t\}$ has in-degree equal to its out-degree (in $R_{\mathrm{multi}}$); in addition, if $s,t\in V(Q)$, then $d^+(s)=d^-(s)+1$ and $d^-(t)=d^+(t)+1$ (in $R_{\mathrm{multi}}$). By Theorem \ref{prop:euler}, the following conditions are satisfied.
\begin{itemize}
\item If $s,t\in V(R)$, then there exists an Euler $(s,t)$-trail in $R_{\mathrm{multi}}$. Necessarily, this trail is an $r$-simple $(s,t)$-path that complies with $(Q,\varphi_Q)$ and visits $\psi(f_T(Q))$ at least once. Thus, $J_Q$ is a \yes-instance.
\item If $s,t\notin V(R)$, then there exists an Euler $(\psi(f_T(Q)),\psi(f_T(Q)))$-trail in $R_{\mathrm{multi}}$. Adding the arc $(\psi(f_T(Q)),t')$ creates an $r$-simple $(s',t')$-path that complies with $(Q,\varphi_Q)$ and visits $\psi(f_T(Q))$ at least once. Thus, $J_Q$ is a \yes-instance.
\end{itemize}
Thus, $C'$ is excellent w.r.t.~$(C,M,E,{\cal Q},f_T)$.

Second, suppose that $G$ has a cycle $C'$ that is good w.r.t.~$(C,M,E,{\cal Q},f_T)$. Let $\psi$ be a color-preserving isomorphism between $C'$ and $C$.
Then, for every component $Q\in{\cal Q}$, $J_Q$ is not a \no-instance, and hence the following conditions are satisfied.
\begin{itemize}
\item If $Q$ has vertices with the same colors as $s$ and $t$, then $G$ has an $r$-simple $(s,t)$-path $P^Q$ that weakly complies with $(Q,\varphi_Q)$ and visits $\psi(f_T(Q))$ at least once.
\item If $Q$ does not have vertices with the same colors as $s$ and $t$, then $G$ has an $r$-simple \newline $(\psi(f_T(Q)),\psi(f_T(Q)))$-path $P^Q$ that weakly complies with $(Q,\varphi_Q)$ (ignoring the final arc into $t'$) and visits $\psi(f_T(Q))$ at least once.
\end{itemize}
Let $C''$ be the directed multigraph obtained from $C'$ by duplicating each arc $M$ times. Consider the directed multigraph $H$ on vertex-set $V(H)=V(C'')\cup(\bigcup_{Q\in{\cal Q}}V(P^Q))$ and arc-multiset $A(H)=A(C'')\cup(\bigcup_{Q\in{\cal Q}}A(P^Q))$. (That is, every arc occurs in $H$ the number of times it occurs in $C''$ plus the sum over all $Q\in{\cal Q}$ of the number of times it occurs in $P^Q_{\mathrm{multi}}$.) Then, in $H$, we have that $d^+(s)=d^-(s)+1$ and $d^-(t)=d^+(t)+1$, and the out-degree and in-degree of any other vertex are equal. Moreover, the underlying undirected graph of $H$ is connected since the underlying undirected graph of each $P^Q_{\mathrm{multi}}$ is connected, and for any two distinct $Q,Q'\in{\cal Q}$, $C''$ has subpath from $\psi(f_T(Q))\in V(P^Q_{\mathrm{multi}})$ to $\psi(f_T(Q'))\in V(P^{Q'}_{\mathrm{multi}})$. By Theorem \ref{prop:euler}, this means that there exists an Euler $(s,t)$-trail in $H$. Necessarily, this trail is an $r$-simple $(s,t)$-path $P$ that weakly complies with $(T,\varphi)$ and visits $v_r$ at least once.
\end{proof}

We proceed to utilize the above lemmas in order to describe our recursive algorithm and prove its correctness.

\begin{lemma}\label{lem:solveTopologyProb}
\topologyProb\ can be solved in polynomial time, i.e.,~$(\ell+n+\log r)^{\OO(1)}$.
\end{lemma}

\begin{proof}
  Let $I=(G, \ell, r, s, t, (T, \varphi))$ be an instance of \topologyProb.  If $T$ is a DAG, then we solve $I$ using Lemma~\ref{lem:topologyBasis}.
  Otherwise, we decompose $(T,\varphi)$ into a hierarchy $H$ of rooted $r$-enriched $\ell$-topologies, where the base cases of the hierarchy correspond to DAGs. 
  We then proceed with bottom-up dynamic programming over $H$ to solve $I$. 
  To this end, define a tree $H$ as follows.  Let $s',t'$ be the endpoints of $(T,\varphi)$ and initialize $H$ as a tree with a single node $x$ whose label is $(T,\varphi,s')$. 
  Then recursively, for every leaf $x'$ of $H$ with a label $(T', \varphi', u)$ where $T'$ is not a DAG we decompose $(T',\varphi',u)$ further, as follows.
  \begin{enumerate}
  \item If possible, let $C$ be a simple directed cycle in $T'$ passing through $u$.\footnote{We can decide whether a digraph $D$ has a directed cycle through a vertex $u$ by adding a copy $u'$ of $u$ to $D$ and checking whether there is a directed path from from $u$ to $u'.$}
   Let $B=(C,M,E,\mathcal{Q}, f_T)$ be a relevant tuple. Then for every $Q \in \mathcal{Q}$ we create a child $x_Q$ of $x'$, and label $x_Q$ by the subtopology at $Q$ with respect to $B$. We refer to $x'$ as a \emph{cycle node (processing $B$)}. 
  \item If the previous case does not apply but $u$ is an endpoint of $(T',\varphi')$,
    then create a single child $x''$ of $x'$ by selecting an arbitrary new root $u' \in V(T')$ that is not an endpoint and giving $x''$ the label $(T',\varphi',u')$.
    We refer to $x'$ as a \emph{re-rooting node (away from $u$)}. 
  \item If no previous case applies, note that $(T',\varphi',u)$ is decomposable and
    let $B=(Q_c, s_c, t_c)_{c=1}^d$ be the strong component decomposition of $(T',\varphi',u)$. Create one child $x_c$ of $x'$ for each $c \in [d]$ and label $x_c$ by the $c$:th subtopology of $B$.  
    We refer to $x'$ as a \emph{path node (processing $B$)}.
  \end{enumerate}
  Let us first prove that $H$ is a tree of size polynomial in $\ell$. Say that a node $x'$ of $H$ is a DAG node if the topology $T'$ that $x'$ is labelled by is a DAG.
  For any node $x'$ of $H$, labelled by $(T',\varphi',u')$, let $A_c(x') \subseteq A(T')$ be those arcs that occur in a cycle in $T'$.  We argue the following property of sets $A_c(x')$ in $H$. Let $x'$ be a node of $H$, and let $S$ be the set of children of $x'$. Then (i) for any $x'' \in S$, we have $A_c(x'') \subseteq A_c(x')$, and (ii) if $a \in A_c(x'')$ for some $x'' \in S$, then the arc $a$ does not occur in any other child of $x'$. 

  We verify the property inductively by node type. For any DAG node, the property holds vacuously, and for a re-rooting node the property is trivial.  Assume next that $x'$ is a cycle node with some label $(T',\varphi',u')$, processing some tuple $B=(C,M,E,\mathcal{Q}, f_T)$.  Then for every arc $a$ of $A(T')$, either $a \in E$ or $a$ occurs in precisely one child of $x'$.  Hence (i) and (ii) are both clear.

  Finally, assume that $x'$ is a path node.  Then, since every cycle in $T'$ occurs in a strong component, the arcs of $A_c(x')$ are precisely partitioned by the non-DAG nodes of $S$.  Thus the property holds.

  We can now bound the size of $H$.   
  Let $X$ be the set of nodes $x'$ of $H$ such that $x'$ is not a DAG node, but every child of $x'$ in $H$ is a DAG node.  Then by definition $A_c(x') \neq \emptyset$ for every $x' \in X$.  Furthermore, the above properties imply that (i) $A_c(x') \subseteq A_c(x)$, and (ii) the sets $A_c(x')$ for $x' \in X$ are disjoint (since no node in $X$ is a descendant of another).  It follows that $|X| \leq |A_c(x)|$.
  Every node in $X$ has at most $|V(T)|$ leaves.  Furthermore, the height of $H$ is bounded by $\OO(|V(T)|+|A_c|)$, since at every step either $|A_C|$ decreases (in the case of a cycle node) or $|V(T')|$ decreases (in the case of a path node), and neither can increase. In particular, we only process path nodes if the root $u$ lies in a trivial strongly connected component, hence $|V(T')|$ decreases in this case.
  Hence $H$ has polynomial size in $\ell$.

  We now solve the problem via bottom-up dynamic programming over $H$, for each node tabulating possible choices in $G$ for the endpoints and root vertex of the topology. Concretely, let $x$ be a node of $H$ and let $(T',\varphi',u)$ be the label of $x$.  Let $s',t'$ be the endpoints of $(T',\varphi')$.  Then for $s_x, t_x, v_x \in V(G)$, define
  \[
    I_x(s_x,t_x,v_x) = (G, \ell, r, s_x, t_x, (T', \varphi'), v_x)
  \]
  as the instance corresponding to node $x$ where we have fixed a partial map $\psi(s')=s_x$, $\psi(t')=t_x$ and $\psi(u)=v_x$.  We show that using $H$, we can tabulate for every node $x$ whether $I_x(s_x,t_x,v_x)$ is a \no-instance or not.
  For simplicity, let us proceed bottom-up and accumulate a relation $R_x \subseteq V(G) \times V(G) \times V(G)$ for every node $x$, where $R_x(s_x,t_x,v_x)$ holds if and only if $I_x(s_x,t_x,v_x)$ was not detected to be a \no-instance. 
  Let us consider the node types of $H$ in turn.  Let $x$ be a node of $H$ with label $(T',\varphi',u)$.

  \emph{Case: $x$ is a leaf node.} In this case, $T'$ is a DAG and we can solve every instance $I_x$ using Lemma~\ref{lem:topologyBasis}.
  Thus we can assume that $R_x$ has been tabulated for every leaf of $H$.

  \emph{Case: $x$ is a re-rooting node.} In this case, by assumption $u$ is an endpoint of $(T',\varphi')$.  Let $s',t'$ be the endpoints of $(T',\varphi')$ and assume first that $u=s'$.
  Let $x'$ be the child of $x$ and let $(T',\varphi',u')$ be its label.  Then $(s_x,t_x,v_x) \in R_x$ if and only if $v_x=s_x$ and $(s_x,t_x,v_x') \in R_{x'}$ for some $v_x' \in V(G)$, which can clearly be checked in polynomial time.
  The case that $u=t'$ is symmetric.
  
  \emph{Case: $x$ is a cycle node.}  Let $x$ be a cycle node processing a tuple $B=(C,M,E,\mathcal{Q},f_T)$.
  Assume that we are deciding an instance $I_x(s_x,t_x,v_x)$, and immediately reject the instance unless $s_x$, $t_x$, $v_x$ share the colors of the endpoints $(T',\varphi')$ and the root $u$, respectively. Otherwise, by Lemma~\ref{lem:topologyStep} we need to decide whether $G$ contains a cycle that is good with respect to~$B$.  By Definition~\ref{def:goodCyc} we need to check for the existence of a cycle $C'$ in $G$ such that (i) there is a color-preserving isomorphism $\psi$ between $C$ and $C'$, and (ii) for every $Q \in \mathcal{Q}$ the instance $J_Q$ defined from $Q$ and $\psi$ is not a \no-instance.  For the latter, we note that the instance $J_Q$ is identical to one of the instances $I_{x_Q}(s_x',t_X',v_x')$ already tabulated in $R_{x_Q}$. 

  Indeed, first assume that $Q$ contains the endpoints of $(T',\varphi')$. In this case,
  we simply have    
  \[
    J_Q = I_{x_Q}(s_x,t_x,\psi(f_T(Q))),
  \]
  where $s_x$, $t_x$ are the vertices we are currently processing. Thus every such instance $J_Q$ has been tabulated.
  
  Next, assume that $Q$ does not contain the endpoints of $(T',\varphi')$ and let $v$ be the successor of $f_T(Q)$ in $C$.
  Then
  \[
    J_Q = I_{x_Q}(\psi(f_T(Q)), \psi(v), \psi(f_T(Q))),
  \]
  hence again every such instance $J_Q$ has been tabulated.
  
  Now, let $u=u_1, \ldots, u_d$ be the vertices of $C$ in $T$, reading in the forward direction starting from the root. 
  We decide whether $I_x(s_x,t_x,v_x)$ is a \no-instance via an auxiliary $d$-partite graph $H$ as follows.
  Let $V_1=\{v_x\}$ and for $i=2, \ldots, d$ let $V_i \subseteq V(G)$ be the set of vertices of $G$ having the same color as $u_i$. 
  Then the vertex set of $H$ is partitioned as $V(H)=V_1 \cup \ldots \cup V_d$, and the arcs of $H$ are the   
  candidate targets for arcs of $C$, i.e., for $v, v' \in V(H)$ where $v \in V_i$, we have $(v,v') \in A(H)$
  if and only if (i) $v' \in V_j$, where $j=i+1$ for $i<d$ and $j=1$ otherwise; and (ii) if $u_i=f_T(Q)$ for
  some $Q \in \mathcal{Q}$, then the instance $J_Q$ defined by $\psi(u_i)=v$ and $\psi(u_j)=v'$ is not a \no-instance.
  Note that by the above, the latter can be checked using only the identities of $v$ and $v'$
  and the DP table $R_{x_Q}$. 
  It now follows that $I_x(s_x,t_x,v_x)$ is not a \no-instance if and only if 
  $H$ contains a simple cycle, i.e., $H$ is not a DAG, which is easily checked \cite{DBLP:books/JBJGG}.
  Thus $R_x$ can be tabulated.
  
  \emph{Case: $x$ is a path node.} Finally, let $x$ be a path node processing the decomposition $B=(Q_c,s_c,t_c)_{c=1}^d$ of $(T',\varphi',u)$.  Assume that we are deciding an instance $I_x(s_x,t_x,v_x)$, and that the vertex colors of $s_x, t_x, v_x$ are consistent with the endpoints of $(T',\varphi')$ respectively $u$, otherwise the instance is negative.  Also note that a path node is only created if $u \in \{s_c,t_c: c \in [d]\}$.
  By Lemma~\ref{lem:strong-recurse}, we need to decide whether there exists a color-preserving map
  $\psi \colon \{s_c,t_c: c \in [d]\} \to V(G)$ such that the sequence $(\psi(s_c),\psi(t_c))_{c=1}^d$ is good with respect to $B$, i.e., by Definition~\ref{def:strong-subinstance}, whether (i) $(\psi(t_c),\psi(s_{c+1}))\in A(G)$ for every $c \in [d-1]$, (ii) the triple $(c,\psi(s_c),\psi(t_c))$ is valid, for every $c \in [d]$, and (iii) for every $c \in [d]$, the subinstance $I_{c,\psi(s_c),\psi(t_c)}$ is not a \no-instance. Furthermore, assuming condition (i) has been verified, the triple $(c,\psi(s_c),\psi(t_c))$ is valid for every $c \in [d]$ on the conditions that $\psi(s_1)=s_x$, $\psi(t_d)=t_x$ and $\psi(u)=v_x$.
  
  As in the previous case, we create an auxiliary graph $H$ to aid the search for $\psi$.  Define
  sets $U_c$, $V_c$ for each $c \in [d]$ where $U_c$ is the set of vertices of $V(G)$ sharing a color with $s_c$,
  and $V_c$ is a set of copies of the set of vertices of $V(G)$ sharing a color with $t_c$. (That is,
  if $s_c=t_c$ then the same vertices would be represented in sets $U_c$ and $V_c$ but treated distinctly in $H$.)
  Furthermore, delete vertices in $H$ so that $U_1=\{s_x\}$, $V_d=\{t_x\}$, and such that any set corresponding
  to the root $u$ only contains the single vertex $v_x$, and for every $c \in [d-1]$ create
  an arc $(v,v') \in V_c \times U_{c+1}$ if and only if the corresponding arc exists in $G$.
  
  We are now ready to identify the subinstances $I_{c,\psi(s_c),\psi(t_c)}$ among the previously tabulated instances $I_{x'}$.
  Let $c \in [d]$ and $(v,v') \in U_c \times V_c$, such that furthermore if $s_c=t_c$ then $v=v'$. 
  Assume that $(c,v,v')$ is a valid triple (otherwise, no arc $(v,v')$ will be added to $H$).
  First let $c<d$. Then $v'$ has an out-neighbor $v''$ in $H$ since $(c,v,v')$ is valid.
  Let $x_c$ be the child of $x$ corresponding to the $c$:th subtopology.
  Then (up to the choice of $v''$) we have
  \[
    I_{c,v,v'} = I_{x_c}(v,v'',v').
  \]
  For $c=d$, since $(c,v,v')$ is a valid triple there is an in-neighbor $v''$ of $v$ in $H$. 
  We now necessarily have $v'=t_x$, and again up to the choice of $v''$ we have
  \[
    I_{c,v,v'} = I_{x_c}(v'',v',v).
  \]
  Finally, regarding the choice of $v''$ it is easy to see from the construction that all choices
  create equivalent instances, thus we may select $v''$ arbitrarily. 
  We add an arc $(v,v')$ to $H$ if and only if the corresponding subinstance created this way
  is not a \no-instance according to $R_{x_c}$.
  
  It now follows that $I_x(s_x,t_x,v_x)$ is not a \no-instance if and only if
  the graph $H$ has a directed path from $s_x$ to $t_x$, thus $R_x$ can be tabulated.

  \emph{Wrapping up.}  Finally, let $x$ be the root node of $H$ and assume that $R_x$ has been tabulated as above.  To finish the computation, we simply check whether $R_x(s,t,s)$ holds, which by the above is equivalent to $I$ not being a \no-instance.
  The total running time of the procedure consists of at most $\OO(n^3)$ simple checks for every node of $H$, thus it takes polynomial time in total. 
\end{proof}

\subsection{Putting It All Together}

Finally, we are ready to conclude the correctness of Theorem~\ref{thm:directedPathFPT}.

\begin{proof}[Proof of Theorem \ref{thm:directedPathFPT}] By Lemma \ref{lem:solveTopologyProb}, \topologyProb\ can be solved in time and space $(\ell+n+\log r)^{\OO(1)}$. Thus, by Lemma \ref{lem:reduceToTopologyProb}, \diColPathRst\ can be solved in time $2^{\OO(\mathsf{b}(k/r)\log(\mathsf{b}(k/r)))}\cdot (n+\log k)^{\OO(1)}$ and polynomial space. Substituting $\mathsf{b}(k/r)$, this running time is upper bounded by $2^{\OO((k/r)^2\log(k/r))}\cdot (n+\log k)^{\OO(1)}$. In turn, by Lemma \ref{lem:color}, we have that \diPathRst\ on strongly connected digraphs can be solved in time $2^{\OO((k/r)^2\log(k/r))}\cdot (n+\log k)^{\OO(1)}$ and polynomial space. Finally, by Lemma \ref{lem:stronglyConnected}, we conclude that \diPathR\ can be solved in time $2^{\OO((k/r)^2\log(k/r))}\cdot (n+\log k)^{\OO(1)}$ and polynomial space.
\end{proof}


\section{Undirected $r$-Simple $k$-Path: Single-Exponential Time}\label{sec:undirectedPath}

In this section, we focus on the proof of the following theorem. As discussed in the introduction, for varied relations between $k$ and $r$, the running time in this theorem is {\em optimal} under the \ETH.

\begin{theorem}\label{thm:undirectedPathFPT}
\undiPathR\ is solvable in time $2^{\OO(\frac{k}{r})}(n+\log k)^{\OO(1)}$.
\end{theorem}

We will first show (in Sections \ref{sec:undirectedDistinct}--\ref{sec:mainLemmaProof}) how to prove the following result (which is the main part of our proof). 

\begin{lemma}\label{lem:rUnboundUndirectedPathFPT}
\undiPathR\ is solvable in time $2^{\OO(\frac{k}{r})}(r+n+\log k)^{\OO(1)}$.
\end{lemma}

Afterwards we will explain how to bound $r$. More precisely, let us refer to the special case of \undiPathR\ where $r>\sqrt{k}$ as the \undiPathRSpecial\ problem. 
\begin{defproblem}
{\undiPathRSpecial}
{An $n$-vertex undirected graph $G$ and positive integers $k,r$ such that $r>\sqrt{k}$.}
{Does $G$ have an $r$-simple $k$-path?}
\end{defproblem}\\
Then, we focus (in Section \ref{sec:boundR}) on the following result.

\begin{lemma}\label{lem:undirectedPathFPTspecial}
\undiPathRSpecial\ is solvable in time $2^{\OO(\frac{k}{r})}(n+\log k)^{\OO(1)}$.
\end{lemma}

Note that if $r\leq \sqrt{k}$, then $k/r = \Omega(\sqrt{k})$, in which case $r \leq \sqrt{k} \leq 2^{O(k/r)}$. Thus, Lemmas \ref{lem:rUnboundUndirectedPathFPT} and \ref{lem:undirectedPathFPTspecial} together imply Theorem \ref{thm:undirectedPathFPT}.
In this section, we require the following theorem instead of Theorem \ref{prop:euler}.

\begin{theorem}[\cite{DBLP:books/daglib/0030488}]\label{prop:eulerUndirected}
Let $G$ be a connected multigraph and let $s,t\in V(G)$.
\begin{itemize}
\item If $s\neq t$, then  $G$ has an Euler $(s,t)$-trail if and only if $d(s)$ and $d(t)$ are odd, and the degree of any other vertex in $G$ is even.
\item If $s=t$, then $G$ has an Euler $(s,t)$-trail if and only if the degree of every vertex in $G$ is even.
\end{itemize}
\end{theorem}


\subsection{Bounding the Number of Distinct Edges}\label{sec:undirectedDistinct}

This subsection is essentially a significantly simpler version of Sections \ref{sec:directedPathReduction},  \ref{sec:numDistinctArc} and \ref{sec:numDistinctArcTight}. For the sake of completeness, we give the sequence of adapted statements required to derive the bound on the number of distinct (i.e., non-parallel) edges stated at the end of this subsection.

Here, we say that an instance $(G,k,r)$ of \undiPathR\ is {\em nice} if $G$ has no path of length at least $k/r$. Observe that if an instance $(G,k,r)$ of \undiPathR\ is not nice, then it is necessarily a \yes-instance, since by traversing a path of length at least $k/r$ back and forth $r$ times, we obtain an $r$-simple $k$-path. Recall that by Theorem \ref{prop:longPath},  we can test the existence of a path of length at least $\ell$ from a vertex $s$ to a vertex $t$ in a digraph in time $2^{\OO(\ell)}n^{\OO(1)}$.
Clearly, we can utilize this algorithm to test the existence of a path of length at least $k/r$ in an undirected graph $G$: given an undirected graph $G$, let $\vec{G}$ be the graph obtain from $G$ by creating two opposing directed arcs from each edge, and run the algorithm with every choice of $s,t\in V(G)$. Thus, we have the following observation.

\begin{observation}\label{obs:longCycPathUndirected}
Given an instance $(G,k,r)$ of \undiPathR, it can be determined in time $2^{\OO(k/r)}\cdot n^{\OO(1)}$ and polynomial space whether $(G,k,r)$ is not nice, in which case it is a \yes-instance.
\end{observation}

Let us adapt Definition \ref{def:simpleMulti} to undirected graphs.

\begin{definition}\label{def:simpleMultiUndirected}
Let $P$ be an $r$-simple path in an undirected graph $G$.
\begin{itemize}
\item $P_{\mathrm{simple}}$ is the subgraph of $G$ on the vertices and edges visited at least once by $P$, and $P_{\mathsf{multi}}$ is the multigraph obtained from $P_{\mathrm{simple}}$ by duplicating each edge to occur the same number of times in $P_{\mathsf{multi}}$ and in $P$.
\item $V(P,r)=\{v\in V(G): v$ occurs $r$ times in $P\}$, and $P_{\mathrm{simple}}^{-r} = P_{\mathrm{simple}} - V(P,r)$.
\end{itemize}
\end{definition}

Recall that Corollary~\ref{cor:shortenPath} states that, for a digraph $\vec{G}$ with an $r$-simple $k'$-path $\vec{P}$ for some integer $k'\geq 2k$, it holds that $\vec{G}$ has an $r$-simple $k''$-path $\vec{Q}$, for some integer $k''\in\{k,k+1,\ldots,2k\}$, such that $\vec{Q}_{\mathrm{simple}}$ is a subgraph of $\vec{P}_{\mathrm{simple}}$ that is not equal to $\vec{P}_{\mathrm{simple}}$. This directly extends to undirected graphs. Thus, we have the following result.

\begin{lemma}\label{lem:shortenPathUndirected}
Let $(G,k,r)$ be a nice instance of \undiPathR. Let $P$ be an $r$-simple $k'$-path in $G$ for some integer $k'\geq 2k$. Then, $G$ has an $r$-simple $k''$-path $Q$, for some integer $k''\in\{k,k+1,\ldots,2k\}$, such that $Q_{\mathrm{simple}}$ is a subgraph of $P_{\mathrm{simple}}$ that is not equal to $P_{\mathrm{simple}}$.
\end{lemma}

Similarly, Lemma \ref{lem:Psimple} is adaptable to undirected graphs. Having Lemma \ref{lem:shortenPathUndirected} at hand, the arguments used to prove Lemma \ref{lem:Psimple} directly extend to prove the adaptation below as well, where one only has to view edges $\{u,v\}$ in $P_{\mathrm{simple}}$ as cycles $u-v-u$.

\begin{lemma}\label{lem:PsimpleUndirected}
Let $(G,k,r)$ be a nice \yes-instance of \undiPathR. Then, $G$ has an $r$-simple $k'$-path $P$, for some $k'\in\{k,k+1,\ldots,2k\}$, that satisfies the following two properties.
\begin{enumerate}
\item\label{item:Psimple1Undirected} $P_{\mathrm{simple}}^{-r}$ is edgeless.
\item\label{item:Psimple2Undirected} Every two distinct vertices in $V(P,r)$ have at most one common neighbor in $P_{\mathrm{simple}}$ that does not belong to $V(P,r)$. 
\end{enumerate}
\end{lemma}

We now finish the proof of the bound on the number of distinct
edges. Since the structure of undirected ($r$-simple) paths is 
significantly easier than directed ($r$-simple) paths, we are able to
do this in a single step, rather than the more complex proof used for
the directed case.

\begin{lemma}\label{lem:distinctUndirected}
Let $(G,k,r)$ be a nice \yes-instance of \undiPathR. Then, $G$ has an $r$-simple $k$-path with fewer than $30(k/r)$ distinct edges.
\end{lemma}

\begin{proof}
  We provide a proof sketch, since the details are similar to
Lemma~\ref{lem:improvedBound} but somewhat simpler. 
Let $P$ be an $r$-simple path chosen by the same conditions as in
Lemma~\ref{lem:improvedBound}, i.e., chosen to minimize the number of distinct edges used, to maximize $|V(P,r)|$, and which satisfies the properties in Lemma \ref{lem:PsimpleUndirected}. Let $s'$ and $t'$ be its
endpoints, and let $X=V(P,r) \cup \{s',t'\}$. 
Let $F=E(P_{\mathrm{simple}})$, partitioned as $F=F_1 \cup F_2$ 
where $F_1$ is the edge set of a tree that spans $X$. Then $|F_1| < 2|X|$
by Property~\ref{item:Psimple1Undirected} in Lemma \ref{lem:PsimpleUndirected}. Let $H$ be the graph with
edge set $F_2$. 

Let $C$ be a cycle in $H$ that is either of even length or contains at
least one vertex $v \notin X$ (or both), if such a cycle exists. We assign either the sign $-1$ or the sign $+1$ to each edge of $C$, so that for every vertex 
$v \in X \cap V(C)$, the edges incident with $v$ in $C$ have opposite
signs. Note that this is possible due to the conditions on $C$.  As in
Lemma~\ref{lem:improvedBound}, modifying the multiplicity in
$P_{\mathrm{multi}}$ of every edge in $C$ by $t \in \{\pm 1\}$ times the
sign of the edge creates a new graph with an Euler $(s',t')$-trail
that forms an $r$-simple path (in particular, the maximum degree in
$P_{\mathrm{multi}}$ is up to $2r$, and the degree of any vertex
$v \in V(C) \setminus X$ is at most $2r-2$ before the modification). 
As in Lemma~\ref{lem:improvedBound}, the existence of such a
modification contradicts our choice of $P$.
Thus we assume that every cycle in $H$ is of odd length and lies
entirely within $X$. 

We can now bound $|F_2|$. Consider first the multigraph $H'$ formed by
deleting all edges in $H[X]$. 
Then
$H'$ is a simple forest (since $H$ has no cycle that contains at least one vertex $v \notin X$), for which $X$ is a vertex cover due to Property~\ref{item:Psimple1Undirected} in Lemma~\ref{lem:PsimpleUndirected}; hence it
contains fewer than $2|X|$ edges. Furthermore, $H[X]$ itself is a
cactus graph, hence contains fewer than $(3/2)|X|$ edges (this is
folklore). Thus  
\[
|F| = |F_1| + |F_2| < 2|X| + 2|X| + (3/2)|X| = 5.5|X|.
\]
Hence the total number of distinct edges is less than
\[
5.5|X| \leq 5.5(|V(P,r)|+2) \leq 5.5(2k/r+2) < 30k/r
\]
as required. (We chose $30$ simply because is it a sufficiently large constant that is easier to work with than $5.5(2k/r+2)$.)
\end{proof}


\subsection{Partition into a Sparse Eulerian Multigraph and a Treewidth~$2$ Graph}

Having Lemma \ref{lem:distinctUndirected} at hand, we could have continued our analysis with simplified arguments of those presented for the directed case and thus obtain an algorithm that solves \undiPathR\ in time $2^{\OO(\frac{k}{r}\log(\frac{k}{r}))}(n+\log k)^{\OO(1)}$ and polynomial space. However, in order to obtain a single-exponential running time bound of $2^{\OO(\frac{k}{r})}(n+\log k)^{\OO(1)}$, we now take a very different route.

In this subsection, we gain a deeper understanding of the structure of a solution. The starting point for this understanding is the following lemma.

\begin{lemma}\label{lem:treeAndNoEven}
Let $(G,k,r)$ be a nice \yes-instance of \undiPathR. Then, $G$ has an $r$-simple $k$-path $P$ with fewer than $30(k/r)$ distinct edges, such that the edge multiset of $P_{\mathrm{multi}}$ can be partitioned into two multisets, $M_1$ and $M_2$, with the following properties:
\begin{itemize}
\item $P_{\mathrm{multi}}$ restricted to $M_1$ is a (simple) spanning tree of $P_{\mathrm{multi}}$, and
\item $P_{\mathrm{multi}}$ restricted to $M_2$ has no even cycle of length at least 4.
\end{itemize}
\end{lemma}

\begin{proof}
Let $e_1,e_2,\ldots,e_m$ be some ordering of the edges in $E(G)$. For any walk $W$, define $(x^W_1,x^W_2,$ $\ldots,x^W_m)$ be the vector where $x_i$ is equal to the number of times $e_i$ occurs in $W$ for all $i\in\{1,2,\ldots,m\}$.
By Lemma \ref{lem:distinctUndirected}, $G$ has an $r$-simple $k$-path with fewer than $30(k/r)$ distinct edges. Among all such $r$-simple $k$-paths, let $P$ be one where $(x^P_1,x^P_2,\ldots,x^P_m)$ is lexicographically smallest.
Let $T$ be an arbitrary spanning tree of $P_{\mathrm{multi}}$, and denote $M_1=E(T)$. In addition, denote $M_2=E(P_{\mathrm{multi}})\setminus M_1$. (Note that $M_2$ is a multiset: if an edge $e$ has $x$ copies in $P_{\mathrm{multi}}$, then it has either $x$ or $x-1$ copies in $M_2$.)
Let $H$ denote the restriction of $P_{\mathrm{multi}}$ to $M_2$.

We claim that $H$ has no even cycle of length at least 4. To prove this, suppose by way of contradiction that $H$ does have some even cycle $C$ of length at least $4$. Let $C=v_1-v_2-v_3-\cdots-v_q-v_1$ such that $e=\{v_1,v_2\}$ is the leftmost edge among the edges in $E(C)$ according to our predefined ordering of $E(G)$. Note that $q\geq 4$ is even. In addition, denote $U=\{\{v_i,v_{i+1}\}: i\in\{1,2,\ldots,q-1\}, i$ is odd$\}$. Now, define $H'$ as the graph obtained from $H$ by removing one copy of each edge in $U$ and adding one copy of each edge in $E(C)\setminus U$. Then, every vertex has the same degree in $H'$ and in $H$. Let $\widehat{H}$ denote the multigraph obtained by adding one copy of each edge in $M_1$ into $H'$. Then, every vertex has the same degree in $\widehat{H}$ and in $P_{\mathrm{multi}}$. Moreover, $M_1\subseteq E(\widehat{H})$ means that $\widehat{H}$ has a spanning tree and hence it is connected. Since $P_{\mathrm{multi}}$ has an Euler $(s,t)$-trail for some vertices $s,t\in V(G)$ (this trail is simply $P$), by Theorem \ref{prop:eulerUndirected}, $\widehat{H}$ also has an Euler $(s,t)$-trail, say $Q$. Then, $Q$ is an $r$-simple $k$-path with the same (or fewer) number of distinct edges as $P$. From our choice of $\{v_1,v_2\}$, it follows that $(x^Q_1,x^Q_2,\ldots,x^Q_m)$ is lexicographically smaller than $(x^P_1,x^P_2,\ldots,x^P_m)$. However, this contradict our choice of $P$.
\end{proof}

%

The usefulness in the second property in Lemma \ref{lem:treeAndNoEven} is primarily due to the following result.

\begin{proposition}[folklore, see \cite{DBLP:conf/wg/MisraRRS12,DBLP:journals/jgt/Thomassen88}]\label{prop:twNoEven}
The treewidth of a graph with no even cycle is at~most~$2$.
\end{proposition}
%

Having Proposition \ref{prop:twNoEven} at hand, we derive the following corollary to Lemma \ref{lem:treeAndNoEven}.

\begin{corollary}\label{cor:boundDistinctEdgeAndPartition}
Let $(G,k,r)$ be a nice \yes-instance of \undiPathR. Then, $G$ has an $r$-simple $k$-path $P$ with fewer than $30(k/r)$ distinct edges, such that the edge multiset of $P_{\mathrm{multi}}$ can be partitioned into two multisets, $M_1$ and $M_2$, with the following properties:
\begin{itemize}
\item $P_{\mathrm{multi}}$ restricted to $M_1$ is a (simple) spanning tree of $P_{\mathrm{multi}}$, and
\item $P_{\mathrm{multi}}$ restricted to $M_2$ is a multigraph of treewidth 2.
\end{itemize}
\end{corollary}

Corollary \ref{cor:boundDistinctEdgeAndPartition} partitions some solution into two parts: a spanning tree and a multigraph of low treewidth. However, for the DP approach considered later, we need the first part to have some Euler $(s,t)$-trail rather than just be a spanning tree. The reason for this is that the two parts will be computed somewhat independently. In particular, if some vertices of the first part will have odd degrees, our algorithm cannot ensure that each of these vertices will be reused an odd number of times (or even used at all) in the second part. We can guarantee that a ``color'' (for some vertex-coloring defined later) will be used in total an even number of times, but each part is ``oblivious'' to the identity of the vertices that ``realize'' this color in the other part. (The endpoints of the solution walk will be an exception to the above---since these are only two vertices, they can be guessed and thus handled easily.)

Before we proceed with our plan of having a new partition (based on the old one) of the edge multiset of a solution, we would like to make another remark. At this point, the reader may wonder if such a new partition is required, or whether we can bound the treewidth of the entire solution (for at least one solution) by a constant. However, it can be proven that for some instances, all solutions correspond to graphs with very high treewidth (in particular, of treewidth that cannot be bounded by a fixed constant). This is of course not a contradiction to Corollary \ref{cor:boundDistinctEdgeAndPartition} since even the composition of two graphs of treewidth 1 (say, trees) can be a graph of huge treewidth (e.g., a huge grid). For the sake of completeness, let us present a proof for this claim.

\begin{lemma}\label{lem:twLowerBound}
Let $r \geq 5$.  For any constant $c\in\mathbb{N}$, there exists a nice \yes-instance $(G,k,r)$ of \undiPathR\ such that every $r$-simple $k$-path $P$ in $G$ satisfies the following property: the treewidth of $P_{\mathrm{simple}}$ is larger than $c$.
\end{lemma}

\begin{proof}
Let $G$ be a $c \times c$ grid graph, with edges added to make a
4-regular graph (e.g., a grid embedded on a torus). 
We create a graph $G'$ by first replacing each edge $uv$ of $G$ by a path $uxyv$ on four vertices, where $x,y$ are new vertices, and
then adding a pendant vertex to every vertex (including
those vertices created by subdivision); see Fig.~\ref{fig:tw}. Let $W$ be the set of pendant
vertices, and $V'=V(G') \setminus W$. Let $P$ be an $r$-simple path on
$G'$ of maximum length. We will assume that $P$ is a closed walk as the 
other case can be treated similarly. Let $H$ be the Euler multigraph induced by
$P_{\rm multi}$ on $V'$. We show that $H$ is a (simple) graph and, moreover, $H=G'[V']$. 

\begin{figure}
\includegraphics[scale=0.75]{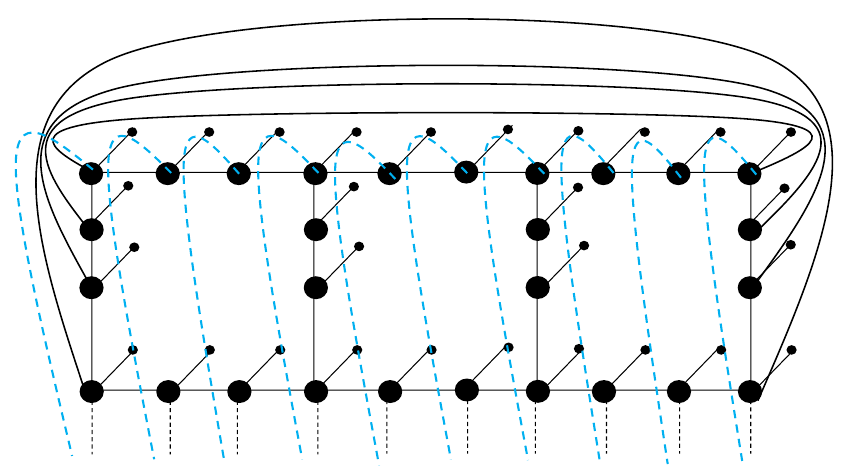}
\caption{The construction in the proof of Lemma \ref{lem:twLowerBound} with $c=4$.}
\label{fig:tw}
\end{figure}

First, observe that due to the pendant vertices, every vertex of $V'$
has precisely $r$ visits in $P$. Furthermore, since fewer visits to
$v$ in $H$ means more visits to the pendant vertex $v'$ of $v$, the
total number of visits to $v$ and $v'$ is $2r-d_H(v)$. Hence the total
number of visits of $P$ is
\begin{equation}\label{eq1}
\sum_{v \in V(H)} (2r-d_H(v)) = 2r|V(H)|-2|E(H)|.
\end{equation}
Observe that no edge $uv$ can be of multiplicity at least 3 in $H$ as otherwise by (\ref{eq1})
we could remove two copies of $uv$ from $P$ (and $H$) and add two copies of the edge between $u$ and its pendant as well as two copies of the edge between $v$ and its pendant, thereby increasing the size of $P$, a contradiction. 
It follows that $d_H(v) \leq 8$ for every $v \in V'$, since otherwise
some edge of $G'[V']$ has multiplicity at least 3 in $H$. 

Next, we argue that $H$ spans $V'$. 
Indeed, assume that there is an edge $uv \in E(G'[V'])$ where $\{u,v\}
\cap V(H)=\{u\}$. Since $r \geq 5$ and $d_H(u) \leq 8$, we may add two
copies of the edge $uv$ to $H$ and by (\ref{eq1}) raise the size of $P$ by
$2r-4 > 0.$

Now, finally, let $uv \in E(G)$, and let $uxyv$ be the corresponding
$P_4$ in $G'$. Since $P$ visits $x$ and $y$, and since $H$ is Euler, 
$P$ contains either the three edges $ux$, $xy$, $yv$ or at least
four edges, for example two copies each of $ux$ and $xy$.
Thus $|E(H)| \geq 3|E(G)|$, with equality only if the entire $P_4$ in
$G'$ is traversed for every edge $uv \in E(G)$. Hence, the longest
possible $r$-simple walk on $G'$ spans the entire grid, and therefore
$P_{\mathrm{simple}}$ has treewidth $c$. 
\end{proof}

Towards the proof of the new partition, we first give the following simple lemma.

\begin{lemma}\label{lem:eulerSpanning}
Let $G$ be a multigraph which has an Eulerian $(s,t)$-trail for some vertices $s,t\in V(G)$. Then, $G$ has a subgraph $H$ with the following properties:
\begin{itemize}
\item Every distinct edge in $G$ occurs at least once in $H$.
\item $H$ has an Eulerian $(s,t)$-trail.
\item $H$ has only at most $2d$ edges (including multiplicities), where $d$ is the number of distinct edges in $G$.
\end{itemize}
\end{lemma}

\begin{proof}
By Theorem \ref{prop:eulerUndirected}, $G$ is connected, each vertex in $V(G)\setminus\{s,t\}$ is of even degree, and either $s=t$ is of even degree, or $s\neq t$ are of odd degrees. For every edge $e$ reduce its multiplicity $\mu$ to 1 if $\mu$ is odd and 2 if $\mu$ is even. Let us denote the resulting multigraph by $H$. Clearly, $H$ is connected and every distinct edge in $G$ occurs at least once in $H$. Also, the number of edges of $H$ is at most $2d$, where $d$ is the number of distinct edges in $G$.
Since to obtain $H$ for every edge of $G$ we decreased its multiplicity by an even number (possibly, 0), each vertex of $H$ is of the same degree parity in $H$ and in $G$. Thus, by  Proposition \ref{prop:eulerUndirected}, $H$ has an Eulerian $(s,t)$-trail. 
\end{proof}

Having Lemma \ref{lem:eulerSpanning}, we can derive the following claim from Corollary \ref{cor:boundDistinctEdgeAndPartition}.

\begin{lemma}\label{lem:finalPartition}
Let $(G,k,r)$ be a nice \yes-instance of \undiPathR. Then, $G$ has an $r$-simple $k$-path $P$ with fewer than $30(k/r)$ distinct edges, such that the edge multiset of $P_{\mathrm{multi}}$ can be partitioned into two multisets, $M_1$ and $M_2$, with the following properties:
\begin{itemize}
\item $P_{\mathrm{multi}}$ restricted to $M_1$ is a spanning multigraph of $P_{\mathrm{multi}}$ with fewer than $60(k/r)$ edges (including multiplicities) that has an Eulerian $(s,t)$-trail where $s$ and $t$ are the end-vertices of $P$.
\item $P_{\mathrm{multi}}$ restricted to $M_2$ is a multigraph of treewidth 2.
\end{itemize}
\end{lemma}

\begin{proof}
Consider the decomposition of the edges of $P_{\mathrm{multi}}$ into $M_1$ and $M_2$ obtained in Corollary \ref{cor:boundDistinctEdgeAndPartition}. Since $P$ is an Eulerian $(s,t)$-trail of $P_{\mathrm{multi}}$, by Lemma \ref{lem:eulerSpanning}, $P_{\mathrm{multi}}$ has a subgraph $H$ such that every distinct edge in $P_{\mathrm{simple}}$ occurs at least once in $H$, $H$ has an Eulerian $(s,t)$-trail and $H$ has fewer than $60(k/r)$ edges (including multiplicities). Let $M_1'=E(H)$. Because each edge in $M_1$ has an occurrence in $M_1'$ and $M_1$ is a set, without loss of generality, we may assume that $M_1\subseteq M_1'$. Let $M_2'=E(P_{\mathrm{multi}})\setminus M_1'$. Then, $M_2'\subseteq M_2$. Therefore, since  $P_{\mathrm{multi}}$ restricted to $M_2$ is a multigraph of treewidth 2, so is $P_{\mathrm{multi}}$ restricted to $M_2'$.
\end{proof}


\subsection{Color Coding}

Knowing that it suffices for us to deal only with solutions having a small number of distinct vertices (in light of Lemma \ref{lem:finalPartition}), we utilize the method of color coding to focus on the following problem. Here, $\mathsf{b}(k/r)=30k/r+1$.

\begin{defproblem}
{\colPathRst}
{An $n$-vertex $\mathsf{b}(k/r)$-colored undirected graph $G$ and positive integers $k,r$.}
{Output \no\ if $G$ has no $r$-simple $k$-path, and \yes\ if it has a colorful $r$-simple $k$-path with fewer than $30(k/r)$ distinct edges.}
\end{defproblem}\\
With respect to this problem, when $G$ has no $r$-simple $k$-path, the input is called a \no-instance, and when $G$ has a colorful $r$-simple $k$-path with fewer than $30(k/r)$ distinct edges, the input is called a \yes-instance. The explicit requirement of having fewer than $30(k/r)$ distinct edges is meant only to simplify Section \ref{sec:boundR}. Notice that if the input is neither a \yes-instance nor a \no-instance, then the output can be arbitrary. 

The proof of the following lemma follows the lines of the proof of Lemma \ref{lem:color} where instead of Lemma \ref{lem:distinctImproved}, we use Lemma \ref{lem:finalPartition}, and hence it is not repeated here.

\begin{lemma}\label{lem:colorUndirected}
Suppose that \colPathRst\ can be solved in time $f(k/r)\cdot (n+\log k)^{\OO(1)}$. Then, \undiPathR\ can be solved in time $2^{\OO(k/r)}\cdot f(k/r)\cdot (n+\log k)^{\OO(1)}$.
\end{lemma}


\subsection{Guessing the Occurrence Sequence of the Spanning Multigraph Part}

We cannot guess the topology of the spanning multigraph part of a solution in a manner similar to guessing a topology as in the case of digraphs, since trying every possibility already takes times $2^{\OO(\frac{k}{r}\log\frac{k}{r})}$. Instead, inspired by the work of Berger et al.~\cite{DBLP:journals/corr/abs-1804-06361} (which guess a degree-sequence of a certain tree), we only guess a so called ``occurrence sequence'' of the spanning multigraph part of a solution. Let us first define a notion that we call an {\em occurrence sequence}.

\begin{definition}\label{def:degreeSeq}
Let $r,k\in\mathbb{N}$. An {\em $(r,k)$-occurrence sequence} is a tuple $\overline{\bf d}=(d_1,\ldots,d_{\mathsf{b}(k/r)})$ that satisfies the following conditions.
\begin{enumerate}
\item For all $i\in\{1,2,\ldots,\mathsf{b}(k/r)\}$, $d_i$ is an integer between $0$ and $r$. 
\item $\sum_{i=1}^{\mathsf{b}(k/r)}d_i \leq 2\mathsf{b}(k/r)$.
\end{enumerate}
Let ${\cal D}_{r,k}$ be the set of all $(r,k)$-occurrence sequences.
\end{definition}

We now show that the number of occurrence sequences is single-exponential.

\begin{lemma}\label{lem:degreeSeqNumber}
Let $r,k\in\mathbb{N}$. Then, $|{\cal D}_{r,k}|=2^{\OO(k/r)}$.
\end{lemma}

\begin{proof}
Let ${\cal D}_{s,r,k}$ be the set of tuples $\overline{\bf d}=(d_1,d_2,\ldots,d_{\mathsf{b}(k/r)})$ of non-negative integers that satisfy $\sum_{i=1}^{\mathsf{b}(k/r)}d_i = s$. Then, ${\cal D}_{r,k}=\bigcup_{s=0}^{2\mathsf{b}(k/r)}{\cal D}_{s,r,k}$. Thus, to prove that $|{\cal D}_{r,k}|=2^{\OO(k/r)}$, it suffices to show that for any $s\in\{0,1,\ldots,2\mathsf{b}(k/r)\}$, it holds that $|{\cal D}_{s,r,k}|=2^{\OO(k/r)}$. The total number of non-negative integral solutions to $\sum_{i=1}^{\mathsf{b}(k/r)}d_i=s$ can be found using the following well-known combinatorial reduction: consider $s$ identical balls placed in a row and set between them $\mathsf{b}(k/r)-1$ identical sticks (sticks may be placed before the first ball and after the last ball). Now the value of $d_i$ is the number of balls after the stick $i-1$ and before the stick $i$ (for $i=1$ and $i=\mathsf{b}(k/r)-1$, this term refers to the number of balls before the first stick and the number of balls after the last stick, respectively). Clearly, the number of placements of sticks is ${s+\mathsf{b}(k/r)-1 \choose s}<2^{s+\mathsf{b}(k/r)}=2^{\OO(k/r)}$.
\end{proof}

We now define what structures are good and comply with an occurrence sequence. Here, recall that a multigraph $H$ is called even if each of its connected components $C$ has an Euler $(s,t)$-trail with $s=t$ for some $s\in V(C)$. Equivalently (by Theorem \ref{prop:eulerUndirected}), every vertex in $H$ has even degree.
 
\begin{definition}\label{def:degreeSeqGood}
Let $r,k\in\mathbb{N}$. Let $G$ be a $\mathsf{b}(k/r)$-colored undirected graph. A pair $(W,H)$ of an $r$-simple path $W$ in $G$ and an even multigraph $H$ whose underlying simple graph is a subgraph of $G$ is {\em $q$-good} if the following conditions are satisfied.
\begin{enumerate}
\item\label{item:degreeSeqGood1} The treewidth of $H$ is at most $2$.
\item\label{item:degreeSeqGood2} Every connected component of $H$ has at least one vertex that is visited by $W$.
\item\label{item:degreeSeqGood3} The multigraph $H$ is colorful.
\item\label{item:degreeSeqGood4} The sum of the number of edges visited by $W$ and the number of edges (including multiplicities) of $H$ is $q-1$.
\end{enumerate}
If $q$ is not specified, then $q=k$.
\end{definition}

\begin{definition}\label{def:degreeSeqComply}
Let $r,k\in\mathbb{N}$. Let $G$ be a $\mathsf{b}(k/r)$-colored undirected graph, and let $\overline{\bf d},\overline{\bf d}'$ be $(r,k)$-occurrence sequences. A good pair $(W,H)$ {\em complies} with $\overline{\bf d}$ (resp.~$(\overline{\bf d},\overline{\bf d}')$) if for every color $i\in\{1,2,\ldots,\mathsf{b}(k/r)\}$, the following two conditions are satisfied.
\begin{enumerate}
\item\label{item:degreeSeqComply1} The number of times $W$ visits vertices colored $i$ is exactly $d_i$ (resp.~$d'_i$).
\item\label{item:degreeSeqComply2} The degree of any vertex colored $i$ in $H$ is at most $2(r-d_i)$.
\end{enumerate}
\end{definition}

Let us now argue that we can focus on seeking a pair $(W,H)$ as in Definition \ref{def:degreeSeqComply}.

\begin{lemma}\label{lem:equivalenceWithDegSeq}
Let $(G,k,r)$ be an instance of \colPathRst.
\begin{enumerate}
\item If $(G,k,r)$ is a \yes-instance, then there exist $\overline{\bf d}\in {\cal D}_{r,k}$ and a good pair that complies with $\overline{\bf d}$.
\item If there exist $\overline{\bf d}\in {\cal D}_{r,k}$ and a good pair that complies with $\overline{\bf d}$, then $(G,k,r)$ is not a \no-instance.
\end{enumerate}
\end{lemma}

\begin{proof}
{\bf First statement.} To prove the first statement, suppose that $(G,k,r)$ is a \yes-instance. That is, $G$ has a colorful $r$-simple $k$-path, say $P'$. Let $G'$ denote the subgraph of $G$ induced by the set of vertices visited by $P'$. Then, no two vertices in $G'$ have the same color, and $(G',k,r)$ is a \yes-instance (since $P'$ is a colorful $r$-simple $k$-path in $G'$). By Lemma \ref{lem:finalPartition}, $G'$ has an $r$-simple $k$-path $P$ such that the edge multiset of $P_{\mathrm{multi}}$ can be partitioned into two multisets, $M_1$ and $M_2$, with the following properties:
\begin{itemize}
\item $P_{\mathrm{multi}}$ restricted to $M_1$ is a spanning multigraph of $P_{\mathrm{multi}}$ with fewer than $60(k/r)$ edges (including multiplicities) that has an Eulerian $(s,t)$-trail $W$ where $s$ and $t$ are the end-vertices of $P$.
\item $P_{\mathrm{multi}}$ restricted to $M_2$ is a multigraph of treewidth 2.
\end{itemize}
Necessarily, $P$ is colorful. Let $H$ be the restriction of $P_{\mathrm{multi}}$ to $M_2$.
For all $i\in\{1,2,\ldots,\mathsf{b}(k,r)\}$, let $d_i$ denote the number of times $W$ visits the vertex colored $i$, and define $\overline{\bf d}=(d_1,d_2,\ldots,d_{\mathsf{b}(k,r)})$. We claim that $\overline{\bf d}\in {\cal D}_{r,k}$ and that $(W,H)$ is a good pair that complies with $\overline{\bf d}$.

Towards the proof of our claim, first note that since the size of $W$ is at most $60(k/r)$, it holds that $\sum_{i=1}^{\mathsf{b}(k/r)}d_i \leq 2\mathsf{b}(k/r)$. Moreover, since $W$ is an $r$-simple path (because it is a submultigraph of $P_{\mathrm{multi}}$), no vertex is visited by $W$ more than $r$ times, and since $W$ is colorful (because $P$ is colorful), this means that $d_i\leq r$ for all $i\in\{1,2,\ldots,\mathsf{b}(k,r)\}$. Thus, $\overline{\bf d}\in {\cal D}_{r,k}$. Moreover, the definition of $\overline{\bf d}$ directly ensures that Condition \ref{item:degreeSeqComply1} in Definition \ref{def:degreeSeqComply} is satisfied. In addition, since $P$ is an $r$-simple path, and since the number of times $P$ visits any vertex equals the number of times $W$ visits it plus half its degree in $H$, Condition \ref{item:degreeSeqComply2} in Definition \ref{def:degreeSeqComply} is satisfied as well.

It remains to show that the pair $(W,H)$ is good. Condition \ref{item:degreeSeqGood1} in Definition \ref{def:degreeSeqGood} follows directly from the assertion that $P_{\mathrm{multi}}$ restricted to $M_2$, which is precisely $H$, is a multigraph of treewidth 2. Since $P_{\mathrm{multi}}$ is a connected multigraph (since it has an Euler $(s,t)$-trail) and $W$ visits every vertex of $P_{\mathrm{multi}}$ at least once, it follows that every connected component of $H$ has at least one vertex that is visited by $W$. Thus, Condition \ref{item:degreeSeqGood2} in Definition \ref{def:degreeSeqGood} is satisfied as well. By Theorem \ref{prop:eulerUndirected}, because both $P_{\mathrm{multi}}$ and $P_{\mathrm{multi}}$ restricted to $M_1$ have Euler $(s,t)$-trails (where $s$ and $t$ are the end-vertices of $P$), every vertex has even degree in both $P_{\mathrm{multi}}$ and $P_{\mathrm{multi}}$ restricted to $M_1$, except for $s$ and $t$ if $s\neq t$---in this case, both $s$ and $t$ have odd degree in both $P_{\mathrm{multi}}$ and $P_{\mathrm{multi}}$ restricted to $M_1$. Thus, every vertex has even degree in $H$. Next, Condition \ref{item:degreeSeqGood3} in Definition \ref{def:degreeSeqGood} is satisfied because $H$ is colorful (since it is a submultigraph of $P_{\mathrm{multi}}$ which is colorful). Lastly, Condition \ref{item:degreeSeqGood4} in Definition \ref{def:degreeSeqGood} is satisfied because the sum of the number of edges visited by $W$ and the number of edges (including multiplicities) of $H$ is precisely the number of edge visits by $P$, which is $k-1$.

\medskip
\noindent{\bf Second statement.} To prove the second statement, suppose that there exist some $\overline{\bf d}\in {\cal D}_{r,k}$ and a good pair $(W,H)$ that complies with $\overline{\bf d}$. Let $R$ be the multigraph on $V(W_{\mathrm{simple}})\cup V(H)$ and edge multiset $E(W)\cup E(H)$ (that is, the number of times an edge occurs in $R$ is the sum of the number of times is occurs in $W$ and in $H$). By Condition \ref{item:degreeSeqGood2} in Definition \ref{def:degreeSeqGood}, $R$ is connected. Thus, since $W$ is a walk and $H$ is Eulerian, by Theorem \ref{prop:eulerUndirected}, $R$ admits an Euler $(s,t)$-trail $P$ where $s$ and $t$ are the end-vertices of $W$. By Definition \ref{def:degreeSeqComply}, the degree of every vertex in $R$ is at most $2r$, and hence $P$ must be an $r$-simple path. In addition, from Condition \ref{item:degreeSeqGood4} in Definition \ref{def:degreeSeqGood} it follows that the size of $P$ is $k$. Thus, $(G,k,r)$ is not a \no-instance.
\end{proof}

Accordingly, we define the following problem. 

\begin{defproblem}
{\eulerPart}
{An $n$-vertex $\mathsf{b}(k/r)$-colored undirected graph $G$, positive integers $k,r$, and $\overline{\bf d}\in {\cal D}_{r,k}$.}
{Does there exist a good pair that complies with $\overline{\bf d}$?}
\end{defproblem}\\
Let us now state that we can focus on solving the \eulerPart\ problem.

\begin{lemma}\label{lem:pairProblem}
Suppose that \eulerPart\ can be solved in time $f(k/r)\cdot (r+n+\log k)^{\OO(1)}$. Then, \colPathRst\ can be solved in time $2^{\OO(k/r)}\cdot f(k/r)\cdot (r+n+\log k)^{\OO(1)}$.
\end{lemma}

\begin{proof}
Let $\cal A$ be an algorithm that solves \eulerPart\ in time $f(k/r)\cdot (r+n+\log k)^{\OO(1)}$. In what follows, we describe how to solve \colPathRst. To this end, let $(G,k,r)$ be an instance of \colPathRst. For each $\overline{\bf d}\in {\cal D}_{r,k}$, we call ${\cal A}$ with $(G,k,r,\overline{\bf d})$ as input, and if $\cal A$ returns \yes, so do we. At the end, if no call to $\cal A$ returned \yes, we return \no.

The correctness of our algorithm directly follows from Lemma \ref{lem:equivalenceWithDegSeq}. Now, note that by Lemma \ref{lem:degreeSeqNumber}, $|{\cal D}_{r,k}|=2^{\OO(k/r)}$. Thus, it is clear that our algorithm runs in time $2^{\OO(k/r)}\cdot f(k/r)\cdot (r+n+\log k)^{\OO(1)}$.
\end{proof}


\subsection{Two-Level Dynamic Programming (DP)}

We first give a lemma that handles a single connected component of the treewidth-2 multigraph $H$ that is a member of the pair we aim to find.

\begin{lemma}\label{lem:tw2CompDP}
There exists an $2^{\OO(k/r)}\cdot (r+n+\log k)^{\OO(1)}$-time algorithm that, given an undirected graph $G$ where every vertex is assigned a color from $\{1,2,\ldots,\mathsf{b}(k/r)\}$, a set of colors $C\subseteq\{1,2,\ldots,\mathsf{b}(k/r)\}$, a vertex $v^\star\in V(G)$ whose color belongs to $C$, and $\overline{\bf d}=(d_1,\ldots,d_{\mathsf{b}(k/r)})\in {\cal D}_{r,k}$, outputs the largest integer $M$ for which there exists a colorful multigraph $H$ that satisfies the following conditions.
\begin{enumerate}
\item\label{item:tw2CompDP1} For each $v\in V(H)$, the degree of $v$ in $H$ is even and does not exceed $2(r-d_i)$ where $i$ is the color of $v$.
\item\label{item:tw2CompDP2} The underlying simple graph of $H$ is a {\em connected} subgraph of $G$.
\item\label{item:tw2CompDP3} The treewidth of $H$ is at most $2$.
\item\label{item:tw2CompDP4} $v^\star\in V(H)$.
\item\label{item:tw2CompDP5} The number of edges (including multiplicities) in $H$ is exactly $M.$
\item\label{item:tw2CompDP6} Every vertex in $H$ is colored by a color from $C.$
\end{enumerate}
\end{lemma}

\begin{proof}
First, we remove all vertices in $G$ whose color does not belong to $C$. In addition, we remove all vertices in $G$ whose color is the same as the color of  $v^\star$ but which are not $v^\star$. For the sake of simplicity, abuse notation and call the resulting graph $G$ as well. Thus, we can now ignore Condition \ref{item:tw2CompDP6} since it will be automatically satisfied. The proof is based on a standard DP over a tree decomposition (see, e.g., \cite{DBLP:books/sp/CyganFKLMPPS15}) with a slight technicality: we do not know the structure of $H$ and hence we do not have the tree decomposition over which the DP should be performed. Nevertheless, we can repeatedly ``guess'' the current top bag and hence imitate a standard DP over an (unknown) tree decomposition. We remark that DPs over so-called hidden tree decompositions are a well-known tool to design subexponential-time algorithms for parameterized problems in Computational Geometry (see, e.g., \cite{DBLP:conf/soda/KleinM14,DBLP:conf/compgeom/AshokFKSZ17}).

We use a DP table $\mathsf{N}$ with an entry $\mathsf{N}[U,C',f_{\mathrm{deg}},g_{\mathrm{edg}},{\cal S}]$ for all $U\subseteq V(G)$ of size at most $3$, $C'\subseteq C$ that contains the colors of the vertices in $U$, $f_{\mathrm{deg}}: U\rightarrow \{0,1,\ldots,2r\}$ such that $f(v)$ does not exceed $2(r-d_i)$ where $i$ is the color of $v$ for all $v\in U$, $g_{\mathrm{edg}}:\{\{u,v\}\in E(G): u,v\in U\}\rightarrow \{0,1,\ldots,2r\}$, and partition ${\cal S}$ of $U$. The purpose of an entry $\mathsf{N}[U,C',f_{\mathrm{deg}},g_{\mathrm{edg}},{\cal S}]$ is to store the largest integer $M$ for which there exists a colorful multigraph $H$ with a nice tree decomposition $(T,\beta)$ that satisfy the following conditions.
\begin{enumerate}
\item\label{item:purpose1} For each $v\in V(H)\setminus U$, the degree of $v$ in $H$ is even and does not exceed $2(r-d_i)$ where $i$ is the color of $v$. For each $v\in U$, the degree of $v$ in $H$ is $f_{\mathrm{deg}}(v)$.
\item\label{item:purpose2} For each $\{u,v\}\in E(G)$ with $u,v\in U$, the multiplicity of $\{u,v\}$ in $H$ is $g_{\mathrm{edg}}(\{u,v\})$.
\item\label{item:purpose3} The underlying simple graph of $H$ is a subgraph of $G$. In addition, for all $u,v\in U$, it holds that $u,v$ belong to the same connected component of $H$ if and only if $u,v$ belong to the same part in $\cal S$. Furthermore, every connected component of $H$ contains a vertex from $U$.
\item\label{item:purpose4} The width of $(T,\beta)$ is at most $2$, and $\beta(r)=U$ for the root $r$ of $T$.
\item\label{item:purpose5} The number of edges (including multiplicities) in $H$ is exactly $M$.
\item\label{item:purpose6} The set of colors of the vertices in $V(H)$ is precisely $C'$.
\end{enumerate}
Having computed $\mathsf{N}$ correctly, the final output is the maximum value stored in $\mathsf{N}[U,C',f_{\mathrm{deg}},g_{\mathrm{edg}},$ ${\cal S}]$ over all $U\subseteq V(G)$ of size at most $3$, $C'\subseteq C$ such that $C'$ contains the color of $v^\star$, $f_{\mathrm{deg}}: U\rightarrow \{0,1,\ldots,2r\}$ such that $f_{\mathrm{deg}}(v)$ is even for all $v\in U$, $g_{\mathrm{edg}}:\{\{u,v\}\in E(G): u,v\in U\}\rightarrow \{0,1,\ldots,2r\}$, and partition ${\cal S}=\{U\}$.
In what follows, we show how to correctly compute $\mathsf{N}$ in time $2^{\OO(k/r)}\cdot (r+n+\log k)^{\OO(1)}$. Here, every entry $\mathsf{N}[U,C',f_{\mathrm{deg}},g_{\mathrm{edg}},{\cal S}]$ should be computed before all entries $\mathsf{N}[\widehat{U},\widehat{C},\widehat{f}_{\mathrm{deg}},\widehat{g}_{\mathrm{edg}},\widehat{\cal S}]$ such that either $|C'|<|\widehat{C}|$ or both $|C'|=|\widehat{C}|$ and $|\widehat{U}|<|U|$.

\medskip
\noindent{\bf Basis.} If $U=\emptyset$ or $C'=\emptyset$, then $\mathsf{N}[U,C',f_{\mathrm{deg}},g_{\mathrm{edg}},{\cal S}]=0$ if both $U=\emptyset$ and $C'=\emptyset$, and $\mathsf{N}[U,C',f_{\mathrm{deg}},g_{\mathrm{edg}},{\cal S}]=-\infty$ otherwise. Moreover, if the colors of the vertices in $U$ are not contained in $C'$, then also $\mathsf{N}[U,C',f_{\mathrm{deg}},g_{\mathrm{edg}},{\cal S}]=-\infty$.

\medskip
\noindent{\bf Step.} Let $\mathsf{N}[U,C',f_{\mathrm{deg}},g_{\mathrm{edg}},{\cal S}]$ be an entry such that both $U\neq\emptyset$ and $C'$ contains the colors of the vertices in $U$.\\ Then, $\mathsf{N}[U,C',f_{\mathrm{deg}},g_{\mathrm{edg}},{\cal S}]=\max\{F,I,J\}$, where $F,I$ and $J$ are computed as follows. 
	\begin{itemize}
	\item {\bf Forget:} If $|U|=3$, then $F=-\infty$.
	
	Else, $F$ is the maximum of the integers in $\mathsf{N}[U\cup\{v\},C',\widehat{f},\widehat{g},\widehat{\cal S}]$ over each vertex $v\in V(G)\setminus U$, each function $\widehat{f}:U\cup\{v\}\rightarrow \{0,1,\ldots,2r\}$ that extends $f_{\mathrm{deg}}$ so that $\widehat{f}(v)$ is even and does not exceed $2(r-d_i)$ where $i$ is the color of $v$, each function $\widehat{g}:\{\{u,v\}\in E(G): u,v\in U\}\rightarrow \{0,1,\ldots,2r\}$ that extends $g_{\mathrm{edg}}$, and each partition $\widehat{\cal S}$ of $U\cup\{v\}$ such that the part that contains $v$ is of size at least $2$ and if $|U|=2$, then the two vertices in $U$ are in the same part in $\cal S$ if and only if they are in the same part in $\widehat{\cal S}$.
		
	\item {\bf Introduce:} $I$ is the maximum of the sums $\mathsf{N}[U\setminus\{\widetilde{v}\},C'\setminus\{i\},\widehat{f},\widehat{g},\widehat{\cal S}]+f_{\mathrm{deg}}(\widetilde{v})$ over each vertex $\widetilde{v}\in U$, where $i$ is the color of $\widetilde{v}$, each function $\widehat{f}:U\setminus\{\widetilde{v}\}\rightarrow \{0,1,\ldots,2r\}$, where $\widehat{g}$ is the restriction of $g_{\mathrm{edg}}$ to $\{\{u,v\}\in E(G): u,v\in U\setminus\{\widetilde{v}\}\}$, and each partition $\widehat{\cal S}$ of $U\setminus\{\widetilde{v}\}$ that altogether satisfy the following conditions.
	\begin{enumerate}
	\item $f_{\mathrm{deg}}(\widetilde{v})=\sum_{u\in U: \{u,\widetilde{v}\}\in E(G)}g_{\mathrm{edg}}(\{u,\widetilde{v}\})$.
	\item\label{item:fhatIntroduce} For each $u\in U$, $f_{\mathrm{deg}}(u)=\widehat{f}(u)+b$ where $b=\widehat{g}(\{u,\widetilde{v}\})$ if $\{u,\widetilde{v}\}\in E(G)$ and $b=0$ otherwise.
	\item If $|U|=3$ and the two vertices in $U \setminus \{\widetilde{v}\}$ are in the same part in $\widehat{\cal S}$, then they are also in the same part in ${\cal S}$.
	\item For each $u\in U$ with $\{u,\widetilde{v}\}\in E(G)$ and $g_{\mathrm{edg}}(\{u,\widetilde{v}\})\geq 1$, $u,\widetilde{v}$ are in the same part in ${\cal S}$.
	\item For each $u\in U$ with either $\{u,\widetilde{v}\}\notin E(G)$ or $g_{\mathrm{edg}}(\{u,\widetilde{v}\})=0$, if $u,\widetilde{v}$ are in the same part in ${\cal S}$, then there exists $w\in U\setminus\{u,\widetilde{v}\}$ with $\{w,\widetilde{v}\}\in E(G)$, $\widehat{g}(\{w,\widetilde{v}\})\geq 1$ and such that $u,w$ are in the same part in $\widehat{\cal S}$.
	\end{enumerate}
	(If there exists no entry $\mathsf{N}[U\setminus\{v\},C'\setminus\{i\},\widehat{f},\widehat{g},\widehat{\cal S}]$ that satisfies the conditions above, then $I=-\infty$.)
	
	\item {\bf Join:} $J$ is the maximum of the sums \[\mathsf{N}[U,C_1,f_1,g_{\mathrm{edg}},{\cal S}_1]+\mathsf{N}[U,C_2,f_2,g_{\mathrm{edg}},{\cal S}_2]-\sum_{\{u,v\}\in E(G): u,v\in U}g_{\mathrm{edg}}(\{u,v\})\]
	over each $C_1\subseteq C'$ that contains $X$ being the set of colors of the vertices in $U$ as well as at least one other color, and which is not equal to $C'$, where $C_2 = C'\setminus (C_1\setminus X)$, each pair of functions $f_1:U\rightarrow \{0,1,\ldots,2r\}$ and $f_2:U\rightarrow \{0,1,\ldots,2r\}$ such that $f_{\mathrm{deg}}(v)=f_1(v)+f_2(v)-\sum_{\{u,v\}\in E(G): u\in U}g_{\mathrm{edg}}(\{u,v\})$ for every $v\in U$, and each pair of partitions ${\cal S}_1$ and ${\cal S}_2$ of $U$ such that for all $u,v\in U$, $u,v$ are in the same part in ${\cal S}$ if and only if $u,v$ are in the same part in the finest common coarsening of ${\cal S}_1$ and ${\cal S}_2$ (i.e., since $|U|\leq 3$, either $u$ and $v$ are in the same part in $\mathcal{S}_1$ or $\mathcal{S}_2$, or there exists $w\in U\setminus\{u,v\}$ such that $u$ and $w$ are in the same part in $\mathcal{S}_i$ and $w$ and $v$ in $\mathcal{S}_{3-i}$ for some $i \in \{1,2\}$).
	\end{itemize}
We now analyze the running time and prove the correctness of the algorithm.

\medskip
\noindent{\bf Time Complexity.} Notice that there are only $2^{\OO(k/r)}\cdot (r+n+\log k)^{\OO(1)}$ entries $\mathsf{N}[U,C',f_{\mathrm{deg}},g_{\mathrm{edg}},{\cal S}]$ in $\mathsf{N}$---indeed, there are $1+n+{n\choose 2}+{n \choose 3}=n^{\OO(1)}$ choices for $U$, $2^{|C|}\leq 2^{\mathsf{b}(k/r)}=2^{\OO(k/r)}$ choices for $C'$, at most $(2r+1)^3$ choices for $f_{\mathrm{deg}}$ given $U$, at most $(2r+1)^{{3 \choose 2}}=r^{\OO(1)}$ choices for $g_{\mathrm{edg}}$ given $U$, and $\OO(1)$ choices for $\cal S$ given $U$. We now claim that each entry in $\mathsf{N}$ is computed in time $2^{\OO(k/r)}\cdot (r+n+\log k)^{\OO(1)}$.  Indeed, each entry in the basis is computed in $\OO(1)$ time. For the step, the Forget case is computed in time $\OO(n\cdot r\cdot r^2)$ as there are $\OO(n)$ choices for $v$, $\OO(r)$ choices for $\widehat{f}(v)$, $\OO(r^2)$ choices how $g_{\mathrm{edg}}$ is extended by $\widehat{g}$, and $\OO(1)$ choices for $\cal S$. The Introduce case is computed in time $\OO(1)$ as there are $\OO(1)$ choices for $\widetilde{v}$, only one choice for $\widehat{f}$ due to Condition \ref{item:fhatIntroduce},  and $\OO(1)$ choices for $\widehat{\cal S}$. The Join case is computed in time $2^{\OO(k/r)}\cdot r^3$ as there are at most $2^{|C|}=\leq 2^{\mathsf{b}(k/r)}=2^{\OO(k/r)}$ choices for $C_1$ and $C_2$, at most $\OO(r^3)$ choices for $f_1$ and only one choice for $f_2$ given $f_1$, and at most $\OO(1)$ choices for ${\cal S}_1$ and ${\cal S}_2$. Therefore, the total running time is $2^{\OO(k/r)}\cdot (r+n+\log k)^{\OO(1)}$. 

\medskip
\noindent{\bf Correctness.}
By using Theorem \ref{prop:niceTreeDecomp}, the correctness of the computation can be proved by standard induction on the structure of the recursion. For the sake of completeness, we give the proof that each entry $\mathsf{N}[U,C',f_{\mathrm{deg}},g_{\mathrm{edg}},{\cal S}]$ stores the integer $M$ in the definition of the purpose of this entry, which we denote by $M_{U,C',f_{\mathrm{deg}},g_{\mathrm{edg}},{\cal S}}$. Let us denote the collection of pairs $(H,(T,\beta))$ that satisfy the six conditions except Condition \ref{item:purpose5} given in the definition of the purpose of this entry by ${\cal H}_{U,C',f_{\mathrm{deg}},g_{\mathrm{edg}},{\cal S}}$.  The induction is on $|C'|$ where the basis also includes the other cases in the Basis of our algorithm.

In the Basis, when $U=\emptyset$, due to Condition \ref{item:purpose3} (specifically, that every connected component of $H$ contains a vertex from $U$), pairs in ${\cal H}_{U,C',f_{\mathrm{deg}},g_{\mathrm{edg}},{\cal S}}$ must correspond to empty graphs. Thus, if $C'\neq\emptyset$, there is no triple in this collection (due to Condition \ref{item:purpose6}), and otherwise there is exactly one where the number of edges is $0$. For the same reason (Conditions \ref{item:purpose3} and \ref{item:purpose6}), when $C'=\emptyset$, it must be that $U=\emptyset$ so that ${\cal H}_{U,C',f_{\mathrm{deg}},g_{\mathrm{edg}},{\cal S}}$ will not be empty, and when the colors of the vertices in $U$ are not present in $C'$, ${\cal H}_{U,C',f_{\mathrm{deg}},g_{\mathrm{edg}},{\cal S}}$ is empty. This completed the correctness of the Basis.

Now, let us prove correctness for $\mathsf{N}[U,C',f_{\mathrm{deg}},g_{\mathrm{edg}},{\cal S}]$ that does not belong to the Basis, under the assumption of correctness for all entries with a second argument (color set) of size smaller than $|C'|$. Let ${\cal H}^{\mathsf{F}}_{U,C',f_{\mathrm{deg}},g_{\mathrm{edg}},{\cal S}}$, ${\cal H}^{\mathsf{I}}_{U,C',f_{\mathrm{deg}},g_{\mathrm{edg}},{\cal S}}$ and ${\cal H}^{\mathsf{J}}_{U,C',f_{\mathrm{deg}},g_{\mathrm{edg}},{\cal S}}$ denote the subcollections of pairs in ${\cal H}_{U,C',f_{\mathrm{deg}},g_{\mathrm{edg}},{\cal S}}$ where the root of the tree decomposition is a forget node, an introduce node and a join node, respectively. Because $U\neq \emptyset$, no pair in ${\cal H}_{U,C',f_{\mathrm{deg}},g_{\mathrm{edg}},{\cal S}}$ has a root node that is a leaf node. From this, we get that ${\cal H}_{U,C',f_{\mathrm{deg}},g_{\mathrm{edg}},{\cal S}}={\cal H}^{\mathsf{F}}_{U,C',f_{\mathrm{deg}},g_{\mathrm{edg}},{\cal S}}\cup {\cal H}^{\mathsf{I}}_{U,C',f_{\mathrm{deg}},g_{\mathrm{edg}},{\cal S}}\cup {\cal H}^{\mathsf{J}}_{U,C',f_{\mathrm{deg}},g_{\mathrm{edg}},{\cal S}}$. Let $M^{\mathsf{F}}_{U,C',f_{\mathrm{deg}},g_{\mathrm{edg}},{\cal S}}$, $M^{\mathsf{F}}_{U,C',f_{\mathrm{deg}},g_{\mathrm{edg}},{\cal S}}$ and $M^{\mathsf{J}}_{U,C',f_{\mathrm{deg}},g_{\mathrm{edg}},{\cal S}}$ denote the maximum number of edges in a graph $H$ of a pair in the subcollections ${\cal H}^{\mathsf{F}}_{U,C',f_{\mathrm{deg}},g_{\mathrm{edg}},{\cal S}}$, ${\cal H}^{\mathsf{I}}_{U,C',f_{\mathrm{deg}},g_{\mathrm{edg}},{\cal S}}$ and ${\cal H}^{\mathsf{J}}_{U,C',f_{\mathrm{deg}},g_{\mathrm{edg}},{\cal S}}$. Then, due to Condition \ref{item:purpose5} and since ${\cal H}_{U,C',f_{\mathrm{deg}},g_{\mathrm{edg}},{\cal S}}={\cal H}^{\mathsf{F}}_{U,C',f_{\mathrm{deg}},g_{\mathrm{edg}},{\cal S}}\cup {\cal H}^{\mathsf{I}}_{U,C',f_{\mathrm{deg}},g_{\mathrm{edg}},{\cal S}}\cup {\cal H}^{\mathsf{J}}_{U,C',f_{\mathrm{deg}},g_{\mathrm{edg}},{\cal S}}$, we derive that $M_{U,C',f_{\mathrm{deg}},g_{\mathrm{edg}},{\cal S}}=\max\{M^{\mathsf{F}}_{U,C',f_{\mathrm{deg}},g_{\mathrm{edg}},{\cal S}}, M^{\mathsf{F}}_{U,C',f_{\mathrm{deg}},g_{\mathrm{edg}},{\cal S}},M^{\mathsf{J}}_{U,C',f_{\mathrm{deg}},g_{\mathrm{edg}},{\cal S}}\}$. Therefore, because the algorithm computes $\mathsf{N}[U,C',f_{\mathrm{deg}},g_{\mathrm{edg}},{\cal S}]=\max\{F,I,J\}$, it suffices to prove that $F=M^{\mathsf{F}}_{U,C',f_{\mathrm{deg}},g_{\mathrm{edg}},{\cal S}}$, $I=M^{\mathsf{I}}_{U,C',f_{\mathrm{deg}},g_{\mathrm{edg}},{\cal S}}$ and $J=M^{\mathsf{J}}_{U,C',f_{\mathrm{deg}},g_{\mathrm{edg}},{\cal S}}$. We consider each of these cases separately below.

\medskip
\noindent{\bf Proof for Forget Case:} In case $|U|=3$, then ${\cal H}^{\mathsf{F}}_{U,C',f_{\mathrm{deg}},g_{\mathrm{edg}},{\cal S}}=\emptyset$ because any pair $(H,(T,\beta))$ where the root is a forget node whose bag is of size $3$ must have as a child a node whose bag is of size $4$, which implies width at least $3$, and hence cannot belong to ${\cal H}^{\mathsf{F}}_{U,C',f_{\mathrm{deg}},g_{\mathrm{edg}},{\cal S}}$ (due to Condition \ref{item:purpose4}). Thus, in this case, $M^{\mathsf{F}}_{U,C',f_{\mathrm{deg}},g_{\mathrm{edg}},{\cal S}}=-\infty$, which is the value assigned to $F$.

Now, suppose that $|U|\leq 2$. By the computation performed by the algorithm and the inductive hypothesis, it suffices to prove that $M^{\mathsf{F}}_{U,C',f_{\mathrm{deg}},g_{\mathrm{edg}},{\cal S}}$ is equal to the maximum among the integers $M_{U\cup\{v\},C',\widehat{f},\widehat{g},\widehat{\cal S}}$ over each vertex $v\in V(G)\setminus U$, each function $\widehat{f}:U\cup\{v\}\rightarrow \{0,1,\ldots,2r\}$ that extends $f_{\mathrm{deg}}$ so that $\widehat{f}(v)$ is even and does not exceed $2(r-d_i)$ where $i$ is the color of $v$, each function $\widehat{g}:\{\{u,v\}\in E(G): u,v\in U\}\rightarrow \{0,1,\ldots,2r\}$ that extends $g_{\mathrm{edg}}$, and each partition $\widehat{\cal S}$ of $U\cup\{v\}$ such that the part that contains $v$ is of size at least $2$ and if $|U|=2$, then the two vertices in $U$ are in the same part in $\cal S$ if and only if they are in the same part in $\widehat{\cal S}$. 

In one direction, to show that $M^{\mathsf{F}}_{U,C',f_{\mathrm{deg}},g_{\mathrm{edg}},{\cal S}}$ is equal or smaller than this maximum, consider a pair $(H,(T,\beta))\in {\cal H}^{\mathsf{F}}_{U,C',f_{\mathrm{deg}},g_{\mathrm{edg}},{\cal S}}$ where the number of edges of $H$ equals $M^{\mathsf{F}}_{U,C',f_{\mathrm{deg}},g_{\mathrm{edg}},{\cal S}}$.  Let $(T',\beta')$ be the tree decomposition of $H$ obtained from $(T,\beta)$ when we remove the root of $T$. Then, to complete the proof in this case, it suffices to show that $(H,(T',\beta'))\in {\cal H}_{U\cup\{v\},C',\widehat{f},\widehat{g},\widehat{\cal S}}$ for some choice of $v,\widehat{f},\widehat{g}$ and $\widehat{\cal S}$ considered in the definition of the maximum, as this will imply that the aforementioned maximum is at least as large as the number of edges of $H$. Because the root of $T$ is a forget node and $U$ is its bag (by Condition \ref{item:purpose4}), there exists exactly one vertex $v\in V(H)\subseteq V(G)$ that belongs to the top bag of $T'$ but not of $T$ (and hence not to $U$). Define $\widehat{f}$ as the extension of $f$ that assigns to $v$ its degree in $H$. Also, define $\widehat{g}$ as the extension of $g$ that assigns to the edges incident to $v$ and a vertex in $U$ their multiplicity in $H$. Lastly, let $\widehat{\cal S}$ be the partition of the top bag of $T'$ where two vertices are in the same part if and only if they are in the same connected component of $H$. It is clear that if $|U|=2$, then the two vertices in $U$ are in the same part in $\cal S$ if and only if they are in the same part in $\widehat{\cal S}$. The satisfaction of each condition among Conditions \ref{item:purpose1}--\ref{item:purpose4} and \ref{item:purpose6} by $(H,(T,\beta))$ with respect to ${\cal H}^{\mathsf{F}}_{U,C',f_{\mathrm{deg}},g_{\mathrm{edg}},{\cal S}}$ directly yields the satisfaction of the same condition by $(H,(T',\beta'))$ with respect to ${\cal H}_{U\cup\{v\},C',\widehat{f},\widehat{g},\widehat{\cal S}}$. Further, by Condition \ref{item:purpose1} satisfied by $(H,(T,\beta))$ with respect to ${\cal H}^{\mathsf{F}}_{U,C',f_{\mathrm{deg}},g_{\mathrm{edg}},{\cal S}}$, we derive that $\widehat{f}(v)$ is even and does not exceed $2(r-d_i)$ where $i$ is the color of $v$,  and by Condition \ref{item:purpose3} satisfied by $(H,(T,\beta))$ with respect to ${\cal H}^{\mathsf{F}}_{U,C',f_{\mathrm{deg}},g_{\mathrm{edg}},{\cal S}}$, we also derive that the part that contains $v$ in $\widehat{\cal S}$ is of size at least $2$. As argued above, this completes the proof of this direction.

In the other direction, to show that $M^{\mathsf{F}}_{U,C',f_{\mathrm{deg}},g_{\mathrm{edg}},{\cal S}}$ is equal or larger than the aforementioned maximum, consider an integer $M_{U\cup\{v\},C',\widehat{f},\widehat{g},\widehat{\cal S}}$ that attains this maximum, and consider a pair $(H,(T',\beta'))\in {\cal H}_{U\cup\{v\},C',\widehat{f},\widehat{g},\widehat{\cal S}}$ where the number of edges of $H$ equals $M_{U\cup\{v\},C',\widehat{f},\widehat{g},\widehat{\cal S}}$.  Define $(T,\beta)$ as the tree decomposition of $H$ obtained from $(T',\beta')$ by adding a new root node with the previous root node as its only child, whose bag is $U$. Then, the root is a forget node. Now, to complete the proof in this case, it suffices to show that $(H,(T,\beta))\in {\cal H}^{\mathsf{F}}_{U,C',f_{\mathrm{deg}},g_{\mathrm{edg}},{\cal S}}$. However, this follows immediately, since he satisfaction of each condition among Conditions \ref{item:purpose1}--\ref{item:purpose4} and \ref{item:purpose6} by $(H,(T',\beta'))$ with respect to ${\cal H}_{U\cup\{v\},C',\widehat{f},\widehat{g},\widehat{\cal S}}$ directly yields the satisfaction of the same condition by $(H,(T,\beta))$ with respect to ${\cal H}^{\mathsf{F}}_{U,C',f_{\mathrm{deg}},g_{\mathrm{edg}},{\cal S}}$.

\medskip
\noindent{\bf Proof for Introduce Case:} By the computation performed by the algorithm and the inductive hypothesis, it suffices to prove that $M^{\mathsf{I}}_{U,C',f_{\mathrm{deg}},g_{\mathrm{edg}},{\cal S}}$ is equal to the maximum among the sums $M_{U\setminus\{\widetilde{v}\},C'\setminus\{i\},\widehat{f},\widehat{g},\widehat{\cal S}}+f_{\mathrm{deg}}(\widetilde{v})$ over each vertex $\widetilde{v}\in U$, where $i$ is the color of $\widetilde{v}$, each function $\widehat{f}:U\setminus\{\widetilde{v}\}\rightarrow \{0,1,\ldots,2r\}$, where $\widehat{g}$ is the restriction of $g_{\mathrm{edg}}$ to $\{\{u,v\}\in E(G): u,v\in U\setminus\{\widetilde{v}\}\}$, and each partition $\widehat{\cal S}$ of $U\setminus\{\widetilde{v}\}$ that altogether satisfy the following requirements (we will refer to the conditions below as requirements as to distinguish between them and Conditions \ref{item:purpose1}--\ref{item:purpose6} in the definition of the meaning of a table entry).
	\begin{enumerate}
	\item\label{item:fhatIntroduce1} $f_{\mathrm{deg}}(\widetilde{v})=\sum_{u\in U: \{u,\widetilde{v}\}\in E(G)}g_{\mathrm{edg}}(\{u,\widetilde{v}\})$.
	\item\label{item:fhatIntroduce2} For each $u\in U$, $f_{\mathrm{deg}}(u)=\widehat{f}(u)+b$ where $b=\widehat{g}(\{u,\widetilde{v}\})$ if $\{u,\widetilde{v}\}\in E(G)$ and $b=0$ otherwise.
	\item\label{item:fhatIntroduce3} If $|U|=3$ and the two vertices in $U \setminus \{\widetilde{v}\}$ are in the same part in $\widehat{\cal S}$, then they are also in the same part in ${\cal S}$.
	\item\label{item:fhatIntroduce4} For each $u\in U$ with $\{u,\widetilde{v}\}\in E(G)$ and $g_{\mathrm{edg}}(\{u,\widetilde{v}\})\geq 1$, $u,\widetilde{v}$ are in the same part in ${\cal S}$.
	\item\label{item:fhatIntroduce5} 	For each $u\in U$ with either $\{u,\widetilde{v}\}\notin E(G)$ or $g_{\mathrm{edg}}(\{u,\widetilde{v}\})=0$, if $u,\widetilde{v}$ are in the same part in ${\cal S}$, then there exists $w\in U\setminus\{u,\widetilde{v}\}$ with $\{w,\widetilde{v}\}\in E(G)$, $\widehat{g}(\{w,\widetilde{v}\})\geq 1$ and such that $u,w$ are in the same part in $\widehat{\cal S}$.
	\end{enumerate}
In one direction, to show that $M^{\mathsf{I}}_{U,C',f_{\mathrm{deg}},g_{\mathrm{edg}},{\cal S}}$ is equal or smaller than this maximum, consider a pair $(H,(T,\beta))\in {\cal H}^{\mathsf{I}}_{U,C',f_{\mathrm{deg}},g_{\mathrm{edg}},{\cal S}}$ where the number of edges of $H$ equals $M^{\mathsf{I}}_{U,C',f_{\mathrm{deg}},g_{\mathrm{edg}},{\cal S}}$.  Because the root of $(T,\beta)$ is an introduce node, there exists (exactly one) vertex $\widetilde{v}$ that belongs to the bag of the root of $T$ (which equals $U$), but not to the bag of its child. Let $H'$ be the graph obtained from $H$ by removing $\widetilde{v}$ (and all edges incident to it), and let $(T',\beta')$ be the tree decomposition of $H'$ obtained from $(T,\beta)$ when we remove the root of $T$. Then, to complete the proof in this case, it suffices to show that $(H',(T',\beta'))\in {\cal H}_{U\setminus\{\widetilde{v}\},C'\setminus\{i\},\widehat{f},\widehat{g},\widehat{\cal S}}$ for some choice of $v$, $\widehat{f},\widehat{g}$ and $\widehat{\cal S}$ considered in the definition of the maximum, as this will imply that the aforementioned maximum is at least as large as the number of edges of $H'$ plus $f_{\mathrm{deg}}(\widetilde{v})$ (which is equal to the degree of $\widetilde{v}$ in $H$ because $(H,(T,\beta))\in {\cal H}^{\mathsf{I}}_{U,C',f_{\mathrm{deg}},g_{\mathrm{edg}},{\cal S}}$), which equals the number of edges of $H$. Notice that we have already chosen $\widetilde{v}$ (and that $\widetilde{v}\in U$ and hence considered by the algorithm), and that the choice of $\widetilde{g}$ is unique. Moreover, the choice of $\widehat{f}$ is unique as well due to Requirement \ref{item:fhatIntroduce2} above. Thus, we define $\widehat{f}$ and $\widehat{g}$ accordingly. We choose $\widehat{\cal S}$ as the partition of $U\setminus\{\widetilde{v}\}$ where two vertices are in the same part if and only if they are in the same connected component of $H'$. Because $(H,(T,\beta))\in {\cal H}^{\mathsf{I}}_{U,C',f_{\mathrm{deg}},g_{\mathrm{edg}},{\cal S}}$, we know that ($g_{\mathrm{edg}}$ and hence also) $\widetilde{g}$ assigns to each edge in its domain its multiplicity in $H'$, and that two vertices are in the same part in $\cal S$ if and only if they are in the same connected component in $H$. Thus, as $H'$ is obtained from $H$ by the removal of $\widetilde{v}$, we immediately get that Requirements \ref{item:fhatIntroduce3}--\ref{item:fhatIntroduce5} are satisfied by the partition $\widehat{\cal S}$ that we have defined. It remains to prove that $(H',(T',\beta'))\in {\cal H}_{U\setminus\{\widetilde{v}\},C'\setminus\{i\},\widehat{f},\widehat{g},\widehat{\cal S}}$. The satisfaction of each condition among Conditions \ref{item:purpose1}--\ref{item:purpose4} and \ref{item:purpose6} by $(H,(T,\beta))$ with respect to ${\cal H}^{\mathsf{I}}_{U,C',f_{\mathrm{deg}},g_{\mathrm{edg}},{\cal S}}$ directly yields the satisfaction of the same condition by $(H',(T',\beta'))$ with respect to ${\cal H}_{U\setminus\{\widetilde{v}\},C'\setminus\{i\},\widehat{f},\widehat{g},\widehat{\cal S}}$. This completes the proof of this direction.

In the other direction, to show that $M^{\mathsf{I}}_{U,C',f_{\mathrm{deg}},g_{\mathrm{edg}},{\cal S}}$ is equal or larger than the aforementioned maximum, consider an integer $M_{U\setminus\{\widetilde{v}\},C'\setminus\{i\},\widehat{f},\widehat{g},\widehat{\cal S}}+f_{\mathrm{deg}}(\widetilde{v})$ that attains this maximum, and consider a pair $(H',(T',\beta'))\in {\cal H}_{U\setminus\{\widetilde{v}\},C'\setminus\{i\},\widehat{f},\widehat{g},\widehat{\cal S}}$ where the number of edges of $H'$ equals $M_{U\setminus\{\widetilde{v}\},C'\setminus\{i\},\widehat{f},\widehat{g},\widehat{\cal S}}$. Define $H$ as the graph obtained from $H'$ by adding $\widetilde{v}$ to $H'$, as well as an edge from $\widetilde{v}$ to each $u\in U\setminus\widetilde{v}$ that is a neighbor of $\widetilde{v}$ in $G$ with multiplicity $g_{\mathrm{edg}}(\{\widetilde{v},u\})$. Define $(T,\beta)$ as the tree decomposition of $H$ obtained from $(T',\beta')$ by adding a new root node with the previous root node as its only child, whose bag is $U$. Then, the root is an introduce node. Now, to complete the proof in this case, it suffices to show that $(H,(T,\beta))\in {\cal H}^{\mathsf{I}}_{U,C',f_{\mathrm{deg}},g_{\mathrm{edg}},{\cal S}}$---indeed, this follows as the number of edges of $H$ equals the number of edges of $H'$ plus $\sum_{u\in U: \{u,\widetilde{v}\}\in E(G)} g_{\mathrm{edg}}(\{\widetilde{v},u\})$ where the latter sum equals $f_{\mathrm{deg}}(\widetilde{v})$ by Requirement \ref{item:fhatIntroduce1}).  The satisfaction of Condition \ref{item:purpose2} by $(H,(T,\beta))$ with respect to ${\cal H}^{\mathsf{I}}_{U,C',f_{\mathrm{deg}},g_{\mathrm{edg}},{\cal S}}$ directly follows from the satisfaction of this condition by $(H',(T',\beta'))$ with respect to ${\cal H}_{U\setminus\{\widetilde{v}\},C'\setminus\{i\},\widehat{f},\widehat{g},\widehat{\cal S}}$ and the definition of $H$ (for the edges not in $H'$). Then, because this condition is satisfied, we also get that the satisfaction of Condition \ref{item:purpose1} by $(H,(T,\beta))$ with respect to ${\cal H}^{\mathsf{I}}_{U,C',f_{\mathrm{deg}},g_{\mathrm{edg}},{\cal S}}$ follows from the satisfaction of this condition by $(H',(T',\beta'))$ with respect to ${\cal H}_{U\setminus\{\widetilde{v}\},C'\setminus\{i\},\widehat{f},\widehat{g},\widehat{\cal S}}$ and Requirements \ref{item:fhatIntroduce1} and \ref{item:fhatIntroduce2}. The satisfaction of Condition \ref{item:purpose3} by $(H,(T,\beta))$ with respect to ${\cal H}^{\mathsf{I}}_{U,C',f_{\mathrm{deg}},g_{\mathrm{edg}},{\cal S}}$ follows from the satisfaction of the second condition and this condition by $(H',(T',\beta'))$ with respect to ${\cal H}_{U\setminus\{\widetilde{v}\},C'\setminus\{i\},\widehat{f},\widehat{g},\widehat{\cal S}}$, the definition of $H$ and Requirements \ref{item:fhatIntroduce3}--\ref{item:fhatIntroduce5}. The satisfaction of Condition \ref{item:purpose4} by $(H,(T,\beta))$ with respect to ${\cal H}^{\mathsf{I}}_{U,C',f_{\mathrm{deg}},g_{\mathrm{edg}},{\cal S}}$  follows from the satisfaction of this condition by $(H',(T',\beta'))$ with respect to ${\cal H}_{U\setminus\{\widetilde{v}\},C'\setminus\{i\},\widehat{f},\widehat{g},\widehat{\cal S}}$. Lastly, the satisfaction of Condition \ref{item:purpose6} by $(H,(T,\beta))$ with respect to ${\cal H}^{\mathsf{I}}_{U,C',f_{\mathrm{deg}},g_{\mathrm{edg}},{\cal S}}$  follows from the satisfaction of the second conditions and this condition by $(H',(T',\beta'))$ with respect to ${\cal H}_{U\setminus\{\widetilde{v}\},C'\setminus\{i\},\widehat{f},\widehat{g},\widehat{\cal S}}$ and because $i$ is the color of $\widetilde{v}$.

\medskip
\noindent{\bf Proof for Join Case:} By the computation performed by the algorithm and the inductive hypothesis, it suffices to prove that $M^{\mathsf{J}}_{U,C',f_{\mathrm{deg}},g_{\mathrm{edg}},{\cal S}}$ is equal to the maximum among the sums \[M_{U,C_1,f_1,g_{\mathrm{edg}},{\cal S}_1}+M_{U,C_2,f_2,g_{\mathrm{edg}},{\cal S}_2}-\sum_{\{u,v\}\in E(G): u,v\in U}g_{\mathrm{edg}}(\{u,v\})\]
	over each $C_1\subseteq C'$ that contains $X$ being the set of colors of the vertices in $U$ as well as at least one other color, and which is not equal to $C'$, where $C_2 = C'\setminus (C_1\setminus X)$, each pair of functions $f_1:U\rightarrow \{0,1,\ldots,2r\}$ and $f_2:U\rightarrow \{0,1,\ldots,2r\}$ such that $f_{\mathrm{deg}}(v)=f_1(v)+f_2(v)-\sum_{\{u,v\}\in E(G): u\in U}g_{\mathrm{edg}}(\{u,v\})$ for every $v\in U$, and each pair of partitions ${\cal S}_1$ and ${\cal S}_2$ of $U$ such that for all $u,v\in U$, $u,v$ are in the same part in ${\cal S}$ if and only if $u,v$ are in the same part in the finest common coarsening of ${\cal S}_1$ and ${\cal S}_2$ (i.e., since $|U|\leq 3$, either $u$ and $v$ are in the same part in $\mathcal{S}_1$ or $\mathcal{S}_2$, or there exists $w\in U\setminus\{u,v\}$ such that $u$ and $w$ are in the same part in $\mathcal{S}_i$ and $w$ and $v$ in $\mathcal{S}_{3-i}$ for some $i \in \{1,2\}$).

In one direction, to show that $M^{\mathsf{J}}_{U,C',f_{\mathrm{deg}},g_{\mathrm{edg}},{\cal S}}$ is equal or smaller than this maximum, consider a pair $(H,(T,\beta))\in {\cal H}^{\mathsf{J}}_{U,C',f_{\mathrm{deg}},g_{\mathrm{edg}},{\cal S}}$ where the number of edges of $H$ equals $M^{\mathsf{J}}_{U,C',f_{\mathrm{deg}},g_{\mathrm{edg}},{\cal S}}$.  Because the root $r$ of $(T,\beta)$ is a join node, it has exactly two children, $r_1$ and $r_2$, having the same bag as $r$ (which is $U$). For $i\in\{1,2\}$, let $H_i$ be the subgraph of $H$ induced by the union of bags of the descendants of $r_i$ (along with $r_i$ itself), and let $(T_i,\beta_i)$ be the tree decomposition of $H_i$ that is the restriction of $(T,\beta)$ induced by all nodes (and their bags) that are descendants of $r_i$ (along with $r_i$ itself). Then, we claim that to complete the proof in this case, it suffices to show that $(H_1,(T_1,\beta_1))\in {\cal H}_{U,C_1,f_1,g_{\mathrm{edg}},{\cal S}_1}$ and $(H_2,(T_2,\beta_2))\in {\cal H}_{U,C_2,f_2,g_{\mathrm{edg}},{\cal S}_2}$ for some choice of $C_1,C_2,f_1,f_2,{\cal S}_1$ and ${\cal S}_2$ considered in the definition of the maximum. Indeed, this will imply that the aforementioned maximum is at least as large as the sum of the number of edges in $H_1$ plus the number of edges in $H_2$ minus $\sum_{\{u,v\}\in E(G): u,v\in U}g_{\mathrm{edg}}(\{u,v\})$, which is precisely the number of edges in $H$ due to Condition \ref{item:purpose2} and as each edge in $H$ appears in at least one among $H_1$ and $H_2$ where only the edges between the vertices in $U$ appear in both. For each $i\in\{1,2\}$, we choose $C_i$ as the set of colors of the vertices in $H_i$, $f_i$ the function that assigns to each vertex in $U$ its degree in $H_i$, and ${\cal S}_i$ the partition of $U$ where two vertices are in the same part if and only if they are in the same connected component of $H_i$. It is clear that for each $i\in\{1,2\}$, due to the definition of $C_i,f_i$ and ${\cal S}_i$,  the satisfaction of Conditions \ref{item:purpose1}--\ref{item:purpose4} and \ref{item:purpose6} by $(H,(T,\beta))$ with respect to ${\cal H}^{\mathsf{J}}_{U,C',f_{\mathrm{deg}},g_{\mathrm{edg}},{\cal S}}$ implies the satisfaction of these conditions by $(H_i,(T_i,\beta_i))$ with respect to ${\cal H}_{U,C_i,f_i,g_{\mathrm{edg}},{\cal S}_i}$. It remains to show that our choice of $C_1,C_2,f_1,f_2,{\cal S}_1$ and ${\cal S}_2$ is considered in the definition of the maximum. First, notice that $C_1\subseteq C'$ and $C_2=C'\setminus (C_1\setminus X)$ since $C_1$ is the set of colors used in $H_1$, and $C_2$ is the set of colors used in $H_2$, where the union of $H_1$ and $H_2$ yields $H$ (whose set of used colors is $C'$) where the set of common vertices is precisely $X$ and hence the set of their colors is precisely the set of common colors (because $H$ is colorful). Next, we need to argue that $f_{\mathrm{deg}}(v)=f_1(v)+f_2(v)-\sum_{\{u,v\}\in E(G): u\in U}g_{\mathrm{edg}}(\{u,v\})$ for every $v\in U$ (the fact that both assign values upper bounded by $2r$ follows from the fact that $H$, and hence also $H_1$ and $H_2$, have maximum degree upper bounded by $2r$). To this end, notice that $f_{\mathrm{deg}}(v)$ is the degree of $v$ in $H$, and that each edge incident to $v$ in $H$ appears in at least one of $H_1$ and $H_2$ where the only edges appearing in both are those between $v$ and vertices in $U$. Thus, the equality follows from the definition of $f_1$ and $f_2$. Lastly, we need to argue that for all $u,v\in U$, $u,v$ are in the same part in ${\cal S}$ if and only if $u,v$ are in the same part in the finest common coarsening of ${\cal S}_1$ and ${\cal S}_2$. However, this directly follows from the definition of ${\cal S}_1$ and ${\cal S}_2$ and since any two vertices in $U$ are in the same part in ${\cal S}$ if and only if they are in the same connected component of $H$.

In the other direction, to show that $M^{\mathsf{J}}_{U,C',f_{\mathrm{deg}},g_{\mathrm{edg}},{\cal S}}$ is equal or larger than the aforementioned maximum, consider a sum \[M_{U,C_1,f_1,g_{\mathrm{edg}},{\cal S}_1}+M_{U,C_2,f_2,g_{\mathrm{edg}},{\cal S}_2}-\sum_{\{u,v\}\in E(G): u,v\in U}g_{\mathrm{edg}}(\{u,v\})\] that attains this maximum, and for each $i\in\{1,2\}$, consider a pair $(H_i,(T_i,\beta_i))\in {\cal H}_{U,C_i,f_i,g_{\mathrm{edg}},{\cal S}_i}$ where the number of edges of $H_i$ equals $M_{U,C_i,f_i,g_{\mathrm{edg}},{\cal S}_i}$. Define $H$ as the graph whose vertex set and edge set (with multiplicities) are the union of the vertex sets and edge sets of $H_1$ and $H_2$, respectively, where the edge multiplicites between vertices in $U$ are only taken from, say, $H_1$. Define $(T,\beta)$ as the tree decomposition of $H$ obtained by introducing a new node $r$ as a root whose bag is $U$, and attaching it as the parent of the root of $T_1$ and the root of $T_2$ and assigning bags accordingly as done by $\beta_1$ and $\beta_2$. Then, for each $i\in\{1,2\}$, the root of $T_i$ is assigned $U$ as its bag (because $(H_i,(T_i,\beta_i))\in {\cal H}_{U,C_i,f_i,g_{\mathrm{edg}},{\cal S}_i}$), and hence the root of $(T,\beta)$ is a join node. Now, we argue that to complete the proof in this case, it suffices to show that $(H,(T,\beta))\in {\cal H}^{\mathsf{J}}_{U,C',f_{\mathrm{deg}},g_{\mathrm{edg}},{\cal S}}$. To see this, observe that the only common vertices of $H_1$ and $H_2$ are  those in $U$ because $H_1$ and $H_2$ are colorful, $C_1$  is the set of colors of $H_1$ and $C_2$ is the set of colors of $C_2$ (due to membership in ${\cal H}_{U,C_i,f_i,g_{\mathrm{edg}},{\cal S}_i}$ for the corresponding $i\in\{1,2\}$), while $C_2=C'\setminus (C_1\setminus X)$. Moreover, the edges and their multiplicities between vertices in $U$ are the same in $H_1$ and $H_2$ and equal to the values specified by $g_{\mathrm{edg}}$ (again, due to this membership). Thus, the number of edges of $H$ is the number of edges of $H_1$ plus the number of edges of $H_2$ minus $\sum_{\{u,v\}\in E(G): u,v\in U}g_{\mathrm{edg}}(\{u,v\})$. Finally, notice that the satisfaction of Conditions \ref{item:purpose1} by $(H,(T,\beta))$ with respect to ${\cal H}^{\mathsf{J}}_{U,C',f_{\mathrm{deg}},g_{\mathrm{edg}},{\cal S}}$ directly follows from the satisfaction of these conditions by $(H_i,(T_i,\beta_i))$ with respect to ${\cal H}_{U,C_i,f_i,g_{\mathrm{edg}},{\cal S}_i}$ for each $i\in\{1,2\}$ and due to the restrictions on the choice of $C_1,C_2,f_1,f_2,{\cal S}_1$ and ${\cal S}_2$.
\end{proof}

We are now ready to solve the \eulerPart\ problem.

\begin{lemma}\label{lem:eulerPartProblem}
\eulerPart\ can be solved in time $2^{\OO(k/r)}\cdot (r+n+\log k)^{\OO(1)}$.
\end{lemma}

\begin{proof}
Let $\cal A$ denote the algorithm in Lemma \ref{lem:tw2CompDP}. We now describe a DP procedure to solve \eulerPart. To this end, we let $(G,k,r,\overline{\bf d}=(d_1,\ldots,d_{\mathsf{b}(k/r)}))$ be an instance of \eulerPart. We have a DP table $\mathsf{N}$ with an entry $\mathsf{N}[v,\overline{\bf d}',C]$ for every vertex $v\in V(G)$, occurrence sequence $\overline{\bf d}'=(d'_1,\ldots,d'_{\mathsf{b}(k/r)})$ such that $d_i'\in\{0,\ldots,d_i\}$ for every $i\in\{1,\ldots,\mathsf{b}(k/r)\}$, and set of colors $C\subseteq\{1,\ldots,\mathsf{b}(k/r)\}$.

The purpose of each entry $\mathsf{N}[v,\overline{\bf d}',C]$ is to store the largest integer $M$ such that there exists an $M$-good pair  $(W,H)$ (i.e.,~$M=|E(W)|+|E(H)|$) which complies with $(\overline{\bf d},\overline{\bf d}')$, where $v$ is an end-vertex of $W$, and where the set of colors of vertices in $H$ is a subset of $C$. The order of computation is non-decreasing with respect to $\sum_{i=1}^{\mathsf{b}(k/r)}d'_i$.

In the DP basis, we consider every entry $\mathsf{N}[v,\overline{\bf d}',C]$ that satisfies $\sum_{i=1}^{\mathsf{b}(k/r)}d'_i\leq 1$, and let $c$ be the color of $v$. If $d_c'\neq 1$, then $\mathsf{N}[v,\overline{\bf d}',C]=-\infty$. Else $\mathsf{N}[v,\overline{\bf d}',C]$ is the maximum of $0$ and the output of algorithm $\cal A$ when called with input $(G,C,v,\overline{\bf d})$.

For the DP step, we consider every entry $\mathsf{N}[v,\overline{\bf d}',C]$ that satisfies $\sum_{i=1}^{\mathsf{b}(k/r)}d'_i\geq 2$. Let $c$ be the color of $v$. If $d'_c=0$, then $\mathsf{N}[v,\overline{\bf d}',C]=-\infty$. Now, suppose that $d'_c\geq 1$. Denote $\overline{\bf d}''=(d''_1,\ldots,d''_{\mathsf{b}(k/r)})$ where $d''_c=d'_c-1$ and $d''_i=d'_i$ for all $i\in\{1,\ldots,\mathsf{b}(k/r)\}\setminus\{c\}$. In addition, for every subset $C'\subseteq C$, let $A_{C'}$ be the output of algorithm $\cal A$ when called with input $(G,C',v,\overline{\bf d})$.
 Then, 
\[\mathsf{N}[v,\overline{\bf d}',C]=\max_{u:\{u,v\}\in E(G)}\left(\max\left\{1+\mathsf{N}[u,\overline{\bf d}'',C],\max_{C'\subseteq C}(A_{C'}+1+\mathsf{N}[u,\overline{\bf d}'',C\setminus C'])\right\}\right).\]

After the DP computation is complete, we return \yes\ if and only if there exists an entry $\mathsf{N}[v,\overline{\bf d},C]$ for some $v\in V(G)$ and $C\subseteq\{1,\ldots,\mathsf{b}(k/r)\}$ that stores an integer that is at least $k-1.$

\medskip
\noindent{\bf Time Complexity.} The table $\mathsf{N}$ has $2^{\OO(k/r)}n$ entries since there are $n$ choices of $v$, $2^{\OO(k/r)}$ choices for $\overline{\bf d}'$ (by Lemma \ref{lem:degreeSeqNumber}), and $2^{|C|}=2^{\mathrm{b}(k/r)}=2^{\OO(k/r)}$ choices for $C'$. The computation of each entry entails at most $n2^{\OO(k/r)}$ calls to the algorithm in Lemma \ref{lem:tw2CompDP}, and each call takes time $2^{\OO(k/r)}\cdot (r+n+\log k)^{\OO(1)}$. Thus, the total running time of our algorithm is $2^{\OO(k/r)}\cdot (r+n+\log k)^{\OO(1)}$. 

\medskip
\noindent{\bf Correctness.} The correctness of our algorithm can be verified by a simple induction on $\sum_{i=1}^{\mathsf{b}(k/r)}d'_i$. For the sake of completeness, we give the details. For any entry $\mathsf{N}[v,\overline{\bf d}',C]$, let ${\cal P}_{v,\overline{\bf d}',C}$ be the collection of all good pairs  $(W,H)$ that comply with $\overline{\bf d}'$, where $v$ is an end-vertex of $W$, and where the set of colors of vertices in $H$ is $C$, and let $M_{v,\overline{\bf d}',C}$ denote the maximum number of edges (with multiplicities) in a pair in ${\cal P}_{v,\overline{\bf d}',C}$. Then, we need to prove that for each entry, $\mathsf{N}[v,\overline{\bf d}',C]=M_{v,\overline{\bf d}',C}$.

In the basis, consider an entry $\mathsf{N}[v,\overline{\bf d}',C]$ such that $\sum_{i=1}^{\mathsf{b}(k/r)}d'_i$. First, notice that ${\cal P}_{v,\overline{\bf d}',C}$ is empty when $d'_c\neq 1$ where $c$ is the colors of $v$, and hence the assignment of $-\infty$ is correct. Now, suppose that $d'_c= 1$. On the one hand, any $q$-good pair  $(W,H)$ (for any $q$) that complies with $(\overline{\bf d},\overline{\bf d}')$ must be such that $W$ consists of a single vertex. Adding the demand that the last (and only) vertex of $W$ is $v$ and $H$ uses exactly the colors in $C$, we get that $(W,H)\in {\cal P}_{v,\overline{\bf d}',C}$ must be such that $W$ consists only of $v$ and that the set of colors of $H$ must be$C$ (and hence $H$ satisfies Constraint \ref{item:tw2CompDP6} in Lemma \ref{lem:tw2CompDP}). Further, $H$ must be either empty or contain a vertex from $W$ (by the second requirement in the definition of a $q$-good pair); in the latter case, $H$ must contain $v$ (and hence satisfy Constraint \ref{item:tw2CompDP4} in Lemma \ref{lem:tw2CompDP}) as well as satisfy Constraints \ref{item:tw2CompDP1}, \ref{item:tw2CompDP3} and \ref{item:tw2CompDP6} in Lemma \ref{lem:tw2CompDP} (because $(W,H)$ is a $q$-good pair that complies with $\overline{\bf d}'$). Thus, when $H$ is not empty, the output of the algorithm in Lemma \ref{lem:tw2CompDP} (due to Constraint \ref{item:tw2CompDP5}) is at least as large as $|E(H)|$. From this, we get that $\mathsf{N}[v,\overline{\bf d}',C]\geq M_{v,\overline{\bf d}',C}$. On the other hand, notice that any pair $(W,H)$ where $W$ consists only of $v$ and $H$ is either empty or any multigraph that can attain the maximum returned by Lemma \ref{lem:tw2CompDP} belongs to ${\cal P}_{v,\overline{\bf d}',C}$ due to the constraints in this lemma, and therefore we also have that $\mathsf{N}[v,\overline{\bf d}',C]\leq M_{v,\overline{\bf d}',C}$. Thus, the basis is correct.

Now, we prove correctness for an entry $\mathsf{N}[v,\overline{\bf d}',C]$ such that $\sum_{i=1}^{\mathsf{b}(k/r)}d'_i\geq 2$ under the assumption of correctness for all  entries $\mathsf{N}[v'',\overline{\bf d}'',C'']$ where $\overline{\bf d}''$ is such that $\sum_{i=1}^{\mathsf{b}(k/r)}d''_i<\sum_{i=1}^{\mathsf{b}(k/r)}d'_i$. Let $c$ be the color of $v$, and let $\overline{\bf d}''$ be as defined by the algorithm when it computes $\mathsf{N}[v,\overline{\bf d}',C]$, that is, $\overline{\bf d}''=(d''_1,\ldots,d''_{\mathsf{b}(k/r)})$ where $d''_c=d'_c-1$ and $d''_i=d'_i$ for all $i\in\{1,\ldots,\mathsf{b}(k/r)\}\setminus\{c\}$. By the inductive hypothesis and the formula used by the algorithm, we need to prove that 
\[M_{v,\overline{\bf d}',C}=\max_{u:\{u,v\}\in E(G)}\left(\max\left\{1+M_{u,\overline{\bf d}'',C},\max_{C'\subseteq C}(A_{C'}+1+M_{u,\overline{\bf d}'',C\setminus C'})\right\}\right).\]
Let us denote the left hand side above by $Q$.

In one direction, to prove that $M_{v,\overline{\bf d}',C}\leq Q$, consider a pair $(W,H)\in {\cal P}_{v,\overline{\bf d}',C}$ whose number of edges (including multiplicities) is $M_{v,\overline{\bf d}',C}$. Let $W''$ denote $W$ without its last vertex (and edge) occurrence, and let $u$ denote the last vertex of $W'$. Then, $\{u,v\}\in E(G)$. 
In case $H$ does not contain $v$, we immediately get that $(W'',H)\in {\cal P}_{u,\overline{\bf d}'',C}$, which means that $M_{u,\overline{\bf d}'',C}\geq |E(W'')|+|E(H)|$, and therefore $Q\geq 1+M_{u,\overline{\bf d}'',C}\geq 1+|E(W'')|+|E(H)| = |E(W)|+|E(H)|= M_{v,\overline{\bf d}',C}$. Next, suppose that $H$ has some connected component $H'$ that contains $v$, and let $H''$ denote the graph $H$ without $H'$. Let $C'$ be the set of colors used by $H_v$, and denote $C''=C\setminus C'$. Then, because $H$ is colorful, the set of colors used by vertices in $H''$ is precisely $C''$. First, notice that $(W'',H'')\in {\cal P}_{u,\overline{\bf d}'',C''}$, which means that $M_{u,\overline{\bf d}'',C''}\geq |E(W'')|+|E(H'')|=|E(W)|-1+|E(H)|-|E(H_v)|=M_{v,\overline{\bf d}',C}-1-|E(H_v)|$, and therefore $Q\geq A_{C'} + 1 + M_{u,\overline{\bf d}'',C''}\geq M_{v,\overline{\bf d}',C}+A_{C'}-|E(H_v)|$. Hence, it remains to prove that $A_{C'}\geq |E(H_v)|$. To this end, it suffices to prove that $H_v$ satisfies the conditions in Lemma  \ref{lem:tw2CompDP} with respect to $v$, $C_v$ and $\overline{\bf d}'$. The satisfaction of all of these conditions directly follows because $(W,H)$ is a $q$-good pair that complies with $(\overline{\bf d},\overline{\bf d}')$.

In the other direction, to prove that $M_{v,\overline{\bf d}',C}\geq Q$, we consider two cases. In the first case, suppose that the maximum with respect to $Q$ is attained by $1+M_{u,\overline{\bf d}'',C}$ for some neighbor $u$ of $v$, and consider a pair $(W'',H)\in {\cal P}_{u,\overline{\bf d}'',C}$ whose number of edges (including multiplicities) is $M_{u,\overline{\bf d}'',C}$. Define $W$ as the walk obtained from $W''$ by visiting $v$ after $u$ at the end. Because $(W'',H)\in {\cal P}_{u,\overline{\bf d}'',C}$, we  derive that $(W,H)\in {\cal P}_{u,\overline{\bf d}',C}$. Because $|E(W)|=|E(W'')|+1$, this means that $M_{v,\overline{\bf d}',C}\geq |E(W)|+|E(H)| = 1+(|E(W'')|+|E(H)|)= 1+M_{u,\overline{\bf d}'',C}=Q$.
In the second case, suppose that the maximum with respect to $Q$ is attained by $A_{C'}+1+M_{u,\overline{\bf d}'',C\setminus C'}$ for some neighbor $u$ of $v$ and subset $C'\subseteq C$. Consider a pair $(W'',H'')\in {\cal P}_{u,\overline{\bf d}'',C\setminus C'}$ whose number of edges (including multiplicities) is $M_{u,\overline{\bf d}'',C\setminus C'}$ as well as a multigraph $H'$ that satisfies the conditions in Lemma \ref{lem:tw2CompDP} with respect to $\overline{\bf d}, C',v$ and $M=A_{C'}$. Define $W$ as the walk obtained from $W''$ by visiting $v$ after $u$ at the end, and $H$ as the graph obtained by taking the union of $H''$ and $H'$. Because $(W'',H'')\in {\cal P}_{u,\overline{\bf d}'',C\setminus C'}$ and due to the conditions in Lemma \ref{lem:tw2CompDP}, we derive that $(W,H)\in {\cal P}_{u,\overline{\bf d}',C}$. Because $|E(W)|=|E(W'')|+1$ and $|E(H)|=|E(H'')|+|E(H')|=|E(H'')|+A_{C'}$, this means that $M_{v,\overline{\bf d}',C}\geq |E(W)|+|E(H)| = 1+(|E(W'')|+|E(H'')|)+A_{C'}= 1+M_{u,\overline{\bf d}'',C\setminus C'}+A_{C'}=Q$. This completes the proof.
\end{proof}

\subsection{Proof of Lemma \ref{lem:rUnboundUndirectedPathFPT}}\label{sec:mainLemmaProof}

By Lemma \ref{lem:eulerPartProblem}, \eulerPart\ can be solved in time $2^{\OO(k/r)}\cdot  (r+n+\log k)^{\OO(1)}$. Thus, by Lemma \ref{lem:pairProblem}, \colPathRst\ can be solved in time $2^{\OO(k/r)}\cdot (r+n+\log k)^{\OO(1)}$. In turn, by Lemma \ref{lem:colorUndirected}, \undiPathR\ can be solved in time $2^{\OO(k/r)}\cdot (r+n+\log k)^{\OO(1)}$, which completes the proof.\qed


\subsection{Bounding $r$}\label{sec:boundR}

In what follows, we focus on the proof of Lemma \ref{lem:undirectedPathFPTspecial}. Without loss of generality, we implicitly suppose that given an instance $(G,k,r)$ of \undiPathRSpecial, the graph $G$ is connected, else the problem can be solved by considering each connected component separately.

\paragraph{Bounding the Vertex Cover Number.} The reason why the case where $r>\sqrt{k}$ is simpler than the general case lies in the following lemma.

\begin{lemma}\label{lem:specialPathMM}
Let $(G,k,r)$ be an instance of \undiPathRSpecial. If $G$ has a matching of size $\lceil k/r \rceil$, then $(G,k,r)$ is a \yes-instance.
\end{lemma}

\begin{proof}
Suppose that $G$ has a matching $M$ of size $s=\lceil k/r \rceil$, and denote $M=\{\{u_1,v_1\},\{u_2,$ $v_2\},\ldots,\{u_s,v_s\}\}$.
For every $i\in\{1,2,\ldots,s-1\}$, let $P_i$ denote an arbitrary path in $G$ from $v_i$ to $u_{i+1}$ (such a path exists since $G$ is assumed to be connected). Consider the following walk:
\[W=u_1-v_1-P_1-u_2-v_2-P_2-\cdots-u_{s-1}-v_{s-1}-P_{s-1}-u_s-v_s.\]
For every $i\in\{1,2,\ldots,s-1\}$, let $occ_i$ denote the maximum of the number of occurrences of $u_i$ in $W$ and the number of occurrences of $v_i$ in $W$. Note that each vertex occurs at most once in each path $P_j$, $j\in\{1,2,\ldots,s-1\}$. In particular, $occ_i\leq s\leq r$  for all $i\in\{1,2,\ldots,s\}$. To describe our modification of $W$ we need the following notation: for every $i\in \{1,2,\ldots,s\}$, let 
$Q_i$ denote the $(u_i,v_i)$-walk that traverses the edge $\{u_i,v_i\}$ 
several times such that each vertex among $u_i$ and $v_i$ occurs in $Q_i$ exactly $r-occ_i+1$ times. Now, we modify $W$ as follows:
\[W' = Q_1-P_1-Q_2-P_2-\cdots-Q_{s-1}-P_{s-1}-Q_s.\]
Then, every vertex occurs at most $r$ times in $W'$. Moreover, for every $i\in\{1,2,\ldots,s\}$, at least one among the vertices $u_i$ and $v_i$ occurs exactly $r$ times in $W'$. Thus, the size of $W'$ is at least 
$s\cdot r= \lceil k/r\rceil\cdot r\geq k$. Thus, $G$ has an $r$-simple $k$-path.
\end{proof}

Since the set of endpoints of any maximal matching is a vertex cover, and a maximal matching can be computed greedily in polynomial time, we derive the following corollary. Here, $2\lceil k/r \rceil\leq 2(k/r + 1)\leq 3k/r$ because we can assume that $k/r\geq 2$ (else, $(G,k,r)$ is a \yes-instance\ if and only if $G$ is not edgeless).

\begin{corollary}\label{cor:specialPathMM}
There exists a polynomial-time algorithm that, given an instance $(G,k,r)$ of \undiPathRSpecial, either correctly concludes that it is a \yes-instance or outputs a vertex cover of $G$ of size at most $2\lceil k/r \rceil\leq 3k/r$.
\end{corollary}

\paragraph{Color Coding and Vertex Guessing.} We define the following problem. 

\begin{defproblem}
{\colPathRspecial}
{An $n$-vertex $\mathsf{b}(k/r)$-colored undirected graph $G$, positive integers $k,r$, and a vertex cover $U$ of $G$ of size at most $3k/r$ where each vertex in $U$ has a unique color.}
{Output \no\ if $G$ has no $r$-simple $k$-path, and \yes\ if it has a colorful $r$-simple $k$-path that visits every vertex in $U$ and which has fewer than $30(k/r)$ distinct edges.}
\end{defproblem}\\
We refer to any instance where we must output \no\ as a \no-instance, and to any instance where we must output \yes\ as a \yes-instance). 
Notice that if the input is neither a \yes-instance nor a \no-instance, then the output can be arbitrary. 

Now, we have the following result.

\begin{lemma}\label{lem:colorProblemSpecial}
Suppose that \colPathRspecial\ can be solved in time $f(k/r)\cdot (n+\log k)^{\OO(1)}$. Then, \undiPathRSpecial\ can be solved in time $2^{\OO(k/r)}\cdot f(k/r)\cdot (n+\log k)^{\OO(1)}$.
\end{lemma}

\begin{proof}
By Lemma \ref{lem:colorUndirected}, it suffices to show that \colPathRst\ where $r>\sqrt{k}$ can be solved in time $2^{\OO(k/r)}\cdot f(k/r)\cdot (n+\log k)^{\OO(1)}$. Let $\cal A$ be an algorithm that solves \colPathRspecial\ where $r>\sqrt{k}$ in time $f(k/r)\cdot (n+\log k)^{\OO(1)}$. Then, given an instance $(G,k,r)$ of \colPathRst\ where $r>\sqrt{k}$, we first use the algorithm in Corollary \ref{cor:specialPathMM} to  either correctly conclude that $(G,k,r)$ is a \yes-instance or find a vertex cover $U$ of $G$ of size at most $3k/r$. For every subset $U'\subseteq U$, we call $\cal A$ with $(G',k,r,U')$ as input where $G'=G-X$ for $X=(U\setminus U')\cup\{v\in V(G)\setminus U: $ there exists a vertex in $U'$ with the same color as $v\}$. Thus, we obtain $G'$ from $G$ by removing $U\setminus U'$ and all vertices with the same color as vertices from $U'$. Notice that $U'$ is a vertex cover for $G'$. Then, we accept if and only if at least one of the calls accepts. 

For correctness, first suppose that $(G,k,r)$ is a \yes-instance, thus $G$ has a colorful $r$-simple $k$-path $P$. Let $U'$ be the set of vertices in $U$ visited by $P$. Because $P$ is colorful, and by the choice of $U'$, we know that $P$ does not visit any vertex in $U\setminus U'$ as well as any vertex in $G$ having the same color as a vertex in $U'$. Thus, when the algorithm examines this $U'$, the call to $\cal A$ must return \yes\ (because $P$ is a colorful $r$-simple $k$-path in $G'$). On the other hand, if some call to $\cal A$, say, with input $(G',k,r,U')$ returned \yes, then $G'$ has an $r$-simple $k$-path, and therefore so does any supergraph of $G'$ including $G$.
For running time, recall that the algorithm in Corollary \ref{cor:specialPathMM} runs in polynomial time. Since $|U|\leq 3k/r$, the algorithm makes only $2^{\OO(k/r)}$ calls to algorithm $\cal A$, which runs in time $f(k/r)\cdot (n+\log k)^{\OO(1)}$. Thus, the total running time is $2^{\OO(k/r)}\cdot f(k/r)\cdot (n+\log k)^{\OO(1)}$.
\end{proof}

\paragraph{Occurrence Sequence.} The presence of a small vertex cover gives rise to the definition of a problem simpler than \eulerPart, which we will be able to solve while having a polylogarithmic (rather than polynomial) dependency on $r$. To this end, we need a new definition.

\begin{definition}\label{def:fit}
Let $r,k\in\mathbb{N}$. Let $G$ be a $\mathsf{b}(k/r)$-colored undirected graph, and let $U$ be a vertex cover of $G$. In addition, let $\overline{\bf d}=(d_1,d_2,\ldots,d_{\mathsf{b}(k/r)})\in {\cal D}_{r,k}$. An $r$-simple path $W$ in $G$ is a {\em $\overline{\bf d}$-fit} if for every color $i\in\{1,2,\ldots,\mathsf{b}(k/r)\}$, the number of times $W$ visits vertices colored $i$ is exactly $d_i$. A function $\varphi: E(G)\rightarrow \mathbb{N}_0$ is a {\em $\overline{\bf d}$-fit} if {\em (i)} for every vertex $v\in V(G)$, $\sum_{e\in E(G): v\in e}\varphi(e)$ is an even number upper bounded by $2(r-d_i)$ where $i$ is the color of $v$, and {\em (ii)} $\sum_{i=1}^{\mathsf{b}(k/r)}d_i + \sum_{e\in E(G)}\varphi(e)\geq k$.
\end{definition}

We define the the \eulerEdge\ problem as follows.

\begin{defproblem}
{\eulerEdge}
{An $n$-vertex $\mathsf{b}(k/r)$-colored undirected graph $G$, positive integers $k,r$, and a vertex cover $U$ of $G$ of size at most $3k/r$ where each vertex in $U$ has a unique color, and an occurrence sequence $\overline{\bf d}=(d_1,d_2,\ldots,d_{\mathsf{b}(k/r)})\in {\cal D}_{r,k}$ where $d_i\geq 1$ for every color $i$ of a vertex in $U$.}
{Do there exist both an $r$-simple path $W$ in $G$ that is a $\overline{\bf d}$-fit and a function $\varphi: E(G)\rightarrow \mathbb{N}_0$ that is a $\overline{\bf d}$-fit?}
\end{defproblem}\\
Importantly, the two objects that we seek in the \eulerEdge\ problem are independent of each other (unlike the case of \eulerPart). Intuitively, the reason why we can allow this independence is precisely because we know that the walk is going to visit every vertex of a vertex cover, and hence no matter what the second object will be, we will necessarily obtain a connected multigraph at the end when we combine the two. Now, let us formalize this intuition. 

\begin{lemma}\label{lem:pairProblemSpecial}
Suppose that \eulerEdge\ can be solved in time $f(k/r)\cdot (n+\log k)^{\OO(1)}$. Then, \colPathRspecial\ can be solved in time $2^{\OO(k/r)}\cdot f(k/r)\cdot (n+\log k)^{\OO(1)}$.
\end{lemma}

\begin{proof}
Let $\cal A$ be an algorithm that solves \eulerEdge\ in time $f(k/r)\cdot (n+\log k)^{\OO(1)}$. We now describe how to solve \colPathRspecial. To this end, let $(G,k,r,U)$ be an instance of \colPathRspecial. For each $\overline{\bf d}\in {\cal D}_{r,k}$ such that $d_i\geq 1$ for every color $i$ of a vertex in $U$, we call ${\cal A}$ with $(G,k,r,U,\overline{\bf d})$ as input, and if $\cal A$ return \yes, so do we. At the end, if no call to $\cal A$ returned \yes, we return \no.

By Lemma \ref{lem:degreeSeqNumber}, $|{\cal D}_{r,k}|=2^{\OO(k/r)}$. Thus, it is clear that our algorithm runs in time $2^{\OO(k/r)}\cdot f(k/r)\cdot (n+\log k)^{\OO(1)}$. In what follows, we prove that our algorithm is correct.

In one direction, suppose that \colPathRspecial\ is a \yes-instance. Then, $G$ has a colorful $r$-simple $k$-path $P$ that visits all vertices in $U$ and which has fewer than $30(k/r)$ distinct edges.  Then, $P_{\mathrm{multi}}$ is a multigraph that has an Eulerian $(s,t)$-trail for some vertices $s,t\in V(G)$. From Lemma \ref{lem:eulerSpanning}, we derive that $P_{\mathrm{multi}}$ has a colorful $r$-simple $(s,t)$-walk $W$ of length shorter than $60(k/r)$ that visits every vertex visited by $P$. For every $\{1,2,\ldots,{\mathsf{b}(k/r)}\}$, let $d_i$ be the number of times vertices of color $i$ occur in $W$. Then, $W$ is a $\overline{\bf d}$-fit and necessarily, $d_i\geq 1$ for every color $i$ of a vertex in $U$. Moreover, since $W$ has length shorter than $60(k/r)$, $\overline{\bf d}=(d_1,d_2,\ldots,d_{\mathsf{b}(k/r)})\in {\cal D}_{r,k}$. Define $\varphi: E(G)\rightarrow\mathbb{N}_0$ as follows: for every edge $e\in E(G)$, let $\varphi(e)$ be the number of times $e$ is visited by $P$ minus the number of times it is visited by $W$. Since $P$ is a colorful $r$-simple $k$-path, it immediately follows that {\em (i)} $\sum_{e\in E(G): v\in e}\varphi(e)$ is bounded by $2(r-d_i)$ where $i$ is the color of $v$ for every vertex $v\in V(G)$, and {\em (ii)} $\sum_{i=1}^{\mathsf{b}(k/r)}d_i + \sum_{e\in E(G)}\varphi(e)\geq k$. Here, the claim that each sum $\sum_{e\in E(G): v\in e}\varphi(e)$ is even follows by Theorem \ref{prop:eulerUndirected} and since both $P$ and $W$ are Eulerian $(s,t)$-trails with respect to graphs on the same vertex set, which together imply that the parity of the number of occurrences of every vertex in $P$ and in $W$ is the same.

In the other direction, suppose that our algorithm returns \yes. Then, there exists $\overline{\bf d}\in {\cal D}_{r,k}$ such that $d_i\geq 1$ for every color $i$ of a vertex in $U$ and $(G,k,r,U,\overline{\bf d})$ is a \yes-instance. Then, there exist both an $r$-simple path $W$ in $G$ that is a $\overline{\bf d}$-fit and a function $\varphi: E(G)\rightarrow \mathbb{N}_0$ that is a $\overline{\bf d}$-fit. We need the following notations. First, let $s$ and $t$ denote the end-vertices of $W$. Let $R$ denote the vertices in~$G$ incident to at least one edge $e\in E(G)$ such that $\varphi(e)\geq 1$. In addition, let $H$ denote the multigraph whose vertex set consists of the vertices visited at least once by $W$ and the vertices in $R$, and whose edge multiset is defined as follows: for every edge $e\in E(G)$, the number of copies of $e$ in $H$ is the number of occurrences of $e$ in $W$ plus $\varphi(e)$.

Since $W$ is a $\overline{\bf d}$-fit and for every vertex $v\in V(G)$, $\sum_{e\in E(G): v\in e}\varphi(e)$ is an even number upper bounded by $2(r-d_i)$ where $i$ is the color of $v$, we have that in $H$, every vertex is incident to at most $2r$ edges, every vertex apart from $s$ and $t$ has an even degree, and $s$ and $t$ either both have even degrees or both have odd degrees. Moreover, since $W$ is a $\overline{\bf d}$-fit and $\sum_{i=1}^{\mathsf{b}(k/r)}d_i + \sum_{e\in E(G)}\varphi(e)\geq k$, we conclude that if $H$ has an Euler $(s,t)$-trail, then this trail is necessarily an $r$-simple $k$-path in $G$. To this end, by Theorem \ref{prop:eulerUndirected}, it remains to prove that $H$ is connected. For this purpose, first observe that since $d_i\geq 1$ for every color $i$ of a vertex in $U$, and $W$ is a $\overline{\bf d}$-fit, it holds that every two vertices in $U$ are not only present in $H$, but also connected by a path in $H$. Now, $H$ has no isolated vertices (by its definition), and $V(H)\setminus U$ is an independent set in $G$. Thus, since every edge in $H$ is a copy of an edge in $G$, it holds that every vertex in $V(H)\setminus U$ has (in $H$) at least one neighbor in $U$. This implies that $H$ is connected, and hence the proof is complete.
\end{proof}

Notice that the existence of an $r$-simple path $W$ in $G$ that is a $\overline{\bf d}$-fit can be easily tested by using DP. Indeed, we can just use a simplified version of the DP procedure in the proof of Lemma \ref{lem:eulerPartProblem} that avoids all calls to the external algorithm from Lemma \ref{lem:tw2CompDP} (since these calls only concern the construction of $H$); for the details, see Appendix \ref{sec:pseudocode}. Thus, we have the following observation.

\begin{observation}\label{lem:walkFit}
There is a $2^{\OO(k/r)}(n+\log k)^{\OO(1)}$-time algorithm that, given an instance $(G,k,r,U,\overline{\bf d})$ of \eulerEdge, determines whether $G$ has  an $r$-simple path $W$ that is a $\overline{\bf d}$-fit.
\end{observation}

\paragraph{Flow Network.} Finally, we construct a flow network to prove that the existence of a function $\varphi$ that is a $\overline{\bf d}$-fit can be tested in polynomial time.

\begin{lemma}\label{lem:funcFit}
There is a polynomial-time algorithm that, given an instance $(G,k,r,U,\overline{\bf d})$ of the \eulerEdge\ problem, determines whether there exists a function $\varphi: E(G)\rightarrow \mathbb{N}_0$ that is a $\overline{\bf d}$-fit.
\end{lemma}

\begin{proof}
To describe our algorithm $\cal A$, let $(G,k,r,U,\overline{\bf d})$ be an instance of \eulerEdge. For every vertex $v\in V(G)$, denote $c_v=r-d_i$ where $i$ is the color of $v$. In addition, denote $F=\sum_{v\in V(G)}c_v$ and $\ell=2(k-\sum_{i=1}^{\mathsf{b}(k/r)}d_i)$.
We construct a flow network $N$ with source $s$ and sink $t$ as follows.
\begin{itemize}
\item For every vertex $v\in V(G)$, insert (into $N$) two new vertices, $v_1$ and $v_2$, the arc $(v_1,v_2)$ of infinite (upper) capacity and cost $1$, and the arcs $(s,v_1)$ and $(v_2,t)$ both of (upper) capacity $c_v$ and cost $0$. 
\item For every edge $\{u,v\}\in E(G)$, insert (into $N$) the arcs $(u_1,v_2)$ and $(v_1,u_2)$ both of infinite (upper) capacity and cost $0$.
\end{itemize}
The lower capacity of each arc is simply $0$. We seek the minimum cost $C$ required to send $F$ units of (integral) flow from $s$ to $t$ in $N$. This task can be performed in polynomial time \cite{DBLP:books/daglib/0069809}. (We stress that $F$ and capacities are represented in binary, and the running time is polynomial in the size of this representation.) After performing this task, algorithm $\cal A$ checks whether $C\leq F-\ell$. (Intuitively, $C\leq F-\ell$ means that at least $\ell$ edges of cost $0$ between vertices indexed $1$ and $2$ must be used.) If this condition is satisfied, then $\cal A$ accepts, and otherwise it rejects.

Clearly, $\cal A$ runs in polynomial time, and it remains to show that our reduction is correct.

\medskip
\noindent{\bf First direction.} In one direction, suppose that there exists a function $\varphi: E(G)\rightarrow \mathbb{N}_0$ that is a $\overline{\bf d}$-fit. Let $H$ denote the multigraph on vertex set $V(G)$ and where every edge $e\in E(G)$ has multiplicity $\varphi(e)$. Since for every vertex $v\in V(G)$, $\sum_{e\in E(G): v\in e}\varphi(e)$ is an even number, by Theorem \ref{prop:eulerUndirected} we have that $H$ is Eulerian. In particular, we can direct it such that every vertex has in-degree equal to its out-degree, and denote the result by $\widehat{H}$. Now, we define a function $f: A(N)\rightarrow\mathbb{N}_0$ as follows: for every arc $a=(u_1,v_2)\in A(N)$, let $f(a)$ denote the multiplicity of $(u,v)$ in $\widehat{H}$. For each arc $a=(u_1,u_2)\in A(N)$, let $f(a)=c_u-\sum_{v\neq u:a'=(u_1,v_2)\in A(N)}f(a')$. All other arcs (i.e., arcs incident to $s$ or $t$) are assigned flow equal to their capacities. If $f$ is indeed a flow function, then it clearly sends $F$ units of flow (since all arcs incident to $s$ and $t$ have flow equal to their capacities). In addition, then the cost of $f$ is equal to its flow minus $\sum_{e\in E(G)}\varphi(e)$, that is, $F-\sum_{e\in E(G)}\varphi(e)$. Since $\sum_{i=1}^{\mathsf{b}(k/r)}d_i + \sum_{e\in E(G)}\varphi(e)\geq k$, the cost of $f$ is at most $F-\ell$.

It remains to prove that $f$ is a valid flow. It is immediate that the upper capacity constraints are satisfied, and that flow preservation constraints on vertices of the form $u_1$ are satisfied. Let us first verify that the lower capacity constraints are satisfied. To this end, we verify that the flow on each arc $a\in A(N)$ is non-negative. It suffices to consider an arc of the form $a=(u_1,u_2)\in A(N)$, else the claim is immediate. To show that $f(a)\geq 0$, we need to show that $\sum_{v\neq u:a'=(u_1,v_2)\in A(N)}f(a')\leq c_u$. This is equivalent to showing that $\sum_{a=(u,v)\in A(\widehat{H})}\mathrm{mul}(a) \leq r-d_i$ where $i$ is the color of $u$ in $G$ and $\mathrm{mul}(a)$ is the multiplicity of $a$ in $\widehat{H}$. Since for every vertex $v\in V(G)$, $\sum_{e\in E(G): v\in e}\varphi(e)$ is an even number upper bounded by $2(r-d_j)$ where $j$ is the color of $v$, it holds that $u$ is incident to at most $2(r-d_j)$ edges in $H$, and hence to at most $(r-d_j)$ outgoing arcs in $\widehat{H}$. Thus the inequality is satisfied.

Next, we prove that the flow preservation constraints on vertices of the form $u_2$ are satisfied. To this end, consider some vertex $u_2\in V(N)$. By the definition of $f$, we need to verify that $u_2$ receives flow of size exactly $c_u$ (since this is the amount of flow it sends to $t$). Observe that the amount of flow that $u_2$ receives is precisely 
\[\begin{array}{ll}
\sum_{v_1:a=(v_1,u_2)\in A(N)}f(a) & = f((u_1,u_2))+ \sum_{v\neq u:a=(v_1,u_2)\in A(N)}f(a)\\
& = \left(c_u-\sum_{v\neq u:a=(u_1,v_2)\in A(N)}f(a)\right) + \sum_{v\neq u:a=(v_1,u_2)\in A(N)}f(a).
\end{array}\]
Thus, we need to show that
\[\sum_{v\neq u:a=(u_1,v_2)\in A(N)}f(a) = \sum_{v\neq u:a=(v_1,u_2)\in A(N)}f(a).\]
However, this follows from the fact that in $\widehat{H}$, every vertex (and hence in particular $u$) has in-degree equal to its out-degree.

\medskip
\noindent{\bf Second direction.} In the other direction, suppose that $C\leq F-\ell$. Then, there exists a flow function $f: A(N)\rightarrow\mathbb{N}_0$ that sends $F$ units of flow from $s$ to $t$ and whose cost is at most $F-\ell$. We define a function $\varphi: E(G)\rightarrow\mathbb{N}_0$ as follows: for every edge $e=\{u,v\}\in E(G)$, let $\varphi(e)=f((u_1,v_2))+f((v_1,u_2))$. In what follows, we show that $\varphi$ is a $\overline{\bf d}$-fit.

Since $f$ sends $F$ units of flow, all arcs incident to $s$ and $t$ must transfer flow equal to their capacity. Due to the flow conservation constraints (and lower capacity $0$ constraints), for every vertex of the form $u_1\in V(N)$, it holds that
\[\sum_{v:a=(u_1,v_2)\in A(N)}f(a) = c_u.\]
In addition, due to the flow conservation constraints (and lower capacity $0$ constraints), for every vertex of the form $u_2\in V(N)$, it holds that
\[\sum_{v:a=(v_1,u_2)\in A(N)}f(a) = c_u.\]
From this, we have that for every vertex $u\in V(G)$, it holds that
\[\begin{array}{ll}
\sum_{v:e=\{u,v\}\in E(G)}\varphi(e) &= \sum_{v\neq u:a=(u_1,v_2)\in A(N)}f(a) + \sum_{v\neq u:a=(v_1,u_2)\in A(N)}f(a) - 2f((u,u))\\
& = 2(c_u-f((u,u))) = 2(r-d_i-f((u,u))),
\end{array}\]
where $i$ is the color of $u$. Thus, for every vertex $u\in V(G)$, we have that $\sum_{e\in E(G): v\in e}\varphi(e)$ is an even number upper bounded by $2(r-d_i)$ where $i$ is the color of $v$.

To conclude that $\varphi$ is a $\overline{\bf d}$-fit, it remains to show that $\sum_{i=1}^{\mathsf{b}(k/r)}d_i + \sum_{e\in E(G)}\varphi(e)\geq k$. This is equivalent to showing that $\sum_{e\in E(G)}\varphi(e)\geq \ell$. Recall that the cost of $f$ is at most $F-\ell$ and it send $F$ units of flow from $s$ to $t$. Thus, since the cost of arcs of the form $(u_1,v_2)\in A(N)$ is $0$ if $v\neq u$ and $1$ otherwise, we have that $f$ must send at least $\ell$ units of flow through arcs of the form $(u_1,v_2)\in A(N)$ where $v\neq u$. However, $\sum_{e\in E(G)}\varphi(e)$ is precisely the amount of flow $f$ sends through arcs of the form $(u_1,v_2)\in A(N)$ where $v\neq u$. Thus, the proof is complete.
\end{proof}

\paragraph{Conclusion of the Proof.} We are ready to prove Lemma \ref{lem:undirectedPathFPTspecial}.

\begin{proof}[Proof of Lemma \ref{lem:undirectedPathFPTspecial}]
By Observation \ref{lem:walkFit} and Lemma \ref{lem:funcFit}, \eulerEdge\ is solvable in polynomial time. Thus, by Lemma \ref{lem:pairProblemSpecial}, \colPathRspecial\ can be solved in time $2^{\OO(k/r)}\cdot (n+\log k)^{\OO(1)}$. In turn, by Lemma \ref{lem:colorProblemSpecial}, this means that \undiPathR\ can be solved in time $2^{\OO(k/r)}\cdot f(k/r)\cdot (n+\log k)^{\OO(1)}$.
\end{proof}

\section{$p$-Set $(r,q)$-Packing: FPT}\label{sec:packingFPT}

\newcommand{\cH}{{\mathcal{H}}}
\newcommand{\cF}{{\mathcal{F}}}

Recall that in the \setPackingR\ problem, the input consists of a ground set $V$, positive integers $p,q,r$, and a collection $\cH$ of sets of size  $p$ whose elements belong to $V$. The goal is to decide whether there exists a subcollection of $\cH$ of size $q$ where each element occurs at most $r$ times. Note that $\cH$ can contain copies of the same set, i.e., not all elements of $\cH$ are {\em distinct} sets.
In this section, we will show that \setPackingR\  parameterized by $\kappa=pq/r$ is FPT. This result is in sharp contrast with that for \multisetPackingR , where the elements of $\cH$ may be multisets rather than just sets. In Section \ref{sec:W1}, we will prove that \multisetPackingR\ parameterized by $\kappa$ is W[1]-hard.

In what follows, for convenience we will study a slight generalization of \setPackingR\ by allowing sets of $\cH$ to be of size {\em at most} $p.$

Let us consider an instance $(\cH, q, r)$ of \setPackingR, and denote $m=|\cH|$. Observe that if $q\le r$, then \setPackingR\ is trivial. Thus, in the rest of this section, we assume that $q>r$ and hence $p<\kappa.$

We show a reduction of the set-packing instance $(\cH, q, r)$ to a situation where the
ground set has size bounded by $f(\kappa).$  The reduction uses a tool
known as \emph{representative sets} to discard irrelevant parts of the
instance. Representative sets have important applications for both  
FPT algorithms~\cite{DBLP:journals/jacm/FominLPS16} and kernels~\cite{KratschW12FOCS};
see also~\cite[Ch.~12]{DBLP:books/sp/CyganFKLMPPS15}.
The full power of the tool emerges in a \emph{matroid} setting 
(see Lov\'asz~\cite{Lovasz1977} and Marx~\cite{Marx09-matroid}),
but we need only a restricted setting, which we summarize as follows. 
This follows from Theorem 1.1 of~\cite{DBLP:journals/jacm/FominLPS16}
when applied to the special case of \emph{uniform matroids}.
Note that a linear representation of a uniform matroid can be computed in
deterministic polynomial time using a Vandermonde
matrix~\cite[Section~2.5]{DBLP:journals/jacm/FominLPS16}, hence the
theorem can be applied.

\begin{theorem}[\cite{DBLP:journals/jacm/FominLPS16}]\label{thm:computeRepSets}
  Let $V$ be a ground set and $\cH$ a collection of $p$-sets
  in $V$. Let $\kappa \in \mathbb{N}$. 
  In time $(\binom{p+\kappa}{p} + |\cH|)^{\OO(1)}$
  we can compute a collection $\cH^* \subseteq \cH$
  with $|\cH^*| \leq \binom{p+\kappa}{p}$ 
  such that the following holds: For every $\kappa$-set $B \subseteq V$,
  there exists a set $A \in \cH$ disjoint from $B$
  if and only if there exists such a set $A \in \cH^*$. 
\end{theorem}

We refer to $\cH^*$ as a \emph{$\kappa$-representative set} (or
\emph{representative family}) of $\cH$, although technically,
$\cH^*$ is \emph{representative for $\cH$ 
in the uniform matroid $U_{n,\kappa+p}$}, where $n=|V|$.
See~\cite{DBLP:books/sp/CyganFKLMPPS15} for details.

Given this result, we need only two simple reduction rules.
 
\begin{redrule}\label{rule:disc}
  Discard any element that occurs at most~$r$
  times.
  Exclude any empty sets and reduce $q$ by the number of empty sets. 
\end{redrule}

\begin{lemma} 
Reduction Rule~\ref{rule:disc} is sound.
\end{lemma}
\begin{proof}
Assume that the original instance has a subcollection of size $q$ where each element occurs at most $r$ times. Clearly, the reduction rule does not increase the number of sets in the subcollection and each element occurs at most $r$ times. Now we prove the opposite direction.
Let $q'$ be $q$ minus the number of empty sets and suppose that the reduced instance is positive, i.e., there is a subcollection ${\cH}'$ of size $q'$ where each element occurs at most $r$ times.
By adding elements discarded from the sets in ${\cH}'$ and the $q-q'$ empty sets, we obtain the required subcollection for the original instance. 
\end{proof}

We now have $m > nr/p$, i.e., $n < mp/r$. Our second rule will
decrease the value of $m$. 

\begin{redrule} \label{rule:pqr2}
  Pad $\cH$ to be $p$-uniform using dummy elements for smaller sets. 
  Compute and put aside $q$ disjoint $\kappa$-representative sets 
 as follows. Let $\cH_1=\cH$; for $i=1, \ldots, q$ compute
    a representative set $\cH_i^* \subseteq \cH_i$ 
    in the uniform matroid $U_{n,\kappa+p}$ 
    using Theorem~\ref{thm:computeRepSets};
    and let $\cH_{i+1}=\cH_i \setminus \cH_i^*$. 
    Finally discard any sets remaining in $\cH_{q+1}$. 
\end{redrule}

\begin{lemma} \label{lemma:rule-pqr2}
  Reduction Rule~\ref{rule:pqr2} is sound and leaves 
  at most 
  \[
  m \leq q \binom{\kappa+p}{p}\le q 4^{\kappa}
  \] 
  sets. The rule can be applied in time polynomial
  in the input size and $\binom{\kappa+p}{p}$. 
\end{lemma}
\begin{proof}
  Each representative set has size at most $\binom{\kappa+p}{p}$, 
  hence the total size of the output is 
  $
  m' \leq q \binom{\kappa+p}{p}.
  $
  Since $p<\kappa$, we have $\binom{\kappa+p}{p}\le 4^{\kappa}$ implying $m'\le q 4^{\kappa}.$
  
  We will argue correctness. Let $\cH'$ be the instance produced.
  Clearly, if $\cH'$ is positive, then so is $\cH$.
  Now assume that $\cH$ is positive, and let $\cF \subseteq \cH$,
  $|\cF|=q$, be a solution with maximum intersection with $\cH'$.
  Assume that there exists a set $E \in \cF \setminus \cH'$. Let $X$ be the 
  set of elements that occur precisely $r$ times in $\cF$, 
  and let $X'=X \setminus E$. Thus $|X'| \leq |X| \leq \kappa$. 
  Then, since $E \notin \cH'$,
  each representative family contains at least one set $E'$
  disjoint from $X'$, i.e., $q$ alternative sets $E'$ in total. Since
  $|\cF\setminus\{E\}|<q$, for at least one such set $E'$ it also holds that $E' \notin \cF$. 
  Then $(\cF\setminus\{E\})\cup\{E'\}$ is a packing of $q$ sets, where every element occurs
  in at most $r$ sets, and with a larger intersection with $\cH'$ than
  $\cF$, which contradicts that $\cF$ was maximal. Thus $\cF \subseteq \cH'$,
  and the output instance is positive. The running time follows from
  the computation of a representative set (see Theorem \ref{thm:computeRepSets}).
\end{proof}

In fact, these two simple rules give us a trivial parameter setting. 

\begin{lemma} \label{lemma:pqr-ground-set-reduction}
  Assume that the two rules have been applied exhaustively. Then 
  $n < f(pq/r)$ where $f(\kappa) = \kappa4^{\kappa}$.
\end{lemma}
\begin{proof}
  On the one hand, since every element of the ground set occurs in
  more than $r$ sets of the input, there are $m > rn/p$ sets in the
  input, hence $n < mp/r$. On the other hand, by the representative
  sets reduction we have $m < q4^{pq/r}$. Then
  \[
  n < mp/r < q4^{pq/r}p/r = f(pq/r).
  \]
\end{proof}

It is now easy to solve the problem via an application of an ILP solver.

\begin{lemma} \label{lemma:pqr-ilp-alg}
  An instance of \setPackingR\ on a ground set of size $n$
  can be solved in time $\OO(n^{\OO(pn^p)}\log q).$
\end{lemma}
\begin{proof}
Let   $M^*$ be the collection of all distinct sets in the input and let $m^*=|M^*|.$  
  Then $m^*=\OO(n^p)$. 
To write an instance of {\sc Feasibility ILP}  that encodes the problem,  let us introduce $m^*$ variables $x_E$ ($E\in M^*$) denoting the number of copies of $E$ to use in the
 solution. The constraints are as follows: (a) for each $E\in M^*$, $X_E\le \mu_E,$ the multiplicity
 of $E$ in the input; (b) $\sum_{E\in M^*_v} X_E\le r$ for each $v\in V,$ where $M^*_v=\{E\in M^*:\ v\in E\};$  (c) $\sum_{E\in M^*} X_E\ge q.$
 Thus, the number of variables, constraints and size of the {\sc Feasibility} ILP instance are $m^*$, $m^*+n+1$ and (since $q>r$) $\OO(n(m^*+\log q))=\OO(n^{p+1}+n\log q),$ respectively. 
Hence, by Theorem \ref{prop:ILP}, we can solve the instance in time  $\OO(n^{\OO(pn^p)}\log q).$  
\end{proof}

Now we can obtain the main result of this section.

\begin{theorem}
\setPackingR\ parameterized by $\kappa$ is FPT.
\end{theorem}
\begin{proof}
Recall that $p<\kappa$. We may assume that our instance of \setPackingR\  has been reduced by the two reduction rules above.
By Lemma \ref{lemma:pqr-ground-set-reduction}, $n<\kappa4^{\kappa}$. Thus, by Lemma \ref{lemma:pqr-ilp-alg}, \setPackingR\ parameterized by $\kappa$ is FPT.
\end{proof}

We observe that the same reduction gives a polynomial kernel when $p$
is a constant.

\begin{theorem}\label{thm:gkernel}
  The \setPackingR\ problem for constant $p$ has a polynomial-time
  reduction to a ground set of size $\OO((q/r)^{p+1})$ and a generalized
  polynomial kernel of $\OO((q/r)^{p^2+p}\log r) = 
  \OO((q/r)^{2(p^2+p)} \log (q/r))$ bits. 
\end{theorem}
\begin{proof}
  By Lemma~\ref{lemma:rule-pqr2}, if the reduction rules have been
  applied then the number of sets is bounded by
  \[
  m \leq q \binom{pq/r+p}{p} \leq q(pq/r+p)^p,
  \]
  and as in Lemma~\ref{lemma:pqr-ground-set-reduction} with $p$ a
  constant we get $n = \OO(m/r)$. Putting them together,
  we get $n = \OO((q/r)^{p+1})$. This gives the first result.
  For the latter, we may observe that the reduction produces a
  multiset where at most 
  \[
  m^* \leq (n+1)^p = \OO((q/r)^{p(p+1)})
  \]
  distinct sets are possible (since sets have size at most $p$).
  Hence the instance can be described by giving the multiplicity in
  the input for each set type, keeping only the first $r$ copies of
  each set. This gives a description with $m^* \log r$ bits. 
  Finally, we note that $r \leq q \leq m$ and that the input instance
  of \setPackingR\ is coded without multiplicities; hence $r$ is
  bounded by the total input size. If the total input size 
  is at least $2^{m^* \log m^*}$ then we can solve the problem
  completely in polynomial time, otherwise we have 
  $\log r \leq m^* \log m^*$.
\end{proof}

We will use the following simple lemma.

\begin{lemma}[\cite{AlonGKSY11}]\label{lem:GKtoK}
Let $L,L'$ be a pair of decidable parameterized problems such that $L'$ is
 in \NP,  and $L$  is \NP-complete.  If there is a  general kernelization from $L$  to $L'$
producing a generalized kernel of polynomial size,  then $L$  has a polynomial-size kernel.
\end{lemma}

Theorem  \ref{thm:gkernel} and Lemma \ref{lem:GKtoK} imply the following:

\begin{corollary}
The \setPackingR\ problem for constant $p$ admits a polynomial size kernel.
\end{corollary}

Let us finally complement Theorem \ref{thm:gkernel} by showing that the lower bound for
$r=1$ carries over to the parameter $q/r$ for arbitrary values of $r$.

\begin{theorem}
  The \setPackingR{} problem with fixed value of $p \geq 3$ does not
  admit a generalized kernel of size $\OO((q/r)^{p-\varepsilon})$ for
  any $\varepsilon > 0$ unless the polynomial hierarchy collapses.
\end{theorem}
\begin{proof}
  Dell and Marx~\cite{DellM12packing} showed that \textsc{Perfect $p$-Set
    Matching} (i.e., the variant where $r=1$ and $n=pq$) does not admit a
  generalized kernel of $\OO(q^{p-\varepsilon})$ bits for any
  $\varepsilon > 0$ unless the polynomial hierarchy collapses. We show  
  a parameter-preserving reduction from the case of $r=1$ to the
  arbitrary case. Let $\cH$ be the input to an instance of 
  \textsc{Perfect $p$-Set Matching} where $\cH \subseteq 2^V$ is a 
  $p$-uniform hypergraph over some ground set $V$, $|V|=n=pq$.
  We produce an output instance of \setPackingR\ by padding $\cH$ 
  with $(r-1)n$ sets, each of which is incident with precisely one
  member of $V$ and which in total cover every element of $V$
  precisely $r-1$ times. (We pad these sets with arbitrary dummy
  elements to produce a $p$-uniform output.) We set $q'=q+(r-1)n$.
  We claim that the output has a $(q',r)$-packing if and only if $\cH$
  contains a $q$-packing.  This is not hard to see. On the one hand,
  any $q$-packing in $\cH$ can be padded to a $q'$-packing in the
  output by including all the padding sets; on the other hand, for any
  ($q', r)$-packing where some element $v \in V$ is covered by two
  non-padding sets, we can get a different $(q', r)$-packing by 
  discarding one set from $\cH$ and replacing it by a further padding
  set covering $v$. The value of $p$ is unchanged. Finally, since
  $q < n$ we have $q'=q+(r-1)n < rn$, hence $q'/r < n = pq = \OO(q)$,
  and the parameter is only increased by a constant factor.
\end{proof}

\section{$(r,k)$-Monomial Detection: para-NP-Hardness}\label{sec:monomDet}

In this section, we prove that if $k$ is not polynomially bounded in the input size, even an \XP\ algorithm for the special case of \monomDetR\ where only two distinct variables are present is out of reach. For this purpose, we present a reduction from the {\sc Partition} problem, which is known to be \NPH~\cite{GJ79}. In this problem, we are given a multiset $M$ of positive integers, and the goal is to determine whether $M$ can be partitioned into two multisets, $M_1$ and $M_2$, such that the sum of the integers in $M_1$ is equal to the sum of the integers in $M_2$.

\begin{theorem}\label{thm:monomParaNPhard}
\monomDetR\ is \paraH\ parameterized by $k/r$ even if the number of distinct variables is $2$ and the circuit is non-canceling.
\end{theorem}

\begin{proof}
To prove this theorem, we give a reduction from {\sc Partition} to \monomDetR\ parameterized by $k/r$. To this end, let $M$ be an instance of {\sc Partition}. We define our set of variables as $\{x,y\}$ (that is, we have only two variables), and we define a polynomial {\sf POL} as follows:
\[\displaystyle{\mathsf{POL} = \sum_{M'\subseteq M}\left((\prod_{n\in M'}x^n)\cdot(\prod_{n\in M\setminus M'}y^n)\right)}.\]
We define $k=\sum_{n\in M}n$, and $r=k/2$. Then $k/r=2$  and our reduction shows that  \monomDetR\ is \NPH for $k/r=2$ which implies that \monomDetR\ is  \paraH\ by
\cite[Theorem 2.14]{DBLP:series/txtcs/FlumG06}.

We now prove that $M$ is a \yes-instance of {\sc Partition} if and only if {\sf POL} has a monomial of degree $k$ where each variable has degree at most $r$. To this end, notice that $M$ is a \yes-instance of {\sc Partition} if and only if there exists $M'\subseteq M$ such that $\sum_{n\in M'}n=k/2$. Now, for any $M'\subseteq M$, the following statement holds: $\sum_{n\in M'}n=k/2$ if and only if $(\prod_{n\in M'}x^n)\cdot(\prod_{n\in M\setminus M'}y^n)$ is a monomial (of degree $k$) where each variable has degree at most $r$. However, by the definition of {\sf POL}, the latter part of the statement is true if and only if {\sf POL} has a monomial of degree $k$ where each variable has degree at most $r$.

Next, we show that {\sf POL} can be encoded by a non-canceling arithmetic circuit of size polynomial in $\log k$. To this end, denote $M=\{n_1,n_2,\ldots,n_{\ell}\}$ where $\ell=|M|$, and let $n^\star$ be the largest number that occurs at least once in $M$. Then, for all $z\in\{x,y\}$ and $i\in\{0,1,\ldots,\lfloor\log_2 n^\star\rfloor\}$, we have a gate $\widehat{g}_{z,i}$ defined recursively as follows. First, for all $z\in\{x,y\}$, we set $\widehat{g}_{z,0}$ to be the input gate $z$. Second, for all $z\in\{x,y\}$ and  $i\in\{1,2,\ldots,\lfloor\log_2 n^\star\rfloor\}$, we set $\widehat{g}_{z,i} = \widehat{g}_{z,i-1}^2$. By simple induction on $i$, for all $z\in\{x,y\}$ and $i\in\{0,1,\ldots,\lfloor\log_2 n^\star\rfloor\}$, it holds that $\widehat{g}_{z,i}$ encodes $z^{2^i}$. Now, for all $z\in\{x,y\}$ and $n\in M$, we have a gate $g_{z,n}$ defined as follows:
\[\displaystyle{g_{z,n}=\prod_{i\in\{1,2,\ldots,\lfloor\log_2 n^\star\rfloor\} \atop \mathrm{s.t.}~\mathrm{digit}(n,i)=1}\widehat{g}_{z,i}},\]
where digit$(n,i)$ is the $i$-th least significant digit of $n$ when encoded in binary.
Then, for all $z\in\{x,y\}$ and $n\in M$, we have that $g_{z,n}$ encodes $z^n$.

For all $i\in\{1,2,\ldots,\ell\}$, we have gates $h'_i$ and $h_i$ defined recursively as follows. First, we set $h_1=h'_1=g_{x,n_1} + g_{y,n_1}$. Second, for all $i\in\{2,3,\ldots,\ell\}$, we set $h'_i=g_{x,n_i} + g_{y,n_i}$ and $h_i = h_{i-1}\cdot h'_i$. By simple induction on $i$, we have that for all $i\in\{1,2,\ldots,\ell\}$, $h_i$ encodes the following polynomial:
\[\displaystyle{\sum_{M'\subseteq \{n_1,n_2,\ldots,n_i\}}\left((\prod_{n\in M'}x^n)\cdot(\prod_{n\in \{n_1,n_2,\ldots,n_\ell\}\setminus M'}y^n)\right)}.\]
Thus, $h_\ell$ encodes {\sf POL}.

Finally, we argue that \monomDetR\ is \paraH\ parameterized by $k/r$. Suppose, by way of contradiction, that this claim is false. Then, \monomDetR\ admits an algorithm, say $\cal A$, that runs in time $|I|^{f(k/r)}$ on input $I$ for some function $f$ that depends only on $k/r$. Thus, we can solve any instance $M$ of {\sc Partition} by using the reduction above to construct (in polynomial time) an equivalent instance $I$ of \monomDetR, and then calling $\cal A$ with $I$. However, the parameter $k/r$ equals $2$ (since $r=k/2$), and hence $|I|^{f(k/r)}=|I|^{\OO(1)}$, that is, we solve {\sc Partition} in polynomial-time.  Since {\sc Partition} is \NPH, we have reached a contradiction. This completes the proof. 
\end{proof}

\section{$p$-Multiset $(r,q)$-Packing and $(r,k)$-Monomial Detection: W[1]-Hardness}\label{sec:W1}

In this section, we prove that \multisetPackingR\ is \WOH. To prove this theorem, we present a reduction from the {\sc Multicolored Clique} problem, which is known to be \WOH~\cite{DBLP:journals/jcss/Pietrzak03,DBLP:journals/tcs/FellowsHRV09}. In this problem, we are given a vertex-colored graph $G$ and a positive integer $k$, where each vertex has a color in $\{1,2,\ldots,k\}$, and our goal is to decide whether $G$ has a {\em multicolored $k$-clique}, that is, a clique with $k$ vertices where each vertex has a distinct color. Later in this section, we show that our theorem implies that a restricted case of \monomDetR\ is \WOH\ as well.

\begin{theorem}\label{thm:multiPackW1hard}
\multisetPackingR\ is \WOH\ parameterized by $pq/r$ even if the size of the universe is $pq/r$.
\end{theorem}
\begin{proof}
Our source problem is {\sc Multicolored Clique}. Given an instance $(G,k)$ (each vertex of $G$ is assigned a color from $\{1,\ldots,k\}$) of {\sc Multicolored Clique}, we construct an instance $(U,{\cal S},p,q,r)$ of \multisetPackingR\ as follows. For each color $i\in \{1,\ldots,k\}$, let $C^i$ be the set of vertices in $G$ whose color is $i$. Let $n$ denote the size of a color class, that is, $n=|C^i|$ for any $i\in\{1,\ldots,k\}$. Moreover, for all $i\in\{1,\ldots,k\}$, denote $C^i=\{v^i_1,v^i_2,\ldots,v^i_n\}$. Define $r=n$, $p=kn$ and $q=k+{k\choose 2}$. Note that $pq/r = (kn)\left(k+{k \choose 2}\right)/n = k\left(k+{k \choose 2}\right)$.

The universe $U$ contains the following distinct elements:
\begin{itemize}
\item For each color $i\in \{1,\ldots,k\}$, we have an element $c^i$.
\item For each pair $(i,j)\in \{1,\ldots,k\}\times\{1,\ldots,k\}$ with $i\neq j$, we have an element $c^{i\rightarrow j}$ and an element $\widehat{c}^{i\rightarrow j}$.
\item For each pair $(i,j)\in \{1,\ldots,k\}\times\{1,\ldots,k\}$ with $i<j$, and for each $t\in\{1,\ldots,k-2\}$, we have an element $c^{(i,j)}_t$.
\end{itemize}
Observe that $|U| = k + 2k(k-1) + {k \choose 2}(k-2) = k\left(k + {k \choose 2}\right) = pq/r$.

Now, we construct $\cal S$ as follows.
\begin{itemize}
\item For each color $i\in\{1,\ldots,k\}$ and for each $x\in\{1,\ldots,n\}$, we insert the multiset
\[M^i_x = \displaystyle{\{[n] c^i\}\cup\left(\bigcup_{j\in \{1,\ldots,k\}\setminus\{i\}}\{[x]c^{i\rightarrow j},[n-x]\widehat{c}^{i\rightarrow j}\}\right)}.\]
Note that $|M^i_x|=n+(k-1)n=p$.
\item For each edge $e=\{v^i_x,v^j_y\}\in E(G)$ (where $v^i_x\in C^i$ and $v^j_y\in C^j$) with $i<j$, we insert the multiset 
\[M^{(i,j)}_{(x,y)} = \displaystyle{\left(\bigcup_{t\in\{1,\ldots,k-2\}}\{[n] c^{(i,j)}_t\}\right)\cup\{[n-x]c^{i\rightarrow j},[x]\widehat{c}^{i\rightarrow j},[n-y]c^{j\rightarrow i},[y]\widehat{c}^{j\rightarrow i}\}}.\]
Note that $|M^{(i,j)}_{(x,y)}|=(k-2)n+2n=p$.
\end{itemize}

\noindent{\bf Proof of Correctness.} In the forward direction, we suppose that we have a multicolored $k$-clique~$K$ in $G$. Let $v^i_{\phi(i)}$ be the (unique) vertex in $C^i$ that belongs to~$K$. Then, it holds that the subcollection ${\cal S}':=\{M^i_{\phi(i)}: i\in\{1,\ldots,k\}\}\cup\{M^{(i,j)}_{(\phi(i),\phi(j))}: (i,j)\in \{1,\ldots,k\}\times\{1,\ldots,k\},i<j\}$ of $\cal S$ is an $r$-relaxed packing of size $q$. (To see that this claim is true, observe that each element in $U$ occurs in this subcollection precisely $n$ times.)

In the reverse direction, we suppose that we have a subcollection ${\cal S}'$ of ${\cal S}$ that is an $r$-relaxed packing of size $q$. Then, we first observe that ${\cal S}'$ can contain at most one multiset from $\{M^i_1,\ldots,M^i_n\}$ for each $i\in\{1,\ldots,k\}$ (since otherwise the element $c^i$ occurs more than $r$ times), and at most one multiset from $\{M^{(i,j)}_{(x,y)}: \{v^i_x,v^j_y\}\in E(G)\}$ for each $(i,j)\in \{1,\ldots,k\}\times\{1,\ldots,k\}$ with $i<j$ (since otherwise the element $c^{(i,j)}_1$ occurs more than $r$ times). Then, because $|{\cal S}'|=q$, we have that ${\cal S}'$ contains exactly one multiset from $\{M^i_1,\ldots,M^i_n\}$ for each $i\in\{1,\ldots,k\}$, and exactly one multiset from $\{M^{(i,j)}_{(x,y)}: \{v^i_x,v^j_y\}\in E(G)\}$ for each $(i,j)\in \{1,\ldots,k\}\times\{1,\ldots,k\}$ with $i<j$. In particular, this means that it is well defined to let $\phi(i)$, $i\in\{1,\ldots,k\}$, denote the integer $x$ such that $M^i_x\in{\cal S}'$. Moreover, it is well defined to let $\varphi(i,j)$, $(i,j)\in \{1,\ldots,k\}\times\{1,\ldots,k\}$ with $i<j$, denote the pair $(x,y)$ such that $M^{(i,j)}_{(x,y)}\in{\cal S}'$.

Define $K=\{v^1_{\phi(1)},\ldots,v^k_{\phi(k)}\}$. Then, we claim that $K$ is a multicolored $k$-clique in $G$. It is clear that $|K|=k$ and that $K$ is multicolored. Thus, it remains to show that for each $(i,j)\in \{1,\ldots,k\}\times\{1,\ldots,k\}$ with $i<j$, it holds that $\{v^i_{\phi(i)},v^j_{\phi(j)}\}\in E(G)$. For this purpose, we arbitrarily select $(i,j)\in \{1,\ldots,k\}\times\{1,\ldots,k\}$ with $i<j$. To show that $\{v^i_{\phi(i)},v^j_{\phi(j)}\}\in E(G)$, it suffices to show that $\varphi(i,j)=(\phi(i),\phi(j))$. Let us denote $\varphi(i,j)=(x,y)$. We only show that $x=\phi(i)$, since the proof that $y=\phi(j)$ is symmetric. Suppose, by way of contradiction, that $x\neq \phi(i)$. We consider two cases.
\begin{itemize}
\item First, suppose that $x<\phi(i)$. Note that $c^{i\rightarrow j}$ occurs $\phi(i)$ times in $M^i_{\phi(i)}$, and it occurs $n-x$ times in $M^{(i,j)}_{\varphi(i,j)}$. However, $\phi(i)+(n-x) > n$, which implies that $c^{i\rightarrow j}$ occurs more than $r$ times in ${\cal S}'$. Thus, we have reached a contradiction.
\item Second, suppose that $x>\phi(i)$. Note that $\widehat{c}^{i\rightarrow j}$ occurs $n-\phi(i)$ times in $M^i_{\phi(i)}$, and it occurs $x$ times in $M^{(i,j)}_{\varphi(i,j)}$. However, $(n-\phi(i))+x > n$, which implies that $\widehat{c}^{i\rightarrow j}$ occurs more than $r$ times in ${\cal S}'$. Thus, we have reached a contradiction.
\end{itemize}
This completes the proof.
\end{proof}

Our reduction heavily relies on the inclusion of input instances that contain multisets rather than sets. In particular, it does not rule out the possibility that \setPackingR\ is \FPT\ parameterized by $(pq)/r$---that is, this proof does not contradict Section \ref{sec:packingFPT}.

As a consequence of Theorem \ref{thm:multiPackW1hard}, we obtain the following theorem.

\begin{theorem}\label{thm:monomW1hard}
\monomDetR\ is \WOH\ parameterized by $k/r$ even if {\em (i)} $k$ is polynomially bounded in the input length, {\em (ii)} the number of distinct variables is {\color{brown} at most} $k/r$, and  {\em (iii)} the circuit is non-canceling.
\end{theorem}

\begin{proof}
The proof of this theorem is based on the standard encoding of set packing problems using multivariate polynomials (see, e.g., \cite{DBLP:journals/talg/KoutisW16}). For the sake of completeness, we present the details. By Theorem \ref{thm:multiPackW1hard}, it suffices to give a reduction from \multisetPackingR\ with $|U|\leq (pq)/r$. To this end, let $(U,{\cal S},p,q,r)$ be an instance of \multisetPackingR\ with $|U|\leq (pq)/r$. Since $p$ is the size of each multiset in the input, it is polynomial in the input size. Moreover, $q$ (and hence also $r$) can be assumed to be polynomial in the input size, since if $q>|{\cal S}|$, then we have a \no-instance. 

We define our set of variables as $X=\{x_u: u\in U\}$ (that is, we have one variable for each element in $U$), and we define a polynomial {\sf POL} as follows:
\[\displaystyle{\mathsf{POL} = \sum_{{\cal S}'\subseteq{\cal S} \atop \mathrm{s.t.}~|{\cal S'}|=q}\prod_{M\in{\cal S}'}\prod_{u\in M}x_u}.\]

Define $k=pq$. For any choice of non-negative integers $d_u$ for each $u\in U$ whose sum is $k$, it holds that {\sf POL} has $\prod_{u\in U}x_{u}^{d_u}$ as a monomial if and only if there exists a subcollection ${\cal S}'\subseteq{\cal S}$ of size $q$ where each element $u\in U$ occurs exactly $d_u$ times. Thus, $(U,{\cal S},p,q,r)$ is a \yes-instance of \multisetPackingR\  if and only if {\sf POL} has a monomial (of total degree $k$) where the degree of each variable is at most $r$. 

Since $k=pq$ and $p$ and $q$ are polynomially bounded in the input length, so is $k$ showing (i). Since $k=pq$ and $|U|\le (pq)/r,$ we have $|U|\le k/r$ proving (ii). 
To complete the proof, it remains to show that {\sf POL} can be encoded by an arithmetic circuit of polynomial size. For this purpose, denote ${\cal S}=\{M_1,M_2,\ldots,M_{\ell}\}$ where $\ell=|{\cal S}|$. For each $M\in{\cal S}$, we have a gate $g_{M}$ which is the multiplication $\prod_{u\in M}x_u$. Now, for all $i\in\{1,2,\ldots,\ell\}$ and $j\in\{1,2,\ldots,q\}$, we have a gate $g_{i,j}$ that is defined as follows.
\begin{itemize}
\item If $j=1$, then $g_{i,j} = \sum^i_{t=1} g_{M_i}$ for all $i\in\{1,2,\ldots,\ell\}$.
\item If $i=1$ and $j>1$, then $g_{i,j} = 0$.
\item If $i>1$ and $j>1$, then $g_{i,j} = g_{i-1,j} + g_{i-1,j-1}\cdot g_{M_i}$.
\end{itemize}
The output of the arithmetic circuit is given by $g_{\ell,q}$.

To see that the circuit above encodes {\sf POL}, we claim that for all $i\in\{1,2,\ldots,\ell\}$ and $j\in\{1,2,\ldots,q\}$, it holds that
\[\displaystyle{g_{i,j} = \sum_{{\cal S}'\subseteq\{M_1,M_2,\ldots,M_i\} \atop \mathrm{s.t.}~|{\cal S'}|=j}\prod_{M\in{\cal S}'}\prod_{u\in M}x_u}.\]
The proof is by induction. In the basis, where $i=1$ or $j=1$, the claim clearly holds. Now, suppose that the claim holds for $i-1\geq 1$, and let us prove it for $i$. Then, by the inductive hypothesis,
\[\begin{array}{ll}
g_{i,j} & = g_{i-1,j} + g_{i-1,j-1}\cdot g_{M_i}\\
& = \displaystyle{\sum_{{\cal S}'\subseteq\{M_1,M_2,\ldots,M_{i-1}\} \atop \mathrm{s.t.}~|{\cal S'}|=j}\prod_{M\in{\cal S}'}\prod_{u\in M}x_u + \left(\sum_{{\cal S}'\subseteq\{M_1,M_2,\ldots,M_{i-1}\} \atop \mathrm{s.t.}~|{\cal S'}|=j-1}\prod_{M\in{\cal S}'}\prod_{u\in M}x_u\right)\cdot g_{M_i}}\\
& = \displaystyle{\sum_{{\cal S}'\subseteq\{M_1,M_2,\ldots,M_i\} \atop \mathrm{s.t.}~|{\cal S'}|=j}\prod_{M\in{\cal S}'}\prod_{u\in M}x_u}.
\end{array}\]
This completes the proof.
\end{proof}

In light of Theorems~\ref{thm:multiPackW1hard} and \ref{thm:monomW1hard}, the reader might wonder whether \multisetPackingR\ and the special case of \monomDetR\ where $r$ is polynomially bounded by the input size are at least in \XP. However, this question has already been resolved positively---the $2^{O((k/r)\log r)}\cdot n^{\OO(1)}$-time algorithms by Abasi et al.~\cite{DBLP:conf/mfcs/AbasiBGH14} and Gabizon et al.~\cite{DBLP:conf/esa/GabizonLP15} imply that this containment holds. 
\section{Conclusion}\label{sec:conclusion}

In this paper, we considered four problems, \diPathR, {\sc Undirected $r$-Simple $k$-Path},  \setPackingR, and \monomDetR, parameterized by $k/r.$ We proved that \diPathR, {\sc Undirected $r$-Simple $k$-Path},  and \setPackingR\ are \FPT, but \monomDetR\ is \paraH. In particular, we obtained a $2^{\OO((k/r)^2\log(k/r))}\cdot (n+\log k)^{\OO(1)}$-time algorithm for {\sc Directed $r$-Simple $k$-Path} and a $2^{\OO(k/r)}\cdot (n+\log k)^{\OO(1)}$-time algorithm for {\sc Undirected $r$-Simple $k$-Path}. 
Our work also resolved an open problem posed by Gabizon et al.~concerning the design of polynomial kernels for problems with relaxed disjointness constraints whose size becomes smaller as the relaxation parameter becomes~larger.

Let us conclude our paper with a couple of open problems. First, it would interesting to characterize input polynomials $P$ for which \monomDetR\ becomes \FPT\ or, at least, find  non-trivial sufficient conditions for $P$ such that  the restricted \monomDetR\ is \FPT\ and both \diPathR\ and \setPackingR\ can be easily reduced to it. Secondly, we would like to point out that the existence of a single-exponential \FPT\ algorithm for \diPathR\ remains an open problem. The question of the existence of a deterministic   $2^{\OO((n/d)\log d)}$-time algorithm for  {\sc Degree-Bounded Spanning Tree},  which we did not consider in this study, is also open.

In general, it would be interesting to study the parameterized complexity of other problems with relaxed disjointness constraints parameterized by $k/r$.  Indeed, we believe that much remains to be explored in the realm of problems with relaxed disjointness constraints. Such problems can enable to obtain substantially (sometimes super-exponentially) better solutions at the expense of allowing repetitions, sometimes with the great advantage of a time complexity that diminishes surprisingly fast as $r$ increases.


\appendix
\section{Pseudocode of the Algorithm}\label{sec:pseudocode}

Given that our algorithm for {\sc Undirected $r$-Simple $k$-Path} is optimal under the ETH, we present its pseudocode (in Algorithm \ref{undirectedPseuodo}) in case it is to be implemented. The pseudocode uses the algorithm in Lemma \ref{lem:tw2CompDP} as a black box. The precise details of the implementation of this black box are explicitly given in the beginning of the proof of Lemma \ref{lem:tw2CompDP}.

\noindent\fbox{\begin{algorithm}[H]
\caption{Algorithm for {\sc Undirected $r$-Simple $k$-Path}.}\label{undirectedPseuodo}
\eIf{$r\leq\sqrt{k}$}{
	Use Algorithm \ref{undirectedPseuodo1} to solve the input instance\;}{
	Use Algorithm \ref{undirectedPseuodo2} to solve the input instance\;}
\end{algorithm}}

\noindent\fbox{\begin{algorithm}[H]
\caption{Algorithm for {\sc Undirected $r$-Simple $k$-Path}: Case 1 ($r\leq \sqrt{k}$).}\label{undirectedPseuodo1}
	\For{$\mathsf{col}\in{\cal F}$, $\overline{\bf d}\in {\cal D}_{k,r}$}{
	Color $G$ by $\mathsf{col}$\;
    Allocate $\mathsf{N}$ with an entry $[v,\overline{\bf d}',C]$ for all $v\in V(G)$, $\overline{\bf d}'=(d'_1,\ldots,d'_{\mathsf{b}(k/r)})\in{\cal D}_{r,k}$ such that $d_i'\in\{0,\ldots,d_i\}$ for all $i\in\{1,\ldots,\mathsf{b}(k/r)\}$, and $C\subseteq\{1,\ldots,\mathsf{b}(k/r)\}$\;
	Initialize all $\mathsf{N}[v,\overline{\bf d}',C]$ where $\sum_{i=1}^{\mathsf{b}(k/r)}d'_i\leq 1$ as follows. If $d_{\mathsf{col}(v)}'\neq 1$, then $\mathsf{N}[v,\overline{\bf d}',C]=-\infty$. Otherwise, $\mathsf{N}[v,\overline{\bf d}',C]$ is the maximum of $0$ and the output of algorithm in Lemma \ref{lem:tw2CompDP} with input $(G,C,v,\overline{\bf d})$\;
		\For{$\mathsf{N}[v,\overline{\bf d}',C]$ in non-decreasing order on $\sum_{i=1}^{\mathsf{b}(k/r)}d'_i\geq 2$}{
			\eIf{$d'_{\mathsf{col(v)}}=0$}{		
				$\mathsf{N}[v,\overline{\bf d}',C]=-\infty$\;}
				{Let $\overline{\bf d}''=(d''_1,\ldots,d''_{\mathsf{b}(k/r)})$ where $d''_{\mathsf{col(v)}}=d'_{\mathsf{col(v)}}-1$ and $d''_i=d'_i$ for all $i\in\{1,\ldots,\mathsf{b}(k/r)\}\setminus\{\mathsf{col}(v)\}$\;
				\For{$C'\subseteq C$}{
					Let $A_{C'}$ be the output of algorithm in Lemma \ref{lem:tw2CompDP} with input $(G,C',v,\overline{\bf d})$\;
				}
				$\displaystyle{\mathsf{N}[v,\overline{\bf d}',C]=\max_{u:\{u,v\}\in E(G)}\left(\max\left\{1+\mathsf{N}[u,\overline{\bf d}'',C],\max_{C'\subseteq C}(A_{C'}+1+\mathsf{N}[u,\overline{\bf d}'',C\setminus C'])\right\}\right)}$\;
			}
		}
		\For{$\mathsf{N}[v,\overline{\bf d}',C]$ with $\overline{\bf d}'=\overline{\bf d}$}{
			\If{$\mathsf{N}[v,\overline{\bf d}',C]\geq k-1$}{
				\KwRet "Yes-instance"\;
			}
		}
	}
	\KwRet "No-instance''\;
\end{algorithm}}

\noindent\fbox{\begin{algorithm}[H]
\caption{Algorithm for {\sc Undirected $r$-Simple $k$-Path}: Case 2 ($r>\sqrt{k}$).}\label{undirectedPseuodo2}
	Compute a maximal matching $M$ in $G$\;
	Let $U$ be the set of endpoints of edges in $M$\;
	\eIf{$|U|>3k/r$}{
		\KwRet "Yes-instance"}
		{\For{$U'\subseteq U$, $\mathsf{col}\in{\cal F}$, $\overline{\bf d}\in {\cal D}_{k,r}$ with $d_{\mathsf{col(v)}}\geq 1$ for all $v\in U'$\label{step:loop}}
			{Color $G$ by $\mathsf{col}$\;
			Let $G'=G-X$ for $X=(U\setminus U')\cup\{v\in V(G)\setminus U: $ there exists a vertex in $U'$ with the same color as $v\}$\;	
			Allocate $\mathsf{N}$ with an entry $[v,\overline{\bf d}']$ for all $v\in V(G')$, and $\overline{\bf d}'=(d'_1,\ldots,d'_{\mathsf{b}(k/r)})\in{\cal D}_{r,k}$ such that $d_i'\in\{0,\ldots,d_i\}$ for all $i\in\{1,\ldots,\mathsf{b}(k/r)\}$\;
Initialize all $\mathsf{N}[v,\overline{\bf d}']$ where $\sum_{i=1}^{\mathsf{b}(k/r)}d'_i\leq 1$ as follows. If $d_{\mathsf{col}(v)}'\neq 1$, then $\mathsf{N}[v,\overline{\bf d}']=\mathsf{false}$. Otherwise, $\mathsf{N}[v,\overline{\bf d}']=\mathsf{true}$\;
		\For{$\mathsf{N}[v,\overline{\bf d}']$ in increasing order on $\sum_{i=1}^{\mathsf{b}(k/r)}d'_i\geq 2$}{
			\eIf{$d'_{\mathsf{col(v)}}=0$}{
				$\mathsf{N}[v,\overline{\bf d}']=\mathsf{false}$\;}
				{Let $\overline{\bf d}''=(d''_1,\ldots,d''_{\mathsf{b}(k/r)})$ where $d''_{\mathsf{col(v)}}=d'_{\mathsf{col(v)}}-1$ and $d''_i=d'_i$ for all $i\in\{1,\ldots,\mathsf{b}(k/r)\}\setminus\{\mathsf{col}(v)\}$\;
				$\displaystyle{\mathsf{N}[v,\overline{\bf d}']=\mathsf{OR}_{u:\{u,v\}\in E(G')}(1+\mathsf{N}[u,\overline{\bf d}''])}$\;
		}}
			\If{$\mathsf{OR}_{v\in V(G')}\mathsf{N}[v,\overline{\bf d}]=\mathsf{true}$}{		
		Let $c_v=r-d_{\mathsf{col(v)}}$ for all $v\in V(G')$,  $F=\sum_{v\in V(G')}c_v$ and $\ell=2(k-\sum_{i=1}^{\mathsf{b}(k/r)}d_i)$\;
		Construct a flow network $N$ with source $s$ and sink $t$ as follows.  For all $v\in V(G')$, insert (into $N$) two new vertices, $v_1$ and $v_2$, the arc $(v_1,v_2)$ of infinite (upper) capacity and cost $1$, and the arcs $(s,v_1)$ and $(v_2,t)$ both of (upper) capacity $c_v$ and cost $0$.  For all $\{u,v\}\in E(G')$, insert (into $N$) the arcs $(u_1,v_2)$ and $(v_1,u_2)$ both of infinite (upper) capacity and cost $0$\;
		Compute the minimum cost $C$ required to send $F$ units of (integral) flow from $s$ to $t$ in $N$ in polynomial time\;
		\If{$C\leq F-\ell$}{
			\KwRet "Yes-instance"\;}					
		}}	
		\KwRet "No-instance"\;
	}
\end{algorithm}}

\end{document}